\documentclass[11pt]{article}
\pdfoutput=1
\usepackage{amsfonts, amsthm, amsmath, comment}
\let\chapter\section
\usepackage[boxed]{algorithm2e}
\usepackage{capt-of}

\usepackage{pifont}
\newcommand{\cmark}{\ding{51}}%

\usepackage{url}
\usepackage{float}
\usepackage[margin=1in]{geometry}
\usepackage[usenames]{color}
\usepackage[dvipsnames,usenames,table]{xcolor}
\usepackage{enumitem}
\usepackage{tikz}
\usetikzlibrary{matrix}
\usepackage{verbatim}
\usepackage{pgfplots}
\usepackage{graphicx}
\usepackage{verbatim}
\usepackage{caption}
\usepackage{subcaption}
\usepackage{multirow}
\usepackage{array}

\usepackage{makecell}

\usepackage{mathtools}
\usepackage{bm}
\usepackage{dsfont}
\pgfplotsset{compat=1.8}

\usepackage{tabularx, booktabs}

\usepackage{algorithmic}

\usepackage{thmtools}
\usepackage{thm-restate}
\usepackage[colorlinks=true,urlcolor=blue,linkcolor=RoyalBlue,citecolor=OliveGreen]{hyperref}
\usepackage{prettyref}
\usepackage[numbers]{natbib}

\DeclareMathOperator*{\E}{\mathbb{E}}
\newtheorem*{theorem*}{Theorem}
\newtheorem{theorem}{Theorem}[section]
\newtheorem{lemma}[theorem]{Lemma}

\newtheorem{claim}{Claim}[section]

\newtheorem{remark}{Remark}[section]

\theoremstyle{definition}

\newtheorem{definition}{Definition}

\DeclareMathOperator*{\argmax}{arg\,max}
\DeclareMathOperator*{\argmin}{arg\,min}

\title{
    The Sample Complexity of Auctions with Side Information%
    \thanks{
        This is the fifth arXiv version of the paper.
        Importantly, this version presents a correct proof for the strong revenue monotonicity property, while the previous proof has a fatal bug.
        We have also rewritten most parts of the paper to improve exposition, and to clarify the scopes within which the technical ingredients hold.
    }
}

\author{
Nikhil R.\ Devanur \thanks{Amazon}\and
Zhiyi Huang \thanks{University of Hong Kong}\and
Alexandros Psomas \thanks{Purdue University} }

\date{}
\begin{document}
\maketitle

\newcommand{\alloc}{x}
\newcommand{\allocset}{\mathcal{X}}
\newcommand{\matroid}{\mathcal{M}}

\renewcommand{\E}{\mathbf{E}}
\renewcommand{\Pr}{\mathbf{Pr}}

\newcommand{\vvec}{\bm{v}}
\newcommand{\emp}{\text{\rm emp}}
\newcommand{\sr}{\textsc{SRev}}
\newcommand{\rev}[2]{\textsc{Rev} \left( #1 , #2 \right)}
\newcommand{\approxopt}{\textsc{Apx}}
\newcommand{\opt}[1]{\textsc{Opt} \big( #1  \big)}
\newcommand{\vcg}[1]{\textsc{VCG} ( #1 )}
\newcommand{\topt}[1]{\textsc{Opt} ( #1 )}
\newcommand{\optt}[1]{\textsc{Opt}^\textrm{signals}\left( #1  \right)}
\newcommand{\opteps}[2]{\textsc{Opt}_{#1} \big( #2  \big)}
\newcommand{\Dbold}{\bm{D}}
\newcommand{\Dboldhat}{\bm{\widehat{D}}}
\newcommand{\Mhat}{\widehat{M}}
\newcommand{\Dhat}{\widehat{D}}
\newcommand{\Dboldtilde}{\bm{\tilde{D}}}
\newcommand{\Dboldbar}{\bm{\bar{D}}}
\newcommand{\Dcal}{\mathcal{D}}
\newcommand{\dsigma}[1]{D^{\sigma_{#1}}}
\newcommand{\dsigmahat}[1]{D^{\hat{\sigma}_{#1}}}
\newcommand{\dtrunc}[2]{{#1}^{\sf trunc}_{#2}}
\newcommand{\fdsigma}{F_{\Dcal,\sigma}}

\newcommand{\defeq}{\stackrel{\text{def}}{=}}

\def\signalsmodel{signals model\xspace}
\def\Signalsmodel{Signals Model\xspace}
\def\crmodel{Myerson model\xspace}
\def\CRmodel{Myerson Model\xspace}

\renewcommand{\vec}[1]{\bm{#1}}

\def\mat{\bm{M}}
\def\I{\mathcal{I}}

\newcommand{\zh}[1]{\textcolor{red}{(\textbf{Zhiyi:} #1})}
\newcommand{\anote}[1]{{\color{red} {(\sf Alex's Note:} {\sl{#1}}     {\sf )}}}
\newcommand{\nikhil}[1]{{\color{red} {(\sf Nikhil's Note:} {\sl{#1}}     {\sf )}}}

\begin{abstract}
    Traditionally, the Bayesian optimal auction design problem has been considered either when the bidder values are i.i.d., or when each bidder is individually identifiable via her value distribution. The latter is a reasonable approach when the bidders can be classified into a few categories, but there are many instances where the classification of bidders is a continuum.  For example, the classification of the bidders may be based on their annual income, their propensity to buy an item based on past behavior, or in the case of ad auctions, the click through rate of their ads. We introduce an alternate model that captures this aspect, where bidders are \emph{a priori} identical, but can be distinguished based (only) on some side information the auctioneer obtains at the time of the auction. 

We extend the sample complexity approach of Dhangwatnotai, Roughgarden, and Yan (2014) and Cole and Roughgarden (2014) to this model and obtain almost matching upper and lower bounds.
As an aside, we obtain a revenue monotonicity lemma which may be of independent interest. We also show how to use Empirical Risk Minimization techniques to improve the sample complexity bound of Cole and Roughgarden (2014) for the non-identical but independent value distribution case.

\end{abstract}

\thispagestyle{empty}

\clearpage

\tableofcontents
\thispagestyle{empty}

\clearpage

\setcounter{page}{1}


\section{Introduction}
A significant part of auction design theory is on auctions that maximize the revenue generated.
The predominant model here is that of \emph{Bayesian Mechanism Design}: 
the bidder values are assumed to be drawn from a given distribution, and the objective is the expected revenue, 
where the expectation is over the draw of values
and randomization in the mechanism itself, if any. 
The seminal result in this field is the much celebrated Myerson's \cite{Myerson} theory of optimal auctions. 
When bidders' values are independently and identically distributed (i.i.d.), it gives the beautiful conclusion that a second price auction with a reserve price is revenue optimal among essentially all reasonable auction formats. 
When the bidders are distinguishable, via a distinct value distribution for each bidder (which is still independent of other bidders' valuations), the theory completely characterizes the optimal auction, even though it is not as simple as in the i.i.d.\ case. We will refer to this as the \textbf{\crmodel}.

\paragraph{Sample Complexity}
The question that we consider in this paper is that of sample complexity: \emph{how many samples from the distribution are necessary/sufficient to approximate the optimal auction? }

The sample complexity approach to revenue maximization \cite{Dhangwatnotai2014revenue,ColeR14,HuangMR15,MR15} assumes that we are given access to samples from the value distribution and measures the number of samples that are 
necessary/sufficient to approximate the optimal auction. 
This study is motivated by the observation that the source of priors is really past data, 
both in principle and in practice \cite{Ostrovsky2011reserve}. 
The goal has been to put such practices on a sound theoretical footing, by 
identifying principled approaches and precisely quantifying the value of data. 

\citet*{Dhangwatnotai2014revenue} initiated this line of research, by giving almost tight bounds for the case of a single agent. 
This was later extended to the \crmodel by \citet{ColeR14}.\footnote{They made an assumption that the distributions are \emph{regular}, which implies that a certain revenue function is concave. This is a common assumption in the literature; see Section \ref{sec:prelims} for a definition.} 
In this paper we improve their \cite{ColeR14} upper bound,   from $\tilde{O}(n^{10} \epsilon^{-7})$ to $\tilde{O}(n\epsilon^{-4})$. 
\prettyref{sec:mainresults} presents a detailed discussion of the upper and lower bounds for different variants, and comparisons with other related work such as \citet{MR15} and \citet*{HuangMR15}.

\paragraph{Side Information}
The \crmodel could be applicable in practice when bidders can be classified into a few different categories, and market research 
provides a distribution for each of these. 
We observe that often, the classification of bidders is a continuum, rather than discrete. 
Suppose for instance that the auctioneer knows the incomes of all the bidders. 
It is reasonable to assume that bidders with higher incomes in general have higher valuations and thus use this information in the auction design. 

In this paper, we propose a model where the side information that can be used to distinguish between the bidders comes from a continuum.
\emph{A priori}, bidders are identical: 
all bidders draw two numbers, a value and a ``signal'', i.i.d.\ from a joint probability distribution. The signal is a real number in $[0,1]$ that captures the side information that the auctioneer has about the bidder; 
in other words \emph{the bidder can only lie about his value and not his signal}. 
We will refer to this as the \textbf{\signalsmodel}.

From the Bayesian Mechanism Design perspective, Myerson's theory readily extends to the \signalsmodel:
once you condition on the signal for each buyer, 
you have a different (and independent) distribution for each buyer, 
and the setting reduces to the \CRmodel. 
The same is not true from the sample complexity perspective. 
Even when the bidder categories are discrete,\footnote{Our model can handle discrete categories as well, where the signal is just an encoding of the categories.} there is a distribution over categories as well.
If bidders from some categories appear more often than others, then it is more important that the auction does the right thing for these categories.
This is represented in the samples as well: 
the number of samples for each category is different if you collect samples at random from the population. 
The sample complexity results in the \crmodel assume that the same number of samples are available for each category, 
or equivalently, the quantity in question is the minimum number of samples across all categories; this can lead to suboptimal guarantees.

\paragraph{Techniques}

\citet{ColeR14} approached the problem by approximating what is called the `Myerson virtual value function'\footnote{This is a function that depends on the probability distribution of values; see Section \ref{sec:prelims} for a definition.} 
based on samples.
We take quite a different approach: our improvement in the sample complexity upper bounds for the \CRmodel comes from an Empirical Risk Minimization (ERM)  approach. 
The main step is to show how to discretize the value space with a (relatively) small support size, 
while losing only an $\epsilon$ fraction of the revenue for some given $\epsilon \in (0,1)$. 
This then allows us to bound the number of different mechanisms that we need to consider. 
Then we use a concentration inequality to argue that for each mechanism in this class, 
its revenue on the sample closely approximates its actual revenue from the distribution. 
Thus, using the best auction on the sample suffices. 

As an important part of this approach we use a new concentration inequality.\footnote{We discovered this inequality without knowing that 
\citet*{BabichenkoBP/2016/MOR} had earlier discovered a very similar version.}
Consider the revenue of a mechanism with $n$ agents in the \CRmodel. 
Consider a matrix of values, each drawn independently, with the values in column $i$ drawn from the distribution corresponding to agent $i$.
Standard concentration inequalities such as Bernstein inequality tell us that the revenue of a mechanism averaged over the row vectors of this matrix is close to its expected revenue with respect to (w.r.t.) the distribution. 
The concentration inequality we use gives the same conclusion for the revenue of the mechanism 
w.r.t.\ the product of the empirical distributions over each column of the matrix. 
Due to this, instead of finding the optimal auction for a correlated distribution, 
it is sufficient to find the optimal auction for a product distribution, a computationally easy task. 

The \signalsmodel presents the following new conceptually important challenge:
we cannot rely on having seen samples with exactly the same signals 
as the participants in the auction. This requires us to be able to interpolate using samples from different distributions, rather than use samples from the same distribution as in the \CRmodel.\footnote{This introduces additional difficulties such as that convex combinations of regular distributions are not regular, e.g., see \citet{Sivan2013vickrey}. } 
We do this by using just the right number of signals in the sample immediately below the actual signal of an agent, and run a version of the {\em Empirical Myerson} auction: the Myerson auction run on the empirical value distributions given by the samples. The approximation factor is proved via a careful charging argument.
Our bounds depend on the following function of the distribution: function $q$ specifies that an $\epsilon$ fraction of the revenue is contributed by the top $q(\epsilon)$ fraction (in probability mass) of signals. 
For instance, when there is a single agent, we show that $O( \log(\tfrac 1 \epsilon) q(\epsilon)^{-1} \epsilon^{-3})$ samples are sufficient 
and $\Omega( \log(\tfrac 1 {q(\epsilon)} ) \epsilon^{-4}) $ are necessary.\footnote{This is not a contradiction since it turns out that $q(\epsilon)$ is always at most $\epsilon$ and $x\log (1/x)$ is a monotone increasing function for small enough $x$. }

An essential component of interpolating from different distributions is {\em revenue monotonicity}. Consider the optimal auction $M$ for a given product distribution $\Dbold$. If $M$ is run on values that come from a 
higher (component wise first order stochastic dominating) distribution, 
does it get a higher revenue? Surprisingly, as far as we know, this fundamental question about optimal auctions was unanswered prior to this paper; we resolve this in the affirmative.
Another application of this property is in devising mechanisms when only 
certain statistics (such as the medians) of the value distributions are known (\citet{Azar2013optimal}).\footnote{A natural mechanism to consider in such a setting is the optimal mechanism w.r.t.\ the minimal distribution with the given statistics. Given the value of the median $\mu$, the distribution which is 0 or $\mu$ with probability $\tfrac 1 2$ each is first order stochastically dominated by all other distributions with median $\mu$. }   
The proof of this property is a nontrivial amortization argument.%
\footnote{Direct approaches seem to fail here. 
Consider the following example, where increasing one component distribution first increases the revenue, and then increasing another component {\em decreases} the revenue (all w.r.t.\ an optimal mechanism for the original distribution). There are two bidders: bidder $1$ has value distribution $U[0,100]$ and bidder $2$ has distribution $U[0,1]$. The virtual value functions are $\phi_1 = 2v - 100$ and $\phi_2 = 2v - 1$ respectively. If bidder $1$'s distribution is replaced by $U[50,100]$, the revenue strictly increases, since the probability of allocating the item to her (for a high price) increases. In fact bidder $1$ always gets the item, except when she has a value in $[50,50.5]$ and bidder $2$ has a value in $[0.5,1]$ with higher virtual value. When this occurs the item is sold to bidder $2$ for a low price. The probability of this event strictly increases when bidder $2$'s distribution is replaced with a point mass at $1$, and thus the revenue slightly decreases. 
The net effect is still an increase in revenue, as guaranteed by our theorem.
This also shows that revenue monotonicity is not true w.r.t \emph{any} mechanism.}
We design two amortized gain functions whose expectations upper and lower bound an auction's expected revenue respectively.
We then couple the (randomized) revenue of Myerson's optimal auction $M$ on distribution $\Dbold$ and its revenue on a higher distribution based on the quantiles of agents' values.
We show that for any given quantiles, the first amortized gain of the former is less than or equal to the second amortized gain of the latter.
See Section~\ref{sec:rev_monotone_proof} for details.

For lower bounds, we extend the instances used by \citet{HuangMR15}, by packing scaled copies of them among the conditional distributions for different signals.
We show that in order to achieve high revenue overall it is necessary to achieve high revenue in most of the conditionals. Moreover, in order to get good revenue from any conditional (even with exact knowledge of the rest) a reduction to classification shows that $\Omega(\epsilon^{-3})$ samples are necessary; the bound follows.

\paragraph{Extensions}

Our results can be extended to more general single parameter families of auction environments.
We can also get better upper bounds for the case of $n$ agents when we are required to  optimize over 
a smaller class of auctions, such as VCG with reserve prices. 
We discuss these extensions in Section~\ref{sec:extensions}.

\subsection{Other Related Work}

\paragraph{Sample Complexity}
\citet{Elkind2007designing} also considered a learning question very similar to the sample complexity line of work \cite{Dhangwatnotai2014revenue,ColeR14,HuangMR15,MR15}, in the presence of an oracle that returns the expected profit of a given mechanism, for distributions with finite support. 
This line of work is also close in spirit to \citet{Balcan2008reducing}, 
who used PAC learning techniques to design prior-free auctions in a general setting. 
In fact, \citet{HuangMR15} noted that the results of \citet{Balcan2008reducing} could be used to deduce sample complexity bounds in the i.i.d.\ model. 
In addition to the lower bound of $\Omega(\epsilon^{-3})$ for regular distributions, 
\citet{HuangMR15} also showed tight bounds of $\tilde\Theta (\epsilon^{- 3/ 2})$ for 
MHR distributions. 
While all these results are for single parameter environments, 
\citet*{Dughmi2014sampling} showed an exponential lower bound for a multi parameter setting. 

\paragraph{Ad Auctions}
Billions of auctions are run each day for ads,  for search and display advertising. 
Setting reserve prices in these auctions has received a lot of attention \cite{Even2008position,Thompson2013revenue,Roberts2013ranking,Ostrovsky2011reserve}, and has been one of the prime applications of Myerson's theory \cite{Ostrovsky2011reserve}.
Common techniques in practice include using machine learning algorithms, to
map {\em features} of an auction into one of few categories, and 
use the corresponding reserve price.  
Most commonly, the mapping results in a continuous variable, and an extra effort is required to map this into discrete categories.
It would be more natural to use the continuous variable itself, which would correspond to the signal in our model.

\subsection{Future Directions}

Much of the auction theory since Myerson, in particular in the algorithmic game theory literature, is devoted to  cases where the distributions are unknown.
One can ask several `simple versus optimal' questions in the \signalsmodel, such as an extension of the famous result of \citet{BulowKlemperer}: how much does the market size have to increase for the Vickrey auction to outperform the optimal auction.
Another interesting direction is to ask for the design of a prior-independent mechanism in this model, which might be easier than a similar attempt in the prior-free setting by \citet{LeonardiRoughgarden2012} and \citet{Bhattacharya2013NearOptimal}. 
Finally, extending this model to align it even closer with how machine learning algorithms are used, by directly considering high dimensional signal space is a really exciting direction.

\section{Preliminaries}
\label{sec:prelims}

\paragraph{Notations}
For any positive integer $\ell$, let $[\ell] \defeq \{ 1, 2, \dots, \ell \}$ denote the set of positive integers that are at most $\ell$.
For any binary vector $\vec{\alloc} \in \{0, 1\}^n$, we abuse notation and let $\vec{\alloc}$ denote the set of bidders $\{ i \in [n] : \alloc_i = 1 \}$ as well.
Hence, if we write $\vec{\alloc}' = \vec{\alloc} + \{ i \} - \{ j \}$, it means that $\vec{\alloc}'$ is identical to $\vec{\alloc}$, except that its $i$-th coordinate is set to $1$ while its $j$-th coordinate is set to $0$.
Similarly, $\vec{\alloc} \setminus \vec{\alloc}'$ is a vector whose $i$-th coordinate is $1$ if and only if $\alloc_i = 1$ and $\alloc'_i = 0$.

\subsection{Single Parameter Revenue Maximization}

Let there be an auctioneer with a single type of item for sale.
Let there be $n$ bidders who are interested in getting a copy of the item (we simply say getting an item for brevity hereafter).
The value of bidder $i$ for the item is $v_i \ge 0$, which is private information known only to bidder $i$ but not to the other bidders and the auctioneer. 
Let there be a feasible set of allocations $\allocset \subseteq [0, 1]^n$;
for any allocation $\vec{\alloc} \in \allocset$, each coordinate $\alloc_i$ denotes the probability that bidder $i$ gets an item.

This paper considers three families of feasible sets $\allocset$, with each one more general than the previous one.
\begin{itemize}
    \item \textbf{Single-item setting:~} 
    $\allocset = \{ \vec{\alloc} \in [0, 1]^n : \sum_{i=1}^n \alloc_i \le 1 \}$.
    It corresponds to the case when the auctioneer has a single item for sale.
    This is arguably the simplest single parameter problem.
    \item \textbf{Matroid setting:~}
    $\allocset$ is the convex hull of a matroid $\matroid \subseteq \{0, 1\}^n$.
    Here, $\matroid$ is a matroid if it satisfies the following properties.
    (1) It includes the empty allocation, i.e., $\vec{0} = (0, 0, \dots, 0) \in \matroid$ \textbf{(nonemptiness)}.
    In fact, we will further assume that all singleton sets are feasible; it is without loss of generality (wlog) up to the removal of irrelevant bidders.
    (2) It is downward-closed, i.e., for any $\vec{\alloc}' \le \vec{\alloc}$ (coordinate-wise), such that $\vec{\alloc} \in \matroid$, we have $\vec{\alloc}' \in \matroid$ \textbf{(hereditary property)}.
    (3) For any distinct $\vec{\alloc}, \vec{\alloc}' \in \matroid$, such that $\vec{\alloc}$ allocates to more bidders than $\vec{\alloc}'$, there exists $i \in \vec{\alloc} \setminus \vec{\alloc}'$ such that $\vec{\alloc}' + \{ i \} \in \matroid$ \textbf{(augmentation property)}.
    For any problem in the matroid setting, all maximal feasible allocations have the same number of bidders, which is called the rank of the matroid.
    Let $k$ denote the rank of the matroid.
    The single-item constraint, and the more general case when the seller has $k$ copies of the item and hence can allocate to up to $k$ bidders, are both special cases of the matroid setting.
    \item \textbf{Downward-closed setting:~}
    Finally, we drop the third property in the matroid constraint and consider arbitrary feasible sets that are downward-closed.
    Let $k = \max_{\vec{\alloc} \in \allocset} \|\vec{\alloc}\|_1$ denote the maximum number of allocated bidders  in any feasible allocation.
\end{itemize}

All results in this paper apply to the single-item setting.
Most of them further generalize to the matroid setting, or even the downward-closed setting.
Our theorem statements will explain the scope in which the results hold.

From the revelation principle, it is sufficient to consider direct revelation mechanisms: 
a mechanism $M$ is a pair of functions $(\vec{\alloc}, \vec{p})$, both taking as input a reported value profile, often referred to as bids, $\vec{b} = \left( b_1, b_2, \dots, b_n \right)$. 
The allocation function $\vec{\alloc}$ has for each bidder $i$, a component $x_i(\vec{b})$ that represents her probability of winning the item. 
It further ensures feasibility, i.e., for any bids $\vec{b}$, we have $\vec{x}(\vec{b}) \in \allocset$.
Similarly, the payment function $\vec{p}$ has a component $p_i(\vec{b})$ for the payment of each bidder $i$. 
We will omit the bids $\vec{b}$ for brevity when it is clear from the context.
The utility of bidder $i$ is $v_i x_i - p_i$.

A mechanism is \textit{dominant strategy incentive compatible} (DSIC) if the utility of a bidder $i$ is maximized, no matter what the other bidders report, by reporting her true value. A mechanism is \textit{individually rational} (IR) if the utility of each bidder is non-negative, for all valuation profiles.
In general, the bids $\vec{b}$ need not equal the true values $\vvec$.
Nonetheless, throughout the paper we only consider DSIC and IR mechanisms, and therefore omit qualifying mechanisms as DSIC and IR everywhere, and abuse notation and use $\vvec$ to denote the bids as well.

The objective we consider in this paper is maximizing the revenue of a mechanism, when the values are drawn from a given probability distribution. 
We let $\textsc{Rev} \left( M , \Dbold \right)$ denote the expected revenue of mechanism $M$ when $\vvec$ is drawn from distribution $\Dbold$: 
\[
    \textsc{Rev} \left( M , \Dbold \right) \defeq \E_{\vvec \sim \Dbold} \left[ \sum_{i = 1}^n p_i\left( \vvec \right) \right] 
    ~.
\]

The mechanism that maximizes $\rev {\cdot} {\Dbold}$ for a given distribution $\Dbold$, among all DSIC and IR mechanisms, is called the {\em optimal} auction, and its revenue is denoted by $\opt \Dbold$.

\paragraph{Myerson's Optimal Auction:}
\citet{Myerson} characterizes the optimal auction for the case of product distributions $\vec{D} = D_1 \times D_2 \times \dots \times D_n$, which is also the focus of this paper.
The optimal auction relies on a concept called {\em virtual value}.
For simplicity, we explain it under the assumption that the distribution is continuous, as in the original paper of \citet{Myerson}.
Let $F_i$ and $f_i$ denote the cumulative distribution function (cdf) and probability density function (pdf) respectively of distribution $D_i$.
The virtual value of a bidder $i$ when his value is $v$ is:
\[
\phi_i (v) \defeq v - \frac{1 - F_i(v)}{f_i(v)} ~.
\]

The optimal auction is simple to describe under the assumption that each $D_i$ is {\em regular}, which means that the virtual value $\phi_i (v)$ is monotonically non-decreasing.
\citet{Myerson} shows that the expected revenue of any truthful auction equals the expected virtual welfare, i.e., 
\[
\E_{\vec{v} \sim \Dbold} \left[ \sum_{i=1}^n x_i(\vec{v}) \phi_i(v_i) \right] ~.
\]

Hence, Myerson's optimal auction allocates to a subset of bidders such that the above sum of virtual values is maximized;
this is equivalent to allocating to the bidder with the highest non-negative virtual value in the single-item setting (breaking ties arbitrarily). 
Each winner pays the minimum bid that will make her win.
For example, in the single-item setting, the payment by the winner is:
\begin{eqnarray*}
p_i &\defeq& \min \big\{ p:~ \forall ~j\neq i,~ {\phi}_i(p) \geq \phi_j (v_j),  ~\text{and}~ \phi_i(p) \geq 0 \big\} \\
& = & \max _{j\neq i} \big\{ {\phi}_i^{-1} \left( {\phi}_j (v_j) \right) \} 
\cup \big\{ {\phi}_i^{-1} ( 0 ) \big\} ~.
\end{eqnarray*}

When the distributions are not regular, on the other hand, the optimal auctions must ``iron'' the virtual value function $\phi_i$ of each bidder $i$ to get an ironed virtual value, denoted as $\bar{\phi}_i$.
Then, it allocates to a subset of bidders so that the sum of ironed virtual values is maximized (breaking ties arbitrarily); 
each winner pays the minimum bid that will make her win.

For the purpose of understanding most results and proofs in this paper, it is unimportant how the ironed virtual value is defined and how to generalize the definition to distribution with point masses;
we will only use the fact that the optimal auction is an ironed virtual value maximizer characterized by the mapping from values to ironed virtual values for each of the $n$ bidders.
The only place that makes use of the definition of ironed virtual value and the generalization to distribution with point masses is Section~\ref{sec:rev_monotone_proof} and, hence, the formal definition is deferred to that section.

\subsection{\CRmodel}

\citet{ColeR14} introduced a sample complexity approach to the design of optimal auctions.%
\footnote{See, also, \citet{Dhangwatnotai2014revenue} and \citet{Elkind2007designing} for some earlier works with a similar flavor.}
Assume that the distribution $\Dbold$ is unknown to the auctioneer, who instead has to design a mechanism $M$ with access to data only in the form of $m$ i.i.d.\ samples $\vec{s}_1, \vec{s}_2, \dots , \vec{s}_m \in [0, +\infty)^n$ drawn from $\Dbold$. 
The mechanism $M$ is then run on a fresh vector of values (the real input) drawn from $\Dbold$. 

Given such an algorithm that maps collections of samples to mechanisms, the {\em sample complexity} of it against a given class of distributions is specified as a function of $\epsilon \in [0,1]$, the number of bidders $n$, and potentially other attributes of the class of distributions and the set of feasible allocations.
It is defined to be the smallest number of samples $m$ such that for all distributions $\Dbold$ in the given class, it holds that:
\[ 
    \E_{\vec{s}_1, \vec{s}_2, \dots , \vec{s}_m \sim \Dbold} \big[ \rev{M}{\Dbold} \big] 
\geq (1-\epsilon) \opt \Dbold ~.
\] 

Here it is implicit that the mechanism $M$ depends on the samples $\vec{s}_1, \vec{s}_2, \dots, \vec{s}_m$.
Further, even though the above definition states the approximation guarantee only in expectation, essentially all positive results in the literature, including those in this paper, achieve the approximation ratio with high probability.

We will refer to it as the \crmodel in this paper.

\subsection{\Signalsmodel}

In this paper, we introduce a model where the auctioneer has some information correlated with the bidders' values. 
We assume that each bidder $i$ has a value-signal pair $(v_i,\sigma_i)$, drawn from a joint distribution $\Dcal$. 
As before, $v_i$ is her valuation for the item and $\sigma_i \in \left[ 0 , 1 \right]$ is a signal observed by the auctioneer, e.g., her annual income, or her propensity to buy an item based on past behavior, or in case of ad auctions the click-through rate of her ad. 
For a given signal $\sigma$, we denote the conditional distribution over values by $D^{\sigma}$. 
It is convenient to think of $\Dcal$ as a family of $1$-dimensional distributions $\{\dsigma{}: \sigma \in [0,1]\}$, along with the marginal distribution over $\sigma$, given by the cdf $\fdsigma$.  
We assume throughout this paper that $\Dcal$ is such that if $\sigma_i > \sigma_j$, then $\dsigma{i}$ first order stochastically dominates $\dsigma j$ (denoted by $\dsigma{i} \succeq \dsigma{j}$).
The mechanism now has as input the {\em actual} signals, in addition to the reported values.

We extend the sample complexity approach to this model: 
the auctioneer first observes $m$ i.i.d.\ value-signal pairs sampled from the given distribution, and the signals of $n$ bidders whose value-signal pairs are fresh samples (the real input);  
then, based on the $m$ samples, the auctioneer designs a mechanism that is tailored to the realized signals of the $n$ bidders.
The sample complexity will as before depends on $\epsilon \in (0,1)$ (among other parameters), and is defined as the minimum number of samples $m$ required to achieve a $1-\epsilon$ fraction of the optimal revenue $\optt \Dcal $, defined as:  
\[ 
    \optt \Dcal \defeq \E_{\sigma_1,\sigma_2,\ldots,\sigma_n \sim \fdsigma} \big[ \opt{\dsigma 1 \times \dsigma{2} \times \cdots \times \dsigma{n}} \big]
    ~.
\] 

The role of signals in this model is similar to that of contexts in the contextual bandit problem.
Therefore, we will refer to this model as the \signalsmodel.

The \signalsmodel~is incomparable to the \crmodel.
Note that in this model, the bidders are \emph{a priori} identical, but become non-identical once the signals are observed. This allows a mechanism to be a priori symmetric, and distinguish between the bidders based only on the signals. 
Also in this model, the samples may contain entirely different signals than the actual ones.
In other words, it might be that the actual signals are never seen in the samples;
the algorithm must nonetheless be able to learn approximately optimal auctions w.r.t.\ the actual signals from the samples. 
Further, there is a distribution over signals themselves which affects the definition of $\optt \Dcal$. 
Both models however generalize the special case of the \crmodel when bidders have i.i.d.\ value distributions.

We call a distribution $\Dcal$ {\em regular} if $\dsigma{} $ is regular for any $\sigma \in [0, 1]$.

\section{Main Results}
\label{sec:mainresults}

Table~\ref{tab:summary} summarizes the main results in this paper. 
The results come in two flavors: improved sample complexity upper bounds for the \crmodel without signals, as well as lower and upper bounds for the \signalsmodel introduced in this paper.
Finally, we highlight some of the technical lemmas that would be of independent interest. 

\begin{table}
	\centering
	\def\arraystretch{2}
	\renewcommand*{\thefootnote}{\alph{footnote}}
	\begin{tabular}{|l|l|l|l|l|}
		\hline
		{\bf Model} & \makecell[cl]{\bf Number\\ \bf of Agents} & {\bf Distributions} & {\bf Lower Bound} & {\bf Upper Bound} \\
		\hline
		\multirow{5}{*}{\CRmodel}  & \makecell[cl]{$1$} & \makecell[cl]{Regular} & \makecell[cl]{$\Omega\left( \epsilon^{-3}\right)$ \cite{HuangMR15}} & \makecell[cl]{$O \left( {\epsilon^{-3}} \log \frac{1}{\epsilon}\right)$ \cite{Dhangwatnotai2014revenue}} \\
		\cline{2-5}
		& $n$ & Regular & \makecell[cl]{$\Omega(n\epsilon^{-1})$ \cite{ColeR14} \\ $\Omega(\epsilon^{-3})$ \cite{HuangMR15}} & \makecell[cl]{$\tilde{O} \left( {n}{\epsilon^{-4}} \right)$\\ (\prettyref{sec:no_signal})} \\	
		\cline{2-5}
		& $n$ & MHR & \makecell[cl]{$\tilde{\Omega}(n \epsilon^{-1/2})$ \cite{ColeR14} \\ $\Omega(\epsilon^{-3/2})$ \cite{HuangMR15}} & \makecell[cl]{$ \tilde{O}  \left( {n}{\epsilon^{-3}} \right)$\\ (\prettyref{sec:no_signal}, see also~\cite{MR15}\footnotemark[1])} \\
		\cline{2-5}
		& $n$ & Support $\subseteq [1,h]$ &  $\Omega\left( h\epsilon^{-2}\right)$ \cite{HuangMR15} & \makecell[cl]{$\tilde{O} \left( {nh}{\epsilon^{-3}} \right)$\\ (\prettyref{sec:no_signal}, see also \cite{MR15}\footnotemark[1]\footnotemark[2])} \\
		\cline{2-5}
		& $n$ & \makecell[cl]{Support $\subseteq [0,h]$\\ additive error} &  \makecell[cl]{$\Omega\left( h^2\epsilon^{-2}\right)$\\ (folklore)} & \makecell[cl]{$\tilde{O} \left( {nh^3}{\epsilon^{-3}} \right)$\\ (\prettyref{sec:no_signal}, see also~\cite{GonczarowskiN/2017/STOC}\footnotemark[2])} \\	
		\hline
		\multirow{2}{*}{\Signalsmodel} & 1 & \makecell[cl]{Regular\\ $q$-bounded tail} & \makecell[cl]{$ \tilde{\Omega} \big( \epsilon^{-4} \big) $\\ (Section \ref{app:signal_lb})} & \makecell[cl]{$\tilde{O} \left( q(\epsilon)^{-1} \epsilon^{-3} \right)$\\ (\prettyref{sec:signal_ub_single_agent})} \\
		\cline{2-5}
		& $n$ & \makecell[cl]{Regular\\ $q$-bounded tail} &     &  \makecell[cl]{$ \tilde{O}  \left( {n^2} {q(\epsilon)^{-1} \epsilon^{-4} } \right) $ \\ (\prettyref{sec:signal_ub_multi})} \\
		\hline 
		\multicolumn{5}{l}
		{\small \makecell[cl]{
			\footnotemark[1] The bounds from previous work are information theoretic while ours are computationally efficient.\\
			\footnotemark[2] The bounds were not explicit in the papers but to our understanding can be derived from their methods.
		}}
	\end{tabular}
	\caption{Summary of sample complexity upper and lower bounds for various auction settings.}
	\label{tab:summary}
\end{table}

\subsection{Sample Complexity: Regular Distributions}
We show a sample complexity upper bound of $\tilde{O}(n \epsilon^{-4})$ for learning revenue optimal auctions with non-i.i.d.\ regular value distributions.
It significantly reduces the gap between the upper and lower bounds from \citet{ColeR14}, which were respectively $\tilde{O}(n^{10} \epsilon^{-7})$ and 
$\Omega(n\epsilon^{-1})$.
For the case of $n=1$, \citet{HuangMR15} showed a lower bound of $\Omega(\epsilon^{-3})$.

\subsection{Sample Complexity: Beyond Regular Distributions and a Single Item}

The conference version of this paper focused mainly on regular distributions, 
but our techniques are general enough to handle other classes of distributions considered in the literature, including monotone hazard rate (MHR) distributions, distributions with bounded support in $[1, h]$, and distributions with bounded support in $[0, h]$ but aiming for an additive $\epsilon$ approximation instead of a multiplicative one.

The first two classes were studied by \citet{MR15}, who showed information theoretic upper bounds of $\tilde{O}(n h^2 \epsilon^{-3})$ for the $[1, h]$-bounded distributions, and $\tilde{O}(n/\epsilon^3)$ for MHR distributions.
They used the technique of bounding the pseudo-dimension, from statistical learning theory, as compared to our more elementary technique of building an $\epsilon$-net explicitly; 
these two approaches are similar in spirit.  
To our understanding, the upper bound from their approach for the $[1, h]$-bounded case can be improved to $\tilde{O}(n h \epsilon^{-3})$ by using Bernstein instead of Chernoff as the concentration inequality.

A simple modification of our techniques obtains matching upper bounds of $\tilde{O}(n h \epsilon^{-3})$ for $[1, h]$-bounded distributions, and $\tilde{O}(n/\epsilon^3)$ for MHR distributions. 
Importantly, with a concentration inequality for product distributions that we shall explain shortly, our results are constructive and computationally efficient while the ones in \citet{MR15} are not.  

Subsequent to the conference version of this paper, \citet{GonczarowskiN/2017/STOC} proved a sample complexity upper bound of $\text{poly}(h, n, \frac{1}{\epsilon})$ for the $[0, h]$-bounded case.
To the best of our understanding, the precise bound from their approach is $\tilde{O}(n h^3 \epsilon^{-3})$, which is also what the techniques in this paper give.
The main difference lies in the proof of why a discretized Myerson auction can give almost optimal revenue up to an $\epsilon$ additive factor.
Our proof builds on an interpretation of the Myerson auction in the quantile space, while \citet{GonczarowskiN/2017/STOC} gave an explicit rounding of the Myerson auction  based on the virtual values.
\citet{GonczarowskiN/2017/STOC} extended their results beyond single-item auction with a worse yet still polynomial sample complexity upper bound.
We explain in Section~\ref{sec:extensions} why our techniques apply to $[0, h]$-bounded case;
interested readers may apply our techniques to get essentially the same bound (with minor differences in logarithmic factors) as \citet{GonczarowskiN/2017/STOC}.

\citet{RoughgardenS/2016/EC} considered the special case of i.i.d.\ bidders with $[0, h]$-bounded support, and showed a sample complexity of $\tilde{O}(n^2 h^2 \epsilon^{-2})$.
Their bound is incomparable with ours in that ours has a worse dependence in $\epsilon$ while theirs has a worse dependence in $n$.
We stress that our upper bound holds for the more general non-i.i.d.\ case.

We rewrite our technical sections (Section \ref{sec:no_signal} and Section \ref{sec:extensions}) to clarify how our techniques are general enough to be applied to these other classes of distributions.

\paragraph{Beyond the Techniques of This Paper}
Subsequent results by \citet*{GuoHZ:STOC:2019}, \citet{GuoHTZ:COLT:2021},  and \citet{ChenHMY:arXiv:2022} built on the notion of revenue monotonicity introduced in this paper to characterize the optimal sample complexity up to logarithmic factors for all single parameter auction settings.
For the more general multi parameter auctions, \citet{GonczarowskiW:JACM:2021} proved the first polynomial sample complexity upper bounds when bidders have additive or unit demand valuations, which were later further improved by \citet{GuoHTZ:COLT:2021}.

\subsection{Signals Model}
Our sample complexity bounds for the \signalsmodel depend on the following property of the joint distribution of value-signal pairs. 
The property is defined via a function $q$ such that 
ignoring the top $q(\epsilon)$ quantiles in the signal space reduces the revenue by less than an $\epsilon$ fraction. 
\begin{definition}[Small-tail in signal space]
	\label{def:onebuyeralg3}
	Given a function $q:[0,1] \rightarrow [0,1]$, we say that a distribution $\Dcal$ has a $q$-bounded tail in the signal space if:
	\[	   
	\E_{\sigma_1, \ldots, \sigma_n \sim\fdsigma} \big[\mathds{1}\{\forall~ i, \fdsigma(\sigma_i) \le 1-q(\epsilon)\} \cdot \opt{ \dsigma1 \times \cdots \times \dsigma n} \big] \ge
	(1- \epsilon) \optt{\Dcal} ~.
	\]
\end{definition}

Intuitively, the sample complexity needs to have at least a linear dependence on $1/q(\epsilon)$, since we need at least that many samples to even see a single $\sigma$ in the top $q(\epsilon)$ quantile. 
If you do not see any signal from this range, then there is no hope of capturing their contribution to the revenue,
which could be an $\epsilon$ fraction. 
Another important observation is that because of the monotonicity of $\Dcal$ in the signal space, the top $\epsilon$ quantile contributes at least an $\epsilon$ fraction of the revenue. 
This implies that $q(\epsilon)$ is always no larger than $\epsilon$. 

For the case of a single bidder, we show a sample complexity upper bound of $\tilde{O}(q(\epsilon)^{-1} \epsilon^{-3})$ in the \signalsmodel, which is larger than the corresponding optimal bound in the \crmodel by essentially a factor of $q(\epsilon)^{-1}$.
We also complement our positive result with a lower bound of $\Omega(\epsilon^{-4} \log(q(\epsilon)^{-1}))$, which is larger than the corresponding optimal bound in the \crmodel by a factor of $\epsilon^{-1} \log\left(1/ {q(\epsilon)} \right)$.
It is to be expected that the \signalsmodel requires more samples, since we need to cover the entire spectrum of different $\dsigma{}$'s.

With $n$ bidders, we show an upper bound of $\tilde{O}  \left( {n^2} {q(\epsilon)^{-1} \epsilon^{-4} } \right)$.
Note that one sample in the \crmodel is actually $n$ values, whereas 
one sample in the \signalsmodel is only one value. Taking this into account, 
the upper bounds still differ by a factor of $n q(\epsilon)^{-1}$. 
We currently do not have any lower bounds for the \signalsmodel in the multi-bidder case, since in this model, having more agents might actually help! 
Closing the gap between the bounds, or parameterizing the bounds using other reasonable quantities are interesting open questions.

\subsection{Strong Revenue Monotonicity}
We say that a product distribution $\Dbold = D_1 \times D_2 \times \dots \times D_n$ component-wise stochastically dominates another product distribution $\Dboldhat = \Dhat_1 \times \Dhat_2 \times \dots \times \Dhat_n$, and denote it by $\Dbold \succeq \Dboldhat$, if for each $i$, we have that $D_i \succeq \widehat{D}_i$.
We prove the following theorem which is of independent interest:
\begin{theorem}[Strong Revenue Monotonicity]
\label{thm:revenueMonotonicity} 
Consider any single parameter mechanism design problem in the matroid setting.
Let $\Dboldhat$ be a product distribution, and $\Mhat$ be an optimal auction for $\Dboldhat$.
Then, for any product distribution $\Dbold \succeq \Dboldhat$, we have:
\[
\rev{\Mhat}{\Dbold} \ge \rev{\Mhat}{\Dboldhat} = \opt{\Dboldhat} ~.
\]
\end{theorem}

\begin{remark}\label{rem:uniqueMechanism}
All optimal auctions are virtual welfare maximizers, identical up to tie breaking rules. We can make the optimal auction for a product distribution $\Dboldhat$ unique as follows: add an infinitesimal $\epsilon_i$ to all values in the (finite) support of $\widehat{D}_i$, such that $\epsilon_i > \epsilon_{i+1}$. In the new product distribution the revenue is the same, but the only optimal tie breaking rule is lexicographical.
\end{remark}

\citet{HartR/2015/TE} considered a weaker notion of revenue monotonicity, namely, the optimal revenue w.r.t.\ the dominating distribution $\vec{D}$ is weakly higher than that of the dominated distribution $\Dboldhat$.
The weaker version holds for arbitrary single parameter mechanism design problems.
The main insight from \citet{HartR/2015/TE} on this topic is that even the weaker version of revenue monotonicity ceases to hold in the multi parameter setting, even with only one bidder.
Note that the stronger notion in Theorem~\ref{thm:revenueMonotonicity} directly implies the weaker one.

Subsequent work by \citet{ChenHMY:arXiv:2022} showed that the stronger notion no longer holds beyond the matroid setting.
Nevertheless, they recovered an approximate version of strong revenue monotonicity which still suffices for proving sample complexity upper bounds.
Whether an approximate version of weak/strong revenue monotonicity notion holds for multi parameter mechanism design problems is an interesting open question.

\subsection{Concentration Inequality for Product Distributions}
\label{sec:concentration_intro}
We prove the following concentration inequality, which is useful for showing that the revenue of a mechanism on a given product distribution is close to its revenue on the {\em empirical product distribution}, which is defined to be a product distribution whose components are uniform distributions over the corresponding sample values.
Subsequent to the conference version, we found out that a very similar theorem was known due to \citet{BabichenkoBP/2016/MOR}.
The main difference is that we use Bernstein instead of Chernoff in the proof to get a better bound when the expectation of the random variable is much smaller than the maximum value in its support.
We include the proof in Section~\ref{sec:prodconcentrate} for completeness. 
\begin{theorem}
	\label{thm:prodconcentrate}
    Let $V$ be an arbitrary measurable set, $f$ be an arbitrary function from $V^n$ to $[0, 1]$.
    Let $D_1, D_2, \dots, D_n$ be $n$ distributions over $V$, and $\vec{D} = D_1 \times D_2 \times \dots \times D_n$.
	Suppose $p = \E_{\vec{x} \sim \Dbold} [ f(\vec{x}) ]$.
    Let $x_{i1}, \dots, x_{im}$ be $m$ i.i.d.\ samples from $D_i$.
    Let $E_i$ be the uniform distribution over $\{ x_{i1}, \dots, x_{im} \}$ for each $i \in [n]$, and $\vec{E} = E_1 \times E_2 \times \dots \times E_n$.
	Then, we have:
	\[
	\textstyle
	\Pr_{x_{ij} \sim D_i, \forall i\in [n], \forall j \in [m]} \Big[ ~ \big| \E_{\vec{y} \sim \vec{E}} [ f(\vec{y}) ] - p \big| \ge 2 \delta ~ \Big] \le 2 e^{- \frac{2 m \delta^2}{4p + \delta} - \ln(\delta)} ~.
	\]
	\end{theorem}

\subsection{Organization of the Proofs}
We begin with the proofs of the technical lemmas, i.e., Theorem~\ref{thm:revenueMonotonicity} and Theorem~\ref{thm:prodconcentrate}, in Section~\ref{sec:rev_monotone_proof} and Section~\ref{sec:prodconcentrate} respectively.
Eager readers may skip these two sections without any problems understanding the sample complexity analysis in the rest of the paper.
The sample complexity bounds for single item auctions in the \crmodel are in Section \ref{sec:no_signal}, 
and the bounds for the \signalsmodel are in Section \ref{sec:signal}. 
The extensions to more general single parameter environments are in Section \ref{sec:extensions}.

\section{Proof of Strong Revenue Monotonicity}
\label{sec:rev_monotone_proof}

This section proves Theorem~\ref{thm:revenueMonotonicity}.
The proof in this version of the paper is different from the one in the original conference version, which has a subtle bug.
The new proof further allows us to handle general distributions (instead of those with discrete supports as in the original version) and the more general matroid constraint (instead of single-item auction as in the original version).

\subsection{Ironed Virtual Value}

We start with the formal definition of ironed virtual values, given the value distribution $D$ of a bidder. 
First, we define the \emph{quantile} (w.r.t.\ $D$) of a value $v$ in the support of $D$ to be:
\[
    q(v) = \Pr_{v' \sim D} \big[ v' > v \big]
    ~.
\]
For continuous distributions, we have $q_D(v) = 1 - F(v)$ where $F$ is the cdf of distribution $D$.
It induces an inverse mapping from quantiles to values such that given any quantile $0 \le q < 1$:
\[
    v(q) = \inf \{ v : q(v) \le q \}
    ~.
\]

Next, we define \emph{revenue curves}.
Let us start with the simpler case when the distribution is continuous.
The revenue curve, in the quantile space, is a function:
\[
    R(q) = q \cdot v(q)
    ~.
\]
That is, $R(q)$ equals the revenue obtained by setting a take-it-or-leave-it price such that the sale probability is equal to $q$.
Given any point $\big( q, R(q) \big)$ on the revenue curve, the slope of the line connecting the point with the origin is equal to the value $v(q)$ by definition.
Further, it is known that the slop of the tangent line equals the corresponding virtual value when the distribution is continuous.
See Figure~\ref{fig:revenue-curve} for an illustrative picture of the revenue curve.

Defining the revenue curve and virtual values for distributions with point masses is more subtle.
To begin with, plot points $(p, p \cdot v)$ for any value $v$ in the support and the corresponding sale probability $p = \Pr_{v' \sim D} \big[ v' \ge v \big]$.
Note that $p = q(v)$ if $v$ is not a point mass.
However, there will be discontinuities in the plot that correspond to the point masses.
Concretely, for any point mass $v$ with $q = \Pr_{v' \sim D} \big[ v' > v \big]$ and $p = \Pr_{v' \sim D} \big[ v' \ge v \big]$, the curve is currently undefined on the semi-open interval $(q, p]$;
extend the definition of the revenue curve by letting it be the straight segment connecting $\big( q, q \cdot v(q) \big)$ and $\big( p, p \cdot v \big)$ on this semi-open interval.
Note that $R(q)$ can be interpreted as the expected revenue obtained by setting a take-it-or-leave-it price that randomizes between $v$ and $v(q)$ appropriately such that the sale probability is equal to $q$.
Then, the virtual value of any value $v$ with quantile $q$ is the right derivative of the revenue curve at $q$.

For example, consider a discrete distribution with support $v_1 > v_2 > \dots > v_k$ with probability  masses  $p_1, p_2, \dots, p_k > 0$.
Let $P_i = \sum_{j=1}^i p_j$ for any $1 \le i \le k$.
Then, first plot the points $\big( P_i, P_i \cdot v_i \big)$, $1 \le i \le k$; let $\big( P_0, P_0 \cdot v_0 \big) = \big( 0, 0 \big)$ for notational convenience.
Note that $0 = P_0 < P_1 < \dots < P_k = 1$.
Connecting neighboring points, the revenue curve $R(q)$ is defined to be the following when the quantile $q$ is between $P_{i-1}$ and $P_i$:
\[
    R(q) = \frac{q - P_{i-1}}{p_i} \cdot P_i \cdot v_i + \frac{P_i - q}{p_i} \cdot P_{i-1} \cdot v_{i-1}
    ~.
\]
The virtual value of $v_i$ is equal to:
\[
    \frac{P_i \cdot v_i - P_{i-1} v_{i-1}}{p_i}
    ~.
\]
This is consistent with the definition of virtual values for discrete distributions, e.g., by \citet{Elkind2007designing}.

Following the notation in auction theory, let $\phi(v)$ denote the virtual value of a given value $v$.
The next lemma follows by the above definition of virtual values.

\begin{lemma}
    \label{lem:virtual_value}
    For any quantile $q'$, we have:
    \[
        \int_0^{q'} \phi\big(v(q)\big) dq = R(q')
        ~.
    \]
\end{lemma}

Finally, let $\bar{R}$ be the convex hull of $R$.
$\bar{R}(q)$ can be viewed as the maximum revenue achievable via a randomized take-it-or-leave-it price, subject to having an overall sale-probability of $q$.
Then, the \emph{ironed virtual value} of the value with quantile $q$ is equal to the slope of the tangent line w.r.t.\ $\bar{R}$.
See Figure~\ref{fig:revenue-curve-convex-hull} for an illustrative picture.

\begin{figure}
    \centering
    \begin{subfigure}{.48\textwidth}
        \includegraphics[width=\textwidth]{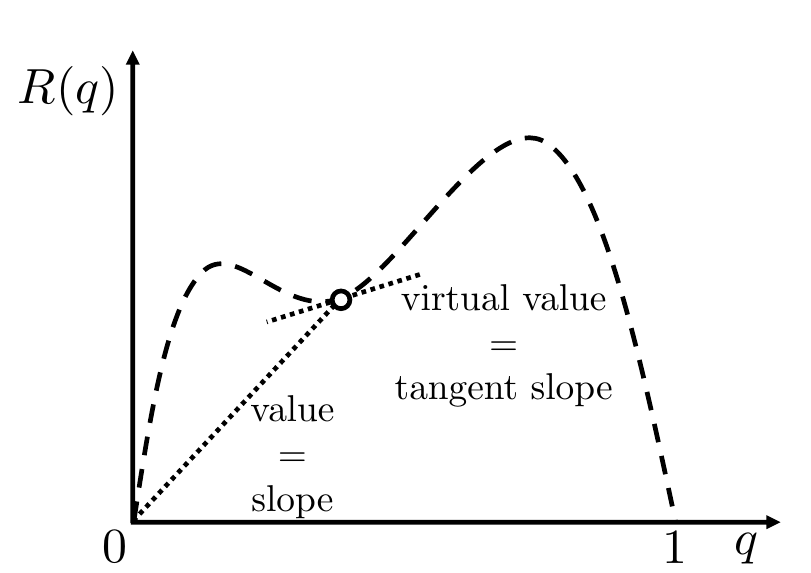}
        \caption{}
        \label{fig:revenue-curve}
    \end{subfigure}
    \begin{subfigure}{.48\textwidth}
        \includegraphics[width=\textwidth]{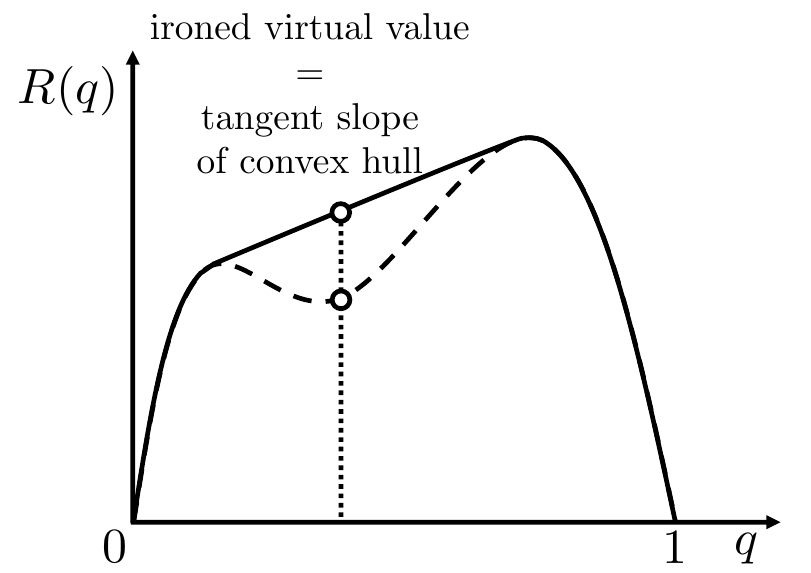}
        \caption{}
        \label{fig:revenue-curve-convex-hull}
    \end{subfigure}
    \caption{Illustrative pictures of the revenue curve (dashed curve) and its convex hull (solid curve) and how they relate to the value, virtual value, and ironed virtual values.}
\end{figure}

\subsection{Simplifications}
We start with some simplifying treatments that are wlog.
Recall that our goal is to show that an optimum auction for $\Dboldhat$ 
gets at least as much revenue when run on $\Dbold$. 
Let $\Mhat$ be the optimal Myerson auction for $\Dboldhat$. 
Let $\phi_i(v_i)$ and $\hat{\phi}_i(v_i)$ denote the virtual values that correspond to a value $v_i$ w.r.t.\ distribution $D_i$ and $\Dhat_i$, respectively, for any bidder $i \in [n]$.
Given any value profile, $\Mhat$ picks the bidder with the highest \emph{ironed} virtual value w.r.t.\ $\Dboldhat$ as the winner. 
We assume that $\Mhat$ breaks ties deterministically by the lexicographical order of bidders.
We write $\hat{\phi}_i(v_i) > \hat{\phi}_j(v_j)$ to mean that either the virtual value of $i$ is strictly larger than that of $j$, or they are equal and $i < j$.

For each bidder $i$, let $V_i$ denote the subset of values whose corresponding points on the revenue curve are on the convex hull.
We may assume wlog that both $\Dbold$ and $\Dboldhat$ have support $\vec{V} = V_1 \times V_2 \times \dots \times V_n$; the reason is as follows.
The {ironed} virtual value of bidder $i$ is the derivative of the convex hull of the revenue curve of $\Dhat_i$, in the quantile space.
We can interpret $\Mhat$ as rounding the value of each bidder $i$ down to the closest one in $V_i$, and then using the virtual values w.r.t.\ the distributions of the rounded values to choose the winner.
We do the same for values in the support of $\Dbold$ but not $\Dboldhat$.  
With this treatment, $\Dboldhat$ is regular in the sense that for every bidder $i$, every value in the support of $\Dhat_i$ is on the convex hull of its revenue curve and, thus, $\hat{\phi}_i$ is also the ironed virtual value.
Note that, however, $\Dbold$ may not be regular.

\subsection{Failing Strategies}

To build up the intuition and to motivate our approach, we continue with some brief discussions on how two seemingly natural strategies fail, even in the case of a single item.

\paragraph{Failing Strategy 1: Direct Comparison of Payments.}
Suppose we realize the values of bidders as follows.
First, draw a quantile profile $\vec{q}$ such that each $q_i$ is i.i.d.\ from the uniform distribution over $[0, 1)$.
Then, for every bidder $i \in [n]$, let $v_i(q_i)$ denote the value that corresponds to quantile $q_i$ w.r.t.\ $D_i$;
define $\hat{v}_i(q_i)$ similarly w.r.t.\ $\Dhat_i$.
Let $\vec{v}(\vec{q})$ and $\vec{\hat{v}}(\vec{q})$ be the corresponding value profiles.
Then, the expected revenues of $\Mhat$ w.r.t.\ $\Dbold$ and $\Dboldhat$ can be written as:
\[
\rev{\Mhat}{\Dbold} = \int \sum_{i = 1}^n \hat{p}_i \big( \vec{v}(\vec{q}) \big) d\vec{q}
~,
\]
and
\[
\rev{\Mhat}{\Dboldhat} = \int \sum_{i = 1}^n \hat{p}_i \big( \vec{\hat{v}}(\vec{q}) \big) d\vec{q}
~. 
\]

A naïve strategy is to show that $\sum_{i = 1}^n \hat{p}_i \big( \vec{v}(\vec{q}) \big) \ge \sum_{i = 1}^n \hat{p}_i \big( \vec{\hat{v}}(\vec{q}) \big)$ for any quantile profile $\vec{q}$.
This inequality does not hold in general because it is known that the optimal mechanism may involve price discrimination.
It is possible that, for example, as bidder $1$'s value gets larger moving from $\Dboldhat$ to $\Dbold$, she may become the winner even though her payment is still lower than what the original winner pays.

\paragraph{Failing Strategy 2: Amortization via Virtual Values.}

An alternative approach is to amortize the gains in the quantile space using virtual values.
Note that $\hat{\phi}_i$ is also the ironed virtual value 
due to the aforementioned treatments, i.e., it is monotonically non-decreasing in $\hat{v}_i$.
Note that similar claims do not hold for $\phi_i$ in general.
Let $\vec{\hat{x}}(\vec{v})$ denote the set of winners that $\Mhat$ picks when the value profiles is $\vec{v}$.
By Myerson's characterization~\cite{Myerson}, the expected revenues can be written as follows:
\[
    \rev{\Mhat}{\Dbold} = \int \sum_{i \in \vec{\hat{x}} ( \vec{v}(\vec{q}) )} \phi_i\big( v_i(q_i) \big) ~ d\vec{q}
    ~,
\]
and
\[
    \rev{\Mhat}{\Dboldhat} = \int \sum_{i \in \vec{\hat{x}} ( \vec{\hat{v}}(\vec{q}) )} \hat{\phi}_i \big( \hat{v}_i(q_i) \big) ~ d\vec{q}
    ~.
\]

Then, one may try to show that the amortized gain is weakly higher in $\Dbold$ than in $\Dboldhat$ for any quantile profile $\vec{q}$, i.e.:
\begin{equation}
\label{eqn:failing_strategy}
\sum_{i \in \vec{\hat{x}} ( \vec{v}(\vec{q}) )} \phi_i\big( v_i(q_i) \big) 
\ge 
\sum_{i \in \vec{\hat{x}} ( \vec{\hat{v}}(\vec{q}) )} \hat{\phi}_i \big( \hat{v}_i(q_i) \big)
~.
\end{equation}

However, this inequality does not hold in general either, even in the special case of selling a single item, i.e., when $\vec{\hat{x}}(\vec{v})$ is either empty or a singleton set for any value profile $\vec{v}$.
If $\Mhat$ picks the same winner for $\vec{v}(\vec{q})$ and $\vec{\hat{v}}(\vec{q})$, inequality \eqref{eqn:failing_strategy} may not hold because stochastic dominance does not imply dominance in terms of virtual values, in the quantile space, that is: 
\[
    \phi_i \big( v_i(q_i) \big) \ge \hat{\phi}_i \big( \hat{v}_i(q_i) \big)
\] 
may not hold in general.

If $\Mhat$ picks different winners, e.g., it picks bidder $i$ when the value profile is $\vec{v}(\vec{q})$, and picks bidder $j$ when the value profile is $\vec{\hat{v}}(\vec{q})$.
Then, inequality \eqref{eqn:failing_strategy} becomes:
\[
    \phi_i\big( v_i(q_i) \big) \ge \hat{\phi}_j \big( \hat{v}_j(q_j) \big)
    ~.
\]
However, the definition of $\Mhat$ indicates a different inequality that:
\[
    \hat{\phi}_i\big(v_i(q_i)\big) \ge \hat{\phi}_j\big( v_j(q_j) \big)
    ~,
\]
since it picks the winner with the highest virtual value w.r.t.\ $\Dboldhat$, when the value profile is $\vec{v}(\vec{q})$.
Again, inequality \eqref{eqn:failing_strategy} may not hold in general due to the subtle mismatches on both sides of the two inequalities.

\subsection{Proof of Theorem~\ref{thm:revenueMonotonicity}}

\subsubsection{Our Amortization}
%
We will design a more subtle amortization.
Specifically, we will design non-negative amortized gain functions $\vec{g}$ and $\vec{\hat{g}}$ for distributions $\Dbold$ and $\Dboldhat$ respectively, with the following properties.
%
\begin{enumerate}[label=(\alph*)]
	\item
	\label{property:amortization_revenue_bounds}
	The expected revenue is bounded by the expectation amortized gain for both distributions, in the respective directions:
    \begin{equation}
        \label{eqn:amortization_revenue_D}
        \rev{\Mhat}{\Dbold} \ge \int \sum_{i \in \vec{\hat{x}} ( \vec{v}(\vec{q}) )} g_i(\vec{q}) ~ d\vec{q}
        ~.
    \end{equation}
	and
    \begin{equation}
        \label{eqn:amortization_revenue_Dhat}
        \rev{\Mhat}{\Dboldhat} \le \int \sum_{i \in \vec{\hat{x}} ( \vec{\hat{v}}(\vec{q}) )} \hat{g}_i(\vec{q}) ~ d\vec{q} 
        ~.
    \end{equation}		
	\item
	\label{property:amortization_same_winner}
    For any quantile profile, and any common winner picked by the mechanism $\Mhat$ for both distributions, the amortized gain w.r.t.\ $\Dbold$ is weakly higher than that w.r.t.\ $\Dboldhat$.
    That is, for any bidder $i \in \vec{\hat{x}} ( \vec{v}(\vec{q}) ) \cap \vec{\hat{x}} ( \vec{\hat{v}}(\vec{q}) )$, we have:
    \begin{equation}
        \label{eqn:amortization_common_winner}
        g_i(\vec{q}) \ge \hat{g}_i(\vec{q})
        ~.
    \end{equation}
	\item
	\label{property:amortization_different_winners}
    For any bidder $i \in \vec{\hat{x}}\big( \vec{v}(\vec{q}) \big) \setminus \vec{\hat{x}} \big( \vec{\hat{v}}(\vec{q}) \big)$ who is a winner only w.r.t.\ $\Dbold$, and any bidder $j \in \vec{\hat{x}}\big(\vec{\hat{v}}(\vec{q}) \big) \setminus \vec{\hat{x}}\big(\vec{v}(\vec{q})\big)$ who is a winner only w.r.t.\ $\Dboldhat$, such that replacing bidder $i$ with bidder $j$ in the set of winners w.r.t.\ $\Dbold$ is feasible, i.e., $\vec{\hat{x}}\big( \vec{v}(\vec{q}) \big) - \{i\} + \{j\} \in \allocset$, we have that 
	\[
	g_i(\vec{q}) \ge \hat{g}_j(\vec{q})
	~.
	\]
	We do not know how to do this directly; instead we do this in two steps. 
    For any bidder $i \in \vec{\hat{x}} \big( \vec{v}(\vec{q}) \big) \setminus \vec{\hat{x}} \big( \vec{\hat{v}}(\vec{q}) \big)$ who is a winner only w.r.t.\ $\Dbold$, and any non-winner $j \notin \vec{\hat{x}}\big(\vec{v}(\vec{q})\big)$ who could have been swapped with bidder $i$ while maintaining feasibility, i.e., $\vec{\hat{x}}\big(\vec{v}(\vec{q})\big) - \{i\} + \{j\} \in \allocset$, we have:
    \begin{equation}
        \label{eqn:amortization_D_winner}
        g_i(\vec{q}) \ge \hat{\phi}_j \big( v_j(q_j) \big)
        ~.
    \end{equation}
    For any bidder $j \in \vec{\hat{x}} \big( \vec{\hat{v}}(\vec{q}) \big) \setminus \vec{\hat{x}} \big( \vec{v}(\vec{q}) \big)$ who is a winner only w.r.t.\ $\Dboldhat$, we have:
    \begin{equation}
        \label{eqn:amortization_Dhat_winner}
        \hat{g}_j(\vec{q}) \le \hat{\phi}_j \big( v_j(q_j) \big)
        ~.
    \end{equation}

\end{enumerate}

\paragraph{Amortized Analysis.}
Next, we present the proof of the theorem assuming the existence of such amortized gain functions $\vec{g}$ and $\vec{\hat{g}}$, deferring the construction of these functions to the next subsection.
By property~\ref{property:amortization_revenue_bounds}, to prove that $\rev{\Mhat}{\Dbold} \ge \rev{\Mhat}{\Dboldhat}$, it suffices to show:
\[
    \int \sum_{i \in \vec{\hat{x}} ( \vec{v}(\vec{q}) )} g_i(\vec{q}) ~ d\vec{q}
    \ge
    \int \sum_{i \in \vec{\hat{x}} ( \vec{\hat{v}}(\vec{q}) )} \hat{g}_i(\vec{q}) ~ d\vec{q} 
    ~.
\]

In fact, we will prove a stronger claim that the inequality holds for any given quantile profile $\vec{q}$, rather than only in expectation over the randomness of $\vec{q}$. 
That is, for any $\vec{q}$:
\begin{equation}
    \label{eqn:amortization_fixed_quantile}
    \sum_{i \in \vec{\hat{x}} ( \vec{v}(\vec{q}) )} g_i(\vec{q}) ~ d\vec{q}
    \ge
    \sum_{i \in \vec{\hat{x}} ( \vec{\hat{v}}(\vec{q}) )} \hat{g}_i(\vec{q}) ~ d\vec{q} 
    ~.
\end{equation}

We use the well known exchange property of matroids to start from the 
winners in $\Dbold$ and move towards the winners in $\Dboldhat$. 
\begin{lemma}[Exchange property, e.g., \citet{Schrijver/2003/Springer}, Corollary~39.12a]
	\label{lem:matroid_exchange}
	For any $\vec{x}, \vec{x'} \in \allocset$ of equal size, i.e., $\big| \vec{x} \big| = \big| \vec{x'} \big|$, there is a one-to-one mapping $\sigma$ from $\vec{x} \setminus \vec{x'}$ and $\vec{x'} \setminus \vec{x}$ such that for any $i \in \vec{x} \setminus \vec{x'}$, we have:
	\[
	\vec{x} - \big\{ i \big\} + \big\{ \sigma(i) \big\} \in \allocset
	~.
	\]
\end{lemma}

We need to show one more thing, that the number of winners in $\Dbold$ is at least the number of winners in $\Dboldhat$.
This is fairly easy, and follows from the 
monotonicity of ironed virtual values, and the exchange property.

\begin{lemma}
    \label{lem:amortization_size}
    For any quantile profile $\vec{q}$, we have:
    \[
        \big| \vec{\hat{x}} ( \vec{v}(\vec{q}) ) \big|
        \ge
        \big| \vec{\hat{x}} ( \vec{\hat{v}}(\vec{q}) ) \big|
        ~.
    \]
\end{lemma}

\begin{proof}
    Suppose not.
    Then, by the augmentation property of matroids, we get that there exists a bidder $j \in \vec{\hat{x}} (\vec{\hat{v}}(\vec{q}) ) \setminus \vec{\hat{x}} ( \vec{v}(\vec{q}) )$ who is a winner only w.r.t.\ $\Dboldhat$ such that mechanism $\Mhat$ could have added bidder $j$ to the set of winners w.r.t.\ $\Dbold$, i.e.:
    \[
            \vec{\hat{x}} \big( \vec{v}(\vec{q}) \big) + \big\{ j \big\} \in \allocset
        ~.
    \]
   
    Recall that $\Mhat$ is the optimal auction w.r.t.\ $\Dboldhat$, which picks a feasible set of winners to allocate to to maximize the sum of their ironed virtual values.
    Here, and in the rest of the proof, ironed virtual values are evaluated according to distributions $\Dboldhat$.
    The decision of not including bidder $j$ in the set of winners indicates that her ironed virtual value, for value $v_j(q_j)$, is negative.

    However, by that $\Dbold \succeq \Dboldhat$, we have $v_j(q_j) \ge \hat{v}_j(q_j)$.
    Monotonicity of ironed virtual values implies that the ironed virtual value of $v_j(q_j)$ is at least  that of $\hat{v}_j(q_j)$, which must be non-negative since mechanism $\Mhat$ chooses bidder $j$ when the value profile is $\vec{\hat{v}}(\vec{q})$.
    We have a contradiction.
\end{proof}

%

By Lemma~\ref{lem:amortization_size}, there are weakly more winners in $\vec{\hat{x}}(\vec{v}(\vec{q}))$ than in $\vec{\hat{x}}(\vec{\hat{v}}(\vec{q}))$.
Let $\vec{x'} = \vec{\hat{x}}(\vec{\hat{v}}(\vec{q}))$.
Let $\vec{x}$ be a subset of $\vec{\hat{x}}(\vec{v}(\vec{q}))$, obtained by removing some bidders in $\vec{\hat{x}}(\vec{v}(\vec{q})) \setminus \vec{\hat{x}}(\vec{\hat{v}}(\vec{q}))$ so that it has the same size as $\vec{x'}$.
Let $\sigma$ be the mapping from $\vec{x} \setminus \vec{x'}$ to $\vec{x'} \setminus \vec{x}$ in Lemma~\ref{lem:matroid_exchange}.
Then, the next lemma follows as a corollary of property~\ref{property:amortization_different_winners}.

\begin{lemma}
    \label{lem:amortization_exchange}
    For any $i \in \vec{x} \setminus \vec{x'}$, we have:
    \[
        \hat{g}_i(\vec{q}) \ge \hat{g}_{\sigma(i)}(\vec{q})
        ~.
    \]
\end{lemma}

%
%
Then, Eqn.~\eqref{eqn:amortization_fixed_quantile} follows from a sequence of inequalities below:
\begin{align*}
    \sum_{i \in \vec{\hat{x}} ( \vec{v}(\vec{q}) )} g_i(\vec{q}) 
    &
    \ge \sum_{i \in \vec{x}} g_i(\vec{q})
    && \text{(non-negativity of $\vec{g}$)} \\
    & 
    = \sum_{i \in \vec{x} \cap \vec{x'}} g_i(\vec{q}) + \sum_{i \in \vec{x} \setminus \vec{x'}} g_i(\vec{q}) \\
    & 
    = \sum_{i \in \vec{x} \cap \vec{x'}} \hat{g}_i(\vec{q}) + \sum_{i \in \vec{x} \setminus \vec{x'}} g_i(\vec{q}) 
    && 
    \text{(property~\ref{property:amortization_same_winner})} \\
    & 
    \ge \sum_{i \in \vec{x} \cap \vec{x'}} \hat{g}_i(\vec{q}) + \sum_{i \in \vec{x} \setminus \vec{x'}} \hat{g}_{\sigma(i)}(\vec{q})
    && 
    \text{(Lemma~\ref{lem:amortization_exchange})} \\
    &
    = \sum_{i \in \vec{x} \cap \vec{x'}} \hat{g}_i(\vec{q}) + \sum_{i \in \vec{x'} \setminus \vec{x}} \hat{g}_i(\vec{q}) 
    &&
    \text{($\sigma$ is a one-to-one mapping)} \\
    &
    = \sum_{i \in \vec{x'}} \hat{g}_i(\vec{q}) \\
    &
    = \sum_{i \in \vec{\hat{x}}(\vec{\hat{v}}(\vec{q}))} \hat{g}_i(\vec{q})
    ~.
    &&
    \text{(definition of $\vec{x'}$)}
\end{align*}

\subsubsection{Construction of the Amortized Gain Functions}

It remains to show that there are such non-negative amortized gain functions $\vec{g}$ and $\vec{\hat{g}}$ with the aforementioned properties.
We start with some notations and simplifications.

The first simplification relies on the notion of threshold quantile above which a bidder becomes a winner.
For any bidder $i \in [n]$, and any quantile profile $\vec{q}_{-i}$ of the bidders other than $i$, let $\theta_i(\vec{q}_{-i})$ denote the threshold quantile below which $i$ is the winner when $\Mhat$ is run on $\Dbold$. 
That is:
\[
    \theta_i(\vec{q}_{-i}) = \sup \Big\{ ~ 0 \le q < 1 : i \in \vec{\hat{x}} \big( v_i(q), \bm{v}_{-i}(q_{-i}) \big) ~ \Big\}
    ~.
\]

Define $\hat{\theta}_i(\vec{q}_{-i})$ similarly as the threshold quantile when $\Mhat$ is run on $\Dboldhat$. 
That is:
\[
    \hat{\theta}_i(\vec{q}_{-i}) = \sup \Big\{ 0 \le q < 1 : i \in \vec{\hat{x}} \big( \hat{v}_i(q), \bm{\hat{v}}_{-i}(q_{-i}) \big) \Big\}
    ~.
\]

Then, we may wlog let $g_i(q_i, \vec{q}_{-i}) = 0$ for any quantile $q_i > \theta_i(\vec{q}_{-i})$ of bidder $i$ that is below the threshold, since $x_i \big( \vec{v}(q_i, \vec{q}_{-i}) \big) = 0$ and thus, they are irrelevant to the properties that need to be satisfied.
Similarly, we may wlog let $\hat{g}_i(q_i, \vec{q}_{-i}) = 0$ for any $q_i > \hat{\theta}_i(\vec{q}_{-i})$.
It remains to design $g_i$ and $\hat{g}_i$ for quantile $q_i$ below the corresponding thresholds.

Property~\ref{property:amortization_revenue_bounds} is where the amortization happens: 
in particular, we only need to amortize over $i$'s own randomization, 
for any fixed profile of all other bidders. We formalize that below. 
We begin by rewriting the expected revenue of running $\Mhat$ on $\Dbold$ and $\Dboldhat$ respectively using the virtual values \emph{à la} \citet{Myerson}:
\[
    \rev{\Mhat}{\Dbold} = \sum_{i = 1}^n \int_{[0, 1]^{n-1}} \int_0^{\theta_i(\vec{q}_{-i})} \phi_i(v_i(q_i)) ~ d q_i ~ d \vec{q}_{-i}
	~,
\]
%
and
%
\[
    \rev{\Mhat}{\Dboldhat} = \sum_{i = 1}^n \int_{[0, 1]^{n-1}} \int_0^{\hat{\theta}_i(\vec{q}_{-i})} \hat{\phi}_i(\hat{v}_i(q_i)) ~ d q_i ~ d \vec{q}_{-i}
~.
\]

Therefore, Inequalities~\eqref{eqn:amortization_revenue_D} and \eqref{eqn:amortization_revenue_Dhat} in property~\ref{property:amortization_revenue_bounds} become:
\[
    \sum_{i = 1}^n \int_{[0, 1]^{n-1}} \int_0^{\theta_i(\vec{q}_{-i})} \phi_i(v_i(q_i)) ~ d q_i ~ d \vec{q}_{-i} 
    \ge
    \sum_{i = 1}^n \int_{[0, 1]^{n-1}} \int_0^{\theta_i(\vec{q}_{-i})}  g_i(\vec{q}) ~ d q_i ~ d \vec{q}_{-i}  
    ~,
\]
and
\[
    \sum_{i = 1}^n \int_{[0, 1]^{n-1}} \int_0^{\hat{\theta}_i(\vec{q}_{-i})} \hat{\phi}_i(\hat{v}_i(q_i)) ~ d q_i ~ d \vec{q}_{-i}
    \le
    \sum_{i = 1}^n \int_{[0, 1]^{n-1}} \int_0^{\hat{\theta}_i(\vec{q}_{-i})} \hat{g}_i(\vec{q}) ~ d q_i ~ d \vec{q}_{-i}
    ~.
\]

To satisfy property~\ref{property:amortization_revenue_bounds}, it suffices to ensure that for any bidder $i$ and any quantile profile $\vec{q}_{-i}$ of the other bidders, we have:
\begin{equation}
\label{eqn:amortization_bound_D}
\int_0^{\theta_i(\vec{q}_{-i})} \phi_i(v_i(q_i)) ~ d q_i 
\ge
\int_0^{\theta_i(\vec{q}_{-i})}  g_i(\vec{q}) ~ d q_i 
~,
\end{equation}
and
\begin{equation}
\label{eqn:amortization_bound_Dhat}
\int_0^{\hat{\theta}_i(\vec{q}_{-i})} \hat{\phi}_i(\hat{v}_i(q_i)) ~ d q_i
\le
\int_0^{\hat{\theta}_i(\vec{q}_{-i})} \hat{g}_i(\vec{q}) ~ d q_i
~.
\end{equation}

\paragraph{Construction of $\vec{\hat{g}}$.}
The properties concerning only $\vec{\hat{g}}$, i.e., Eqn.~\eqref{eqn:amortization_Dhat_winner} and Eqn.~\eqref{eqn:amortization_bound_Dhat}, suggest an obvious choice of the amortized gain function $\vec{\hat{g}}$ for distribution $\Dboldhat$ when the quantile is above the threshold, namely, letting it equal the virtual value w.r.t.\ $\Dboldhat$.
Concretely, we let:
\[
    \hat{g}_i(\vec{q}) = 
    \begin{cases}
        \hat{\phi}_i(\hat{v}_i(q_i)) & \text{if $q_i \le \hat{\theta}_i(\vec{q}_{-i})$} \\
        0 & \text{if $q_i > \hat{\theta}_i(\vec{q}_{-i})$}
    \end{cases}
\]

It is easy to verify that it satisfies Eqn.~\eqref{eqn:amortization_bound_Dhat} with equality.
Further, Eqn.~\eqref{eqn:amortization_Dhat_winner} becomes:
\[
    \hat{g}_i(\vec{q}) = \hat{\phi}_i(\hat{v}_i(q_i)) \le \hat{\phi}_i(v_i(q_i))
    ~,
\]
which holds due to the monotonicity of $\hat{\phi}_i$ and that $\hat{v}_i(q_i) \le v_i(q_i)$.
The latter follows by $\Dbold \succeq \Dboldhat$.


\paragraph{Construction of $\vec{g}$.}
%
This part is the crux of the proof and needs to be further divided into two cases.
Concretely, let us fix any bidder $i$, and any quantile profile $\vec{q}_{-i}$ of the bidders other than $i$.
We will construct the amortized gain function $g_i$ for quantile profiles of the form $(q_i, \vec{q}_{-i})$ for any $0 \le q_i < \theta_i(\vec{q}_{-i})$ to ensure the relevant inequalities in the properties.
In the following discussions, let $R_i(q_i)$ be the revenue curve, in the quantile space, w.r.t.\ distribution $D_i$.
Define $\widehat{R}_i$ similarly w.r.t.\ distribution $\Dhat_i$.
The cases depend on the order of the threshold quantiles for bidder $i$ in the two distributions given $\vec{q}_{-i}$.

\paragraph{Case 1: $\theta_i(\vec{q}_{-i}) \le \hat{\theta}_i(\vec{q}_{-i})$.}
This is the easy in which we just let $g$ be the virtual value.
For any $0 \le q_i < \theta_i(\vec{q}_i)$, let:
\[
    g_i(q_i, \vec{q}_{-i}) = \hat{\phi}_i\big(\hat{v}_i(q_i)\big)
    ~.
\]

The relevant inequality for property~\ref{property:amortization_revenue_bounds}, i.e., Eqn.~\eqref{eqn:amortization_bound_D}, holds as follows: 
\begin{align*}
    \int_0^{\theta_i(\vec{q}_{-i})} \phi_i(v_i(q_i)) ~ d q_i 
    & = R_i \big( \theta_i(q_{-i}) \big)
    && \text{(definition of virtual values)} \\[1ex]
& \ge \widehat{R}_i \big( \theta_i(q_{-i}) \big)
&& \text{($\Dbold \succeq \Dboldhat$)} \\[1ex]
& = \int_0^{\theta_i(q_{-i})} \hat{\phi}_i(\hat{v}_i(q_i)) ~ d q_i
&& \text{(definition of virtual values)} \\
& = \int_0^{\theta_i(q_{-i})} g_i(q_i, \vec{q}_{-i}) ~ d q_i ~.
&& \text{(definition of $g_i$)} 
\end{align*}

The relevant inequality for property~\ref{property:amortization_same_winner} is Eqn.~\eqref{eqn:amortization_common_winner} when $0 \le q_i \le \theta_i(\vec{q}_{-i})$, i.e.:
\[
    g_i(q_i, \vec{q}_{-i}) \ge \hat{g}_i(q_i, \vec{q}_{-i})
    ~.
\]
By our construction, it holds with equality.

We do not need to worry about property~\ref{property:amortization_different_winners}, 
since in this case we have $q_i < \theta_i(\vec{q}_{-i}) \le \hat{\theta}_i(\vec{q}_{-i})$ and thus bidder $i$ is a winner  in both $\Dbold$ 
and $\Dboldhat$. 


\paragraph{Case 2: $\theta_i(\vec{q}_{-i}) > \hat{\theta}_i(\vec{q}_{-i})$.}
\begin{figure}
	\centering
	\includegraphics[width=0.7\textwidth]{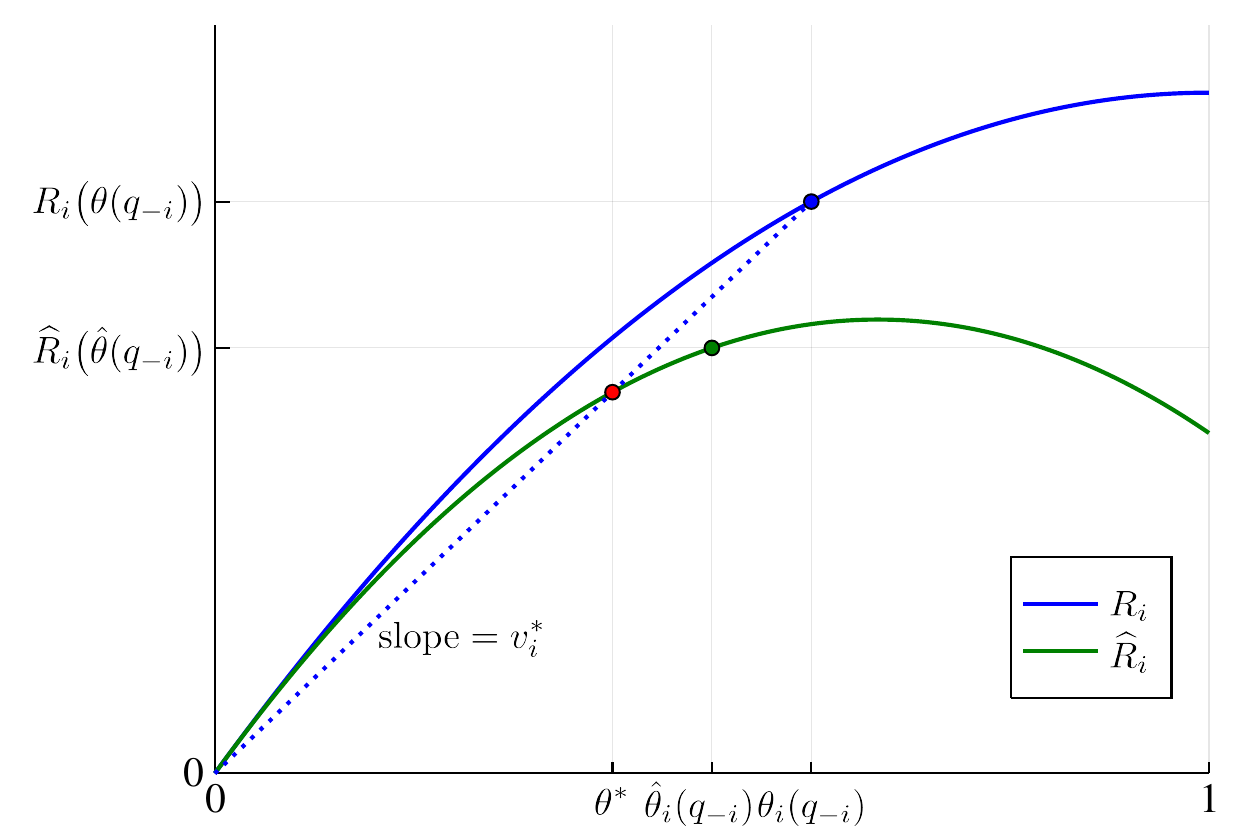}
	\caption{An illustrative picture of the 2nd case. The amortized gain $g_i$ moves along the revenue curve $\widehat{R}_i$ from $0$ to $\theta^*$, and then moves along the dotted straight line with slope $v_i^*$ from $\theta^*$ to $\theta_i(\vec{q}_{-i})$, meeting the revenue curve $R_i$ at $\theta_i(\vec{q}_{-i})$. The amortized gain function $\hat{g}_i$ simply moves along the revenue curve of $\widehat{R}_i$ from $0$ to $\hat{\theta}_i(\vec{q}_{-i})$.}
	\label{fig:revenue_monotone_proof}
\end{figure}

We first give the geometric intuition behind the construction. 
Suppose we define $g_i$ as the slope of a curve that starts at 
the origin. 
We first reinterpret what the three properties mean in terms of this curve and the revenue curves $R_i$ and $\widehat{R}_i$. 
\begin{description}
	\item[{Property~\ref{property:amortization_revenue_bounds}:}]
	The curve should meet the revenue curve $R_i$ at $\theta_i$. 
	\item[Property~\ref{property:amortization_same_winner}:]
	The slope of the curve should be at least that of $\widehat{R}_i$, in the interval $[0,\hat{\theta}_i]$. 
	\item[Property~\ref{property:amortization_different_winners}:]
	This corresponds to the interval $[\hat{\theta}_i, \theta_i]$.
	This property is hard to interpret geometrically, but a sufficient condition will be that the slope is at least  
	the threshold \emph{value} above which bidder $i$ is a winner in $\Dbold$. We denote this threshold value by $v_i^*$. 
\end{description}
Refer to Figure~\ref{fig:revenue_monotone_proof} for an illustration. 
The value $v_i^*$ is the slope of the line joining the origin to 
the revenue curve $R_i$ at $\theta_i$. 
The curve defining $g_i$ follows $\widehat{R}_i$ until 
it touches this line, and then follows this straight line until 
it touches $R_i$. 
It is easy to verify that this curve satisfies the geometric interpretation of the 3 properties as described above. 
We now give a formal description of this construction, and complete the proof.

Let $v_i^*$ be $v_i \big( \theta_i(\vec{q}_{-i}) \big)$, i.e., the threshold value that corresponds to the threshold quantile $\theta_i(\vec{q}_{-i})$ w.r.t.\ $D_i$ such that $i$ becomes the winner when its value is at least $v_i^*$, when mechanism $\Mhat$ is run on $\Dbold$, and when the other bidders quantile profile is $\vec{q}_{-i}$.
By its definition, we have:
\[
    v_i^* = \inf \Big\{ v : i \in \vec{\hat{x}} \big(v, \vec{v}_{-i}(\vec{q}_{-i}) \big) \Big\}
    ~.
\]

Let $\theta^*$ denote the sale probability of $v_i^*$ w.r.t.\ $\Dhat_i$.

By $\Dbold \succeq \Dboldhat$, we have $\theta^* \le \theta_i(\vec{q}_{-i})$.
We further need the following technical lemma.

\begin{lemma}
	$\theta^* \le \hat{\theta}_i(\vec{q}_{-i})$.
\end{lemma}

\begin{proof}
	Similarly, let $\hat{v}_i^*$ be $\hat{v}_i \big( \hat{\theta}_i(\vec{q}_{-i}) \big)$, i.e., the threshold value above which $i$ wins when $\Mhat$ is run on $\Dboldhat$ and when the other bidders' quantile profile is $\vec{q}_{-i}$.
	By its definition, we have that:
	\[
        \hat{v}_i^* = \inf \Big\{ v : i \in \vec{\hat{x}} \big(v, \vec{\hat{v}}_{-i}(\vec{q}_{-i}) \big) \Big\}
	\]
	
	For any $j \ne i$, by $\Dbold \succeq \Dboldhat$ and the monotonicity of $\hat{\phi}_j$, we have that $v_j(q_j) \ge \hat{v}_j(q_j)$ and, thus, $\hat{\phi}_j\big(v_j(q_j)\big) \ge \hat{\phi}_j\big(\hat{v}_j(q_j)\big)$.
	Therefore, we have that $v_i^* \ge \hat{v}_i^*$ by their definitions, i.e., the threshold value above which $i$ wins against $\vec{v}_{-i}(\vec{q}_{-i})$ is weakly higher than that against $\vec{\hat{v}}_{-i}(\vec{q}_{-i})$, because the former is weakly larger coordinate-wise.
	The lemma follows since $\theta^*$ and $\hat{\theta}_i(\vec{q}_{-i})$ are the quantiles of $v_i^*$ and $\hat{v}_i^*$ in $\Dhat_i$ respectively.
\end{proof}

We now define the amortized gain function $g_i$ as follows.
%
\[
    g_i(q_i, \vec{q}_{-i}) = 
    \begin{cases}
        \hat{\phi}_i\big(\hat{v}_i(q_i)\big) & \text{if $0 \le q_i \le \theta^*$} \\
        v_i^* & \text{if $\theta^* < q_i \le \theta_i(\vec{q}_{-i})$} \\
        0 & \text{if $\theta_i(\vec{q}_{-i}) \le q_i < 1$}
    \end{cases}
\]
%

The relevant inequality for property~\ref{property:amortization_revenue_bounds}, i.e., Eqn.~\eqref{eqn:amortization_bound_D}, holds with equality as follows:
\begin{align*}
    \int_0^{\theta_i(\vec{q}_{-i})} \phi_i(v_i(q_i)) ~ d q_i & = R_i \big( \theta_i(\vec{q}_{-i}) \big)
    && \text{(definition of virtual values)} \\[1ex]
    & = \theta_i(\vec{q}_{-i}) v_i^* 
    && \text{(definition of $v_i^*$)} \\[1ex]
    & = \theta^* v_i^* + \int_{\theta^*}^{\theta(\vec{q}_{-i})} v_i^* ~ dq_i \\
    & = \widehat{R}_i \big( \theta^* \big) + \int_{\theta^*}^{\theta(\vec{q}_{-i})} v_i^* ~ dq_i
    && \text{(definition of $\theta^*$)} \\
    & = \int_0^{\theta^*} \hat{\phi}_i(q_i) ~ d q_i + \int_{\theta^*}^{\theta(\vec{q}_{-i})} v_i^* ~ d q_i 
    && \text{(definition of virtual values)} \\
    & = \int_0^{\theta(\vec{q}_{-i})} g_i(q_i, \vec{q}_{-i}) ~ dq_i
    && \text{(definition of $g_i$)}
    ~.
\end{align*}


The relevant inequality for property~\ref{property:amortization_same_winner} is Eqn.~\eqref{eqn:amortization_common_winner} when $0 \le q_i \le \hat{\theta}_i(\vec{q}_{-i})$.
It holds trivially with equality for $0 \le q_i \le \theta^*$ by our construction.
It remains to consider $\theta^* < q_i \le \hat{\theta}_i(\vec{q}_{-i})$.
For any $q_i$ in this range, we have:
\[
    g_i(q_i, \vec{q}_{-i}) = v_i^* \ge \hat{v}_i(q_i) \ge \hat{\phi}_i(\hat{v}_i(q_i)) = \hat{g}_i(q_i, \vec{q}_{-i})
~.
\]

Finally, the relevant inequality for property~\ref{property:amortization_different_winners} is Eqn.~\eqref{eqn:amortization_bound_D} when $\hat{\theta}_i(\vec{q}_{-i}) < q_i \le \theta_i(\vec{q}_{-i})$.
For any $q_i$ in this range, and any $j$ that could have been swapped with $i$ while maintaining feasibility, we have:
\[
    g_i(q_i, \vec{q}_{-i}) = v_i^* \ge \hat{\phi}_i(v_i^*) \ge \hat{\phi}_j\big( v_j(q_j) \big)
    ~,
\]
where the last inequality follows by the definition of $v_i^*$ and that $\Mhat$ is a virtual value maximizer.

\section{Proof of Concentration Inequality for Product Distributions}
\label{sec:prodconcentrate}

In this section we prove Theorems~\ref{thm:prodconcentrate}.
The following lemma follows by Bernstein inequality.

\begin{lemma}
	\label{lem:chernoff}
	Let $\vec{y}_1, \vec{y}_2 \dots, \vec{y}_m$ be $m$ i.i.d.\ samples from $\vec{D}$.
	Then, we have
	\[
	\textstyle \Pr \left[ \big| \frac{1}{m} \sum_{j = 1}^m f(\vec{y}_j) - p \big| \ge \delta \right] \le 2 e^{- \frac{2 m \delta^2}{4p + \delta}} .
	\]
\end{lemma}

We relate the random variables in Theorem \ref{thm:prodconcentrate} and Lemma \ref{lem:chernoff} as follows.
First, draw samples $x_{i1}, x_{i2}, \dots, x_{im}$ i.i.d.\ from $D_i$ for all $i \in [n]$, and draw $n$ permutations $\pi_1, \dots, \pi_n$ of $[m]$ independently and uniformly at random.
Then, let $\vec{y}_j(\vec{x}, \vec{\pi})$ be $(x_{1\pi_1(j)}, x_{2\pi_2(j)}, \dots, x_{n\pi_n(j)})$ for all $j \in [m]$.
We have the following properties which will be useful in the proof:
\begin{enumerate}
	\item[a)] $\vec{y}_j(\vec{x}, \vec{\pi})$ are i.i.d.\ samples from $\vec{D}$ over the randomness of both $\vec{x}$ and $\vec{\pi}$.
	\item[b)] Fixed any $\vec{x}$ (and, thus, $\vec{E}$), for any $j$, $\vec{y}_j(\vec{x}, \vec{\pi})$ follows distribution $\vec{E}$ over the randomness of $\vec{\pi}$.
	Different $\vec{y}_j$'s are, however, correlated. 
\end{enumerate}

Next, we bound the difference between the expectation over the product empirical distribution $\vec{E}$ and the true expectation $p$.
We have the following inequalities:
%
%
%
\begin{align*}
\textstyle \big| \E_{\vec{y} \sim \vec{E}} [ f(\vec{y}) ] - p \big|
& \textstyle = \frac{1}{m} \sum_{j = 1}^m \big| \E_{\pi} [ f \big( \vec{y}_j(\vec{x}, \vec{\pi}) \big) ] - p \big| && \text{(Property b: $\vec{y}_j(\vec{x}, \vec{\pi}) \sim \vec{E}$)} \\
& \le \textstyle	\E_{\vec{\pi}} \left[ \frac{1}{m} \sum_{j = 1}^m \big| f \big( \vec{y}_j(\vec{x}, \vec{\pi}) \big) - p \big| \right] && \text{(Convexity of absolute value)} \\
& \le \textstyle	\Pr_{\vec{\pi}} \left[ \big| \frac{1}{m} \sum_{j = 1}^m f \big( \vec{y}_j(\vec{x}, \vec{\pi}) \big) - p \big| \ge \delta \right] \cdot 1 \\
& \textstyle \quad + \Pr_{\vec{\pi}} \left[ \big| \frac{1}{m} \sum_{j = 1}^m f \big( \vec{y}_j(\vec{x}, \vec{\pi}) \big) - p \big| < \delta \right] \cdot \delta && \text{(By $f(\cdot) \in [0, 1]$)} \\
& \le \textstyle \Pr_{\vec{\pi}} \left[ \big| \frac{1}{m} \sum_{j = 1}^m f \big( \vec{y}_j(\vec{x}, \vec{\pi}) \big) - p \big| \ge \delta \right] + \delta ~.
\end{align*}

Hence, we get that:
\begin{equation}
\label{eqn:prodconcentrate_reduction}
\textstyle
\Pr_{\vec{x}} \left[ \big| \E_{\vec{y} \sim \vec{E}} [ f(\vec{y}) ] - p \big| \ge 2 \delta \right] \le \Pr_{\vec{x}} \left[ \Pr_{\vec{\pi}} \big[ \big| \frac{1}{m} \sum_{j = 1}^m f \big( \vec{y}_j(\vec{x}, \vec{\pi}) \big) - p \big| \ge \delta \big] \ge \delta \right] ~.
\end{equation}

It remains the bound the right hand side.
By property a) and Lemma \ref{lem:chernoff}, we have:
\[
\textstyle \Pr_{\vec{x}, \vec{\pi}} \left[ \big| \frac{1}{m} \sum_{j = 1}^m f \big( \vec{y}_j(\vec{x}, \vec{\pi}) \big) - p \big| \ge \delta \right] \le 2 e^{- \frac{2 m \delta^2}{4p + \delta}}.
\]

Equivalently,
\[
\textstyle \E_{\vec{x}} \left[ \Pr_{\vec{\pi}} \left[ \big| \frac{1}{m} \sum_{j = 1}^m f \big( \vec{y}_j(\vec{x}, \vec{\pi}) \big) - p \big| \ge \delta \right]  \right] \le 2 e^{- \frac{2 m \delta^2}{4p + \delta}} .
\]

Thus, by Markov's inequality,
\begin{equation}
\label{eqn:markov}
\textstyle \Pr_{\vec{x}} \left[ \Pr_{\vec{\pi}} \left[ \big| \frac{1}{m} \sum_{j = 1}^m f \big( \vec{y}_j(\vec{x}, \vec{\pi}) \big) - p \big| \ge \delta \right] > \delta \right]
\le 2 e^{- \frac{2 m \delta^2}{4p + \delta} - \ln(\delta)} ~.
\end{equation}
%
%

Putting together \eqref{eqn:prodconcentrate_reduction} and \eqref{eqn:markov}, the theorem follows.

\section{Single-item Auctions in the \CRmodel}
\label{sec:no_signal}

\subsection{Finite and Bounded-support Distributions}

In this subsection, consider the case that the support of the distributions is a finite subset $V \subset [0, h]$ for some $h > 0$ and suppose we aim to an additive approximation up to a $\beta$ factor.
We propose an algorithm for constructing an empirical Myerson auction (Algorithm~\ref{alg:empirical_myerson_finite}) and present a meta analysis of its sample complexity upper bound. 
They will serve as important building blocks in the analysis of the other classes of distributions.

Concretely, we will show the following theorem, whose proof follows the standard concentration plus union bounds combo, using the concentration bound for product distributions introduced in Section~\ref{sec:concentration_intro}, whose proof is in Section~\ref{sec:prodconcentrate}.

\begin{theorem}
\label{thm:CR_finite}
For any upper bound of the values $h > 0$, any finite support of the value distributions $V \subset [0, h]$, any upper bound of the optimal revenue $0 < \alpha \le h$, and any additive error bound $0 < \beta \le \alpha$, Algorithm~\ref{alg:empirical_myerson_finite} takes $m$ samples and learns an empirical Myerson auction that gets an expected revenue of at least $\opt{\vec{D}} - \beta$ with probability at least $1 - \gamma$ if $m$ is at least: 
\[
    \Theta \big( \alpha h \beta^{-2} \big( n |V| \log (n |V|) + \log \gamma^{-1} \big) \big) ~.
\]
\end{theorem}

\begin{remark}
Importantly, all the parameters $\alpha$, $\beta$, $\gamma$, and $h$ are used only in the analysis. 
The algorithm does not need to know them in advance. 
\end{remark}

 \begin{proof}
 	For any product distribution over $V^n$, Myerson's optimal auction picks the bidder with the largest non-negative ironed virtual value.
 	As any tie-breaking rules give the same expected revenue, we assume wlog that the optimal auction breaks ties deterministically.
 	Then, the auction is characterized by a mapping $\sigma : ([n] \times V) \cup \{ \perp \} \mapsto \{ 1, 2, \dots, n |V|+1 \}$ such that $\sigma(i, v)$ is the rank of bidder $i$ with value $v$ among all possible bidder-value pairs, and $\sigma(\perp)$ is the rank of having zero ironed virtual value. 
 	Given any value profile $(v_1, v_2, \dots, v_n) \in V^n$, Myerson's optimal auction does not allocate the item to anyone if $\min_{i \in [n]} \sigma(i, v_i) > \sigma(\perp)$, and otherwise allocates the item to bidder $i^* = \argmin_{i \in [n]} \sigma(i, v_i)$.
 	Hence, there are at most $(n|V|+1)^{n|V|+1}$ different Myerson's auctions for different product distributions over $V^n$.
 	Let $M_\sigma$ denote the auction that corresponds to such a mapping $\sigma$.
 	
    For any given $\sigma$, we apply Theorem~\ref{thm:prodconcentrate} with $f(\vec{v})$ being the revenue of $M_\sigma$ when the values are $\vec{v}$, dividing it by $h$ so that $f(\vec{v}) \in [0, 1]$, with $\delta = \beta / 4h$, and with $m$ at least:
    \[
        \Theta \big( \alpha h \beta^{-2} \big( n |V| \log (n |V|) + \log \gamma^{-1} \big) \big)
    \]
    samples 
    (for a sufficiently large constant inside the big-$\Theta$ notation).  
	Note that the expectation $p = \E_{\mathbf{x} \in \mathbf{D}} [ f(\mathbf{x}) ]$ is the expected revenue of a mechanism, which is at most $\opt{\vec{D}}/h$ and, thus, $p \le \alpha/h$ by the definition of $\alpha$.
	We get the following inequalities (where the probability is over the randomness of the $m$ samples):
 	\begin{align*}
        \Pr \left[ \big| \rev{M_\sigma}{\vec{E}} - \rev{M_\sigma}{\vec{D}} \big| \ge \tfrac{\beta}{2} \right] 
        &
        \le 2 e^{- \frac{2 m \delta^2}{4p + \delta} - \ln(\delta)} \\
        & 
        \le 2 e^{- \frac{m \beta^2}{40 \alpha h} + \ln(4h/\beta)} \\[1ex]
        &
        \le \gamma (n|V|+1)^{-n|V|-1} 
        ~.
 	\end{align*}

 	%
 	By union bound, we have that with probability at least $1 - \gamma$, for any $\sigma$ it holds that:
 	\[
        \big| \rev{M_\sigma}{\vec{E}} - \rev{M_\sigma}{{\vec{D}}} \big| < \tfrac{\beta}{2} 
        ~.
 	\]
 	%
 	
    Since Algorithm \ref{alg:empirical_myerson_finite} returns the optimal auction w.r.t.\ $\vec{E}$, it gets at least $\opt{\bar{\vec{D}}} - \beta$ revenue in expectation on $\Dbold$.
\end{proof}

\begin{algorithm}[t]
	\begin{algorithmic}[1]
		%
        \STATE \textbf{Input:~} $m$ i.i.d.\ samples $\vec{v}_1, \dots, \vec{v}_{m} \sim \vec{D}$.
		\STATE For every $i \in [n]$, let $E_i$ denote the empirical distribution of buyer $i$, i.e., the uniform distribution over $v_{1i}, v_{2i}, \dots, v_{m i}$.
		\STATE \textbf{Output:~} Myerson's auction w.r.t.\ $\mathbf{E} = E_1 \times E_2 \times \dots \times E_n$.\\
		(If a bidder $i$ submits a bid unseen in the samples, round it down to the closest sample.)
	\end{algorithmic}
	\caption{Empirical Myerson (Finite-support)}
	\label{alg:empirical_myerson_finite}
\end{algorithm}

\subsection{Bounded-support Distributions: Additive Approximation}

In this subsection, consider the case that the distributions have a bounded support $[0, h]$ for some $h > 0$.
We seek to learn a mechanism that is optimal up to an $\epsilon$ additive loss in the expected revenue.
We show the following theorem.

\begin{theorem}
\label{thm:CR_bounded_additive}
Suppose the value distributions $\Dbold$ have bounded supports in $[0, h]$. 
Then, for any $\epsilon > 0$, there is an algorithm that takes $m$ samples and learns a mechanism with revenue at least $\opt{\vec{D}} - \epsilon$ with probability at least $1 - \gamma$ if $m$ is at least:
\[
    \Theta \big(h^2 \epsilon^{-2} ( n h \epsilon^{-1} \log(n h \epsilon^{-1}) + \log \gamma^{-1}) \big) 
    = 
    \tilde{\Theta} \big(n h^3 \epsilon^{-3} \big) 
    ~.
\]
\end{theorem}

We will prove the theorem by reducing it to the finite-support case.
To do so, we first introduce a technical lemma showing that a standard discretization of the value space decreases the optimal revenue by at most $\epsilon$ additively.

\begin{lemma}[Additive discretization of the value space]
 	\label{lem:CR_bounded_additive}
 	Given any product value distribution $\Dbold$, and any $\delta > 0$, let $\Dboldhat$ be the distribution obtained by rounding the values from $\Dbold$ to the closest multiple of $\delta$ from below.
 	Then, we have:
 	\[
 	\opt{\Dboldhat} \ge \opt{\Dbold} - \delta ~.
 	\]
    The lemma holds for any single parameter problem, replacing $\delta$ with $k \delta$ on the right-hand-side.
\end{lemma}

Assume for simplicity of exposition that $h$ itself is a multiple of $\delta$.

\begin{proof}
%
We will prove the lemma by explicitly constructing a mechanism that gets expected revenue at least $\opt{\Dbold} - \delta$ on distribution $\Dboldhat$.
%
%
%
Let $M$ be the optimal mechanism w.r.t.\ $\Dbold$.
We will assume wlog that $M$ breaks ties deterministically over bidders with equal ironed virtual values.
Consider a mechanism $\Mhat$ for distribution $\Dboldhat$ with an allocation rule that proceeds as follows: 
\begin{enumerate}
\item Suppose the reported value profile is $\vec{\hat{v}} = (\hat{v}_1, \hat{v}_2, \dots, \hat{v}_n) \in \{0, \delta, 2\delta, \dots, h \}^n$ from $\Dboldhat$.
\item For any bidder $i \in [n]$, sample $v_i$ independently from $D_i$ conditioned on $\hat{v}_i \le v_i < \hat{v}_i + \delta$.
%
%
\item Use the allocation of mechanism $M$ for values profile $\vec{v}$.
\end{enumerate}

Suppose we fix the randomness used in step 2 of the above allocation rule.
Then, note that $M$ is a $0$-$1$ step function of every bidder $i$'s value fixing the other bidders' values.
Further, step $2$ maps values in the support of $\Dboldhat$ to values in the support of $\Dbold$ monotonically.
Hence, the above allocation rule is a $0$-$1$ step function of every bidder $i$'s value fixing the other bidders' values.
Therefore, we can make it a DSIC mechanism by using a payment rule that charges the winner the threshold value above which she is the winner.
Since it is DSIC for every realization of its random bits, the overall mechanism is DSIC (in fact, uniformly truthful).

Next, we analyze the expected revenue of $\Mhat$.
We will do so by coupling two random variables, the revenue of $M$ over a valuation profile sampled $\vec{v}$ from $\Dbold$, and the revenue of $\Mhat$ over a valuation profile $\vec{\hat{v}}$ sampled from $\Dboldhat$.
In particular, we will couple the value profiles $\vec{v}$ and $\vec{\hat{v}}$, such that $\vec{v}$ is the output of step 2 of the allocation rule of $\Mhat$.

For such a pair of value profiles, by the definition of $\Mhat$, the allocation of $M$ on $\vec{v}$ is the same as that of $\Mhat$ on $\vec{\hat{v}}$.
Suppose bidder $i^*$ is the winner.
Further, the payment of each mechanism is equal to the threshold value of  above which bidder $i^*$ is a winner.
By definition, the threshold value in $\Mhat$ is obtained by rounding the threshold value in $M$ down to the closet multiple of $\delta$.
Therefore, for any bidder who get an item, her payments in both mechanisms differ by no more than $\delta$.
So the lemma follows by that there are at most $k$ winners in any feasible allocation.
\end{proof}

By rounding the values to their closest multiples of $\delta = O(\epsilon)$ from below, we reduce the problem to the finite-support case with support $V = \{0, \delta, 2\delta, \dots, h \}$, which has size $O(h/\epsilon)$.
As a result, we prove Theorem~\ref{thm:CR_bounded_additive} with Algorithm~\ref{alg:empirical_myerson_bounded_additive}.

\begin{algorithm}[t]
 	\begin{algorithmic}[1]
 		\STATE \textbf{Parameters:~} $h > 0$ (upper bound of values); $\epsilon \in (0, h]$ (additive revenue loss);\\ \quad $\gamma > 0$ (failure probability).
        \STATE \textbf{Input:~} $m$ i.i.d.\ samples $\vec{v}_1, \dots, \vec{v}_{m} \sim \vec{D}$.
 		\STATE Let $\delta = \frac{\epsilon}{2}$.
 		\STATE Run Algorithm~\ref{alg:empirical_myerson_finite} with $V = \{0, \delta, 2\delta, \dots, h \}$, rounding the samples down to the\\ closest multiple of $\delta$, $\alpha = h$, $\beta = \frac{\epsilon}{2}$ ($h$ and $\gamma$ remain what they are).
 		\STATE \textbf{Output:~} The empirical Myerson auction outputted by Algorithm~\ref{alg:empirical_myerson_finite}, treating any value\\ as the closest multiple of $\delta$ from below.
 	\end{algorithmic}
 	\caption{Empirical Myerson (Bounded-support, Additive Approximation)}
 	\label{alg:empirical_myerson_bounded_additive}
\end{algorithm}

\begin{proof}[Proof of Theorem~\ref{thm:CR_bounded_additive}]
	Let $\Dboldhat$ denote the distribution obtained by rounding sample values from $\Dbold$ down to the closest multiple of $\delta$.
	Since the output mechanism, denoted as $\Mhat$, treats any value as the closest multiple of $\delta$ from below, running it on $\Dbold$ and $\Dboldhat$ gives the same revenue. 
	That is,
	\[
	\rev{\Mhat}{\Dbold} = \rev{\Mhat}{\Dboldhat} ~.
	\]
	
	By Theorem~\ref{thm:CR_finite}, we have that with probability at least $1 - \gamma$:
	\[
	\rev{\Mhat}{\Dboldhat} \ge \opt{\Dboldhat} - \frac{\epsilon}{2} ~.
	\] 
	
	Further, by Lemma~\ref{lem:CR_bounded_additive}, and that $\delta = \frac{\epsilon}{2}$, we get that:
	\[
	\opt{\Dboldhat} \ge \opt{\Dbold} - \delta = \opt{\Dbold} - \frac{\epsilon}{2} ~.
	\] 
	
	Putting together the above equations proves the lemma. 
\end{proof}

\subsection{Bounded-support Distributions: Multiplicative Bound}

In this subsection, we consider the case that the value distributions have bounded supports in $[1, h]$ for some $h > 1$.
We seek to learn a mechanism that is optimal up to a $1-\epsilon$ multiplicative factor in term of the expected revenue.
We show the following theorem.

\begin{theorem}
\label{thm:CR_bounded_multiplicative}
Suppose the value distributions $\vec{D}$ have bounded supports in $[1, h]$. 
Then, for any $0 < \epsilon < 1$, there is an algorithm that takes $m$ samples and learns a mechanism with revenue at least $\big(1 - \epsilon\big) \opt{\vec{D}}$ with probability at least $1 - \gamma$ if $m$ is at least:
\[
    \Theta \big( h \epsilon^{-2} ( n \epsilon^{-1} \log h \log(n \epsilon^{-1} \log h) + \log \gamma^{-1}) \big) 
    =
    \tilde{\Theta} \big( h n \epsilon^{-3} \big)
    ~.
\]
\end{theorem}

We will again prove the theorem by reducing it to the finite support case.
First, we introduce a technical lemma showing that a standard discretization of the value space, tailored for an multiplicative approximation, decreases the optimal revenue by at most a $1 - \epsilon$ multiplicative factor.

\begin{lemma}[Multiplicative discretization of the value space]
 	\label{lem:CR_bounded_multiplicative}
 	Given any product value distribution $\Dbold$, and any $\delta > 0$, let $\Dboldhat$ be the distribution obtained by rounding the values from $\Dbold$ down to the closest power of $1 + \delta$.
 	Then, we have:
 	\[
 	\opt{\Dboldhat} \ge (1 - \delta) \opt{\Dbold} ~.
 	\]
    The lemma holds for any single parameter problem.
\end{lemma}

The proof is almost a verbatim of that of Lemma~\ref{lem:CR_bounded_additive}, replacing the additive discretization with a multiplicative one.
We include it for completeness.
Assume for simplicity of exposition that $h$ itself is a power of $1 + \delta$.

\begin{proof}
	Similar to the proof of Lemma~\ref{lem:CR_bounded_additive}, we will prove the lemma by explicitly constructing a mechanism that gets expected revenue at least $(1 - \delta) \opt{\Dbold}$ on distribution $\Dboldhat$.
	%
	%
	%
    Let $M$ be the optimal mechanism w.r.t.\ $\Dbold$.
	We will assume wlog that $M$ breaks ties deterministically over bidders with equal ironed virtual values.
	Consider a mechanism $\Mhat$ for distribution $\Dboldhat$ with an allocation rule that proceeds as follows: 
	\begin{enumerate}
		\item Suppose the reported value profile is $\vec{\hat{v}} = (\hat{v}_1, \hat{v}_2, \dots, \hat{v}_n) \in \big\{1, 1+\delta, (1+\delta)^2, \dots, h \big\}^n$ from $\Dboldhat$.
		\item For any $i \in [n]$, sample $v_i$ independently from $D_i$ conditioned on $\hat{v}_i \le v_i < (1+\delta) \hat{v}_i$.
		%
		%
		\item Use the allocation of mechanism $M$ for values profile $\vec{v}$.
	\end{enumerate}
	
	Suppose we fix the randomness used in step 2 of the above allocation rule.
	Then, note that $M$ is a $0$-$1$ step function of every bidder $i$'s value fixing the other bidders' values.
	Further, step $2$ maps values in the support of $\Dboldhat$ to values in the support of $\Dbold$ monotonically.
	Hence, the above allocation rule is a $0$-$1$ step function of every bidder $i$'s value fixing the other bidders' values.
	Therefore, we can make it a DSIC mechanism by using a payment rule that charges the winner the threshold value above which she is the winner.
	Since it is DSIC for every realization of its random bits, the overall mechanism is DSIC (in fact, uniformly truthful).
	
	Next, we analyze the expected revenue of $\Mhat$.
    Similar to the corresponding part in the proof of Lemma~\ref{lem:CR_bounded_additive}, we will do so by coupling two random variables, the revenue of $M$ over a valuation profile sampled $\vec{v}$ from $\Dbold$, and the revenue of $\Mhat$ over a valuation profile $\vec{\hat{v}}$ sampled from $\Dboldhat$.
	In particular, we will couple the value profiles $\vec{v}$ and $\vec{\hat{v}}$, such that $\vec{v}$ is the output of step 2 of the allocation rule of $\Mhat$.
	
	For such a pair of value profiles, by the definition of $\Mhat$, the allocation of $M$ on $\vec{v}$ is the same as that of $\Mhat$ on $\vec{\hat{v}}$.
	Suppose bidder $i^*$ is the winner.
	Further, the payment of each mechanism is equal to the threshold value of  above which bidder $i^*$ is a winner.
	By definition, the threshold value in $\Mhat$ is obtained by rounding the threshold value in $M$ down to the closet power of $1+\delta$.
	Therefore, the payments of both mechanisms differ by no more than $(1+\delta)^{-1} \ge 1 - \delta$.
	So the lemma follows.
\end{proof}

By rounding the values to their closest powers of $1 + \delta$ from below for some $\delta = O(\epsilon)$, we reduce the problem to the finite-support case with support $V = \big\{1, 1+\delta, (1+\delta)^2, \dots, h \big\}$, which has size $O \big( \tfrac{h \log h}{\epsilon} \big)$.
As a result, we prove Theorem~\ref{thm:CR_bounded_multiplicative} with Algorithm~\ref{alg:empirical_myerson_bounded_multiplicative}.

\begin{algorithm}[t]
 	\begin{algorithmic}[1]
 		\STATE \textbf{Parameters:~} $h > 0$ (upper bound of values); $\epsilon \in (0, 1)$ (multiplicative loss);\\ \quad $\gamma > 0$ (failure probability).
        \STATE \textbf{Input:~} $m$ i.i.d.\ samples $\vec{v}_1, \dots, \vec{v}_{m} \sim \vec{D}$.
 		\STATE Let $\delta = \frac{\epsilon}{2}$.
 		\STATE Run Algorithm~\ref{alg:empirical_myerson_finite} with $V = \big\{1, 1+\delta, (1+\delta)^2, \dots, h \big\}$, rounding the samples down to the\\ closest powers of $1+\delta$, $\alpha = \opt{\vec{D}}$, and $\beta = \frac{\epsilon}{2} \opt{\vec{D}}$ ($h$ and $\gamma$ remain what they are).
 		\STATE \textbf{Output:~} The empirical Myerson auction outputted by Algorithm~\ref{alg:empirical_myerson_finite}, treating any value\\ as the closest power of $1 + \delta$ from below.
 	\end{algorithmic}
 	\caption{Empirical Myerson (Bounded Support, Multiplicative Approximation)}
 	\label{alg:empirical_myerson_bounded_multiplicative}
\end{algorithm}

\begin{proof}[Proof of Theorem~\ref{thm:CR_bounded_multiplicative}]
Let $\Dboldhat$ denote the distribution obtained by rounding sample values from $\Dbold$ down to the closest power of $1+\delta$.
Since the output mechanism, denoted as $\Mhat$, treats any value as the closest power of $1+\delta$ from below, running it on $\Dbold$ and $\Dboldhat$ gives the same revenue. 
That is,
\[
\rev{\Mhat}{\Dbold} = \rev{\Mhat}{\Dboldhat} ~.
\]

By Theorem~\ref{thm:CR_finite}, we have that with probability at least $1 - \gamma$:
\[
\rev{\Mhat}{\Dboldhat} \ge \opt{\Dboldhat} - \frac{\epsilon}{2} \opt{\Dbold} ~.
\] 

Further, by Lemma~\ref{lem:CR_bounded_multiplicative}, and that $\delta = \frac{\epsilon}{2}$, we get that:
\[
\opt{\Dboldhat} \ge (1 - \delta) \opt{\Dbold} = \opt{\Dbold} - \frac{\epsilon}{2} \opt{\Dbold} ~.
\] 

Putting together the above equations proves the lemma.
\end{proof}

\subsection{Regular Distributions}
\label{sec:cr_regular}

In this subsection, we show an $\tilde{O} \left( n \epsilon^{-4} \right)$ upper bound on the sample complexity of learning the optimal auction in the \crmodel when the value distributions are regular.
Given what we have shown for the bound-support distributions, the new challenge of this case is that the support of the distributions could be unbounded. 
The main technical ingredient that overcomes this challenge is to show that we can truncate the values down to a high enough finite quantity without losing too much revenue. 
This is captured in the following lemma, whose proof is deferred to Subsection~\ref{sec:tailbound}.

\begin{lemma}
	\label{lem:tailbound_regular}
	For any product regular distribution $\Dbold$, any $\frac{1}{4} \ge \delta > 0$, and any $\bar{v} \ge \frac{1}{\delta} \opt{\Dbold}$, let $\bar{D}_1, \bar{D}_2, \dots, \bar{D}_n$ be the distributions obtained by truncating $D_1, D_2, \dots, D_n$ at $\bar{v}$, i.e., a sample $\bar{v}_i$ from $\bar{D}_i$ is obtained by first sampling $v_i$ from $D_i$ and then letting $\bar{v}_i = \min \{ v_i, \bar{v} \}$.
	Then, we have:
	\[
        \opt{\Dboldbar} \ge (1 - 2 \delta) \opt{\Dbold}
        ~.
	\]
    This lemma holds generally in the downward-closed setting.
\end{lemma}

\paragraph{Analysis Assuming $\opt{\Dbold}$ is Given.}
Let us first explain what the analysis looks like assuming that $\opt{\mathbf{D}}$ is known.
Let $\delta = \Theta(\epsilon)$ with a sufficiently small constant inside the big-$\Theta$ notation.
Let $\Dboldbar$ be the distribution obtained by truncating values at $\bar{v} = \frac{1}{\delta} \opt{\Dbold}$ as in Lemma~\ref{lem:tailbound_regular}.
The lemma gives:
\[
\opt{\Dboldbar} \ge (1 - 2 \delta) \opt{\Dbold} ~.
\]

Let $\Dboldtilde$ be the distribution obtained by first sampling from $\Dboldbar$ and then rounding values smaller than $\delta \opt{\Dbold}$ down to $0$.
Since a bidder with value less than $\delta \opt{\Dbold}$ cannot pay more than her value, we have:
\[
\opt{\Dboldtilde} \ge \opt{\Dboldbar} - \delta \opt{\Dbold} ~.
\]

Let $\Dboldhat$ be the distribution obtained by rounding samples from $\Dboldtilde$ down to the closest power of $1+\delta$.
Lemma~\ref{lem:CR_bounded_multiplicative} gives:
\[
\opt{\Dboldhat} \ge (1-\delta) \opt{\Dboldtilde} ~.
\]

By Theorem~\ref{thm:CR_finite}, with $h = \frac{1}{\delta} \opt{\Dbold}$, $|V| = \tilde{O}(\frac{1}{\delta})$, $\alpha = \opt{\Dbold}$, and $\beta = \delta \opt{\Dbold}$, running Algorithm~\ref{alg:empirical_myerson_finite} with $m \ge \tilde{O} \big( n \delta^{-4} \big)$ samples from $\Dboldhat$ returns a mechanism $\Mhat$ that with high probability gets an expected revenue at least:
\[
\rev{\Mhat}{\Dboldhat} \ge \opt{\Dboldhat} - \delta \opt{\Dbold} ~.
\]

Finally, note that if we run $\Mhat$ on $\Dbold$ truncating and rounding values to get a value in the support of $\Dboldhat$, the expected revenue is the same as we run $\Mhat$ on $\Dboldhat$ by the definition of $\Dboldhat$.
That is,
\[
\rev{\Mhat}{\Dbold} = \rev{\Mhat}{\Dboldhat} ~.
\]

Putting together the sequence of equations and $\delta = c \cdot \epsilon$, for a sufficiently small constant $c > 0$, proves the desired sample complexity upper bound for regular value distributions.


\paragraph{Rest of the Subsection.}
It is not difficult to see that it suffices to have an estimate of $\opt{\Dbold}$ up to a constant factor in order to instantiate the above analysis of $\tilde{O}(n \epsilon^{-4})$ sample complexity upper bound.
We will proceed with Subsection~\ref{sec:regular_estimate_opt} which explains how to get a constant approximation of $\opt{\Dbold}$ using a small number of samples from $\Dbold$.
Then, Subsection~\ref{sec:regular_algorithm} will formally present the theorem statement, the algorithm, and the corresponding proof for the case of regular distributions.
Finally, we prove the above technical lemma (Lemma~\ref{lem:tailbound_regular}) in Subsection~\ref{sec:tailbound}.

\subsubsection{Estimating the Optimal Revenue}
\label{sec:regular_estimate_opt}

Now we design an algorithm that estimates the optimal revenue w.r.t.\ a regular value distribution $\Dbold$ up to a constant factor using a small number of samples from it.
To achieve this, we employ a bootstrapping approach that makes use of two families of simple mechanisms:
\begin{itemize}
	\item 
	The first one considers having $n$ copies of the item and sell to each bidder separately. 
	\begin{itemize}
		\item[-] Its optimal revenue is an $n$-approximate of $\opt{\Dbold}$ (folklore, also see Lemma~\ref{lem:napproxopt} below)
		\item[-] The optimal revenue w.r.t.\ each bidder can be estimated up to a factor $2$ with probability at least $1 - \gamma$ with $O \big( \log(1/\epsilon \gamma) \big)$ samples (e.g., \citet{HuangMR15}).
	\end{itemize}
	\item The second one is running the VCG mechanism with two copies of each bidder, whose values are independently sampled.\footnote{Another potential approach is to use VCG with monopoly reserves, which also gives similar guarantees.}
	\begin{itemize}
		\item[-] Its optimal revenue is a $2$-approximate of $\opt{\Dbold}$ (\citet{HR09}, Theorem~4.4).
		%
		%
	\end{itemize}
\end{itemize}


We run these simpler mechanisms with $\Theta\big(n \log(1/\gamma) \big)$ samples from $\Dbold$ to find a constant approximation of optimal revenue using Algorithm \ref{alg:approxopt}, whose formal approximation guarantee is given in \prettyref{lem:approxopt} below.
 
\begin{algorithm}
 	\begin{algorithmic}[1]
 		\medskip
		\item[] \textbf{Stage 1:}
		\medskip
		
 		\STATE For each bidder $i \in [n]$, let there be $m_1 = \Theta \big(\log(n/\gamma)\big)$ samples from $D_i$ in sorted order:
 		\[
 		v_{i1} \ge v_{i2} \ge \dots \ge v_{im_1} ~.
 		\]
 		
 		\STATE Compute a $4$-approximation of the optimal revenue when $i$ is the only bidder, denoted as $\sr_i$, as follows (e.g., \cite{HuangMR15}):
 		\[
 		\textstyle
 		\sr_i = \max_{\frac{m_1}{2} \le j \le m_1} \frac{j}{m_1} \cdot v_{ij}
 		\]
 		
 		\STATE Let $\sr = 2 \sum_{i = 1}^{n} \sr_i$.
 		
 		\bigskip
 		
 		\item[] \textbf{Stage 2:}
 		\medskip
 		
 		\STATE Run the VCG mechanism with two copies of each bidder for $m_2 = \Theta\big(n \log(1/\gamma) \big)$ times, truncating values at larger than $4 \sr$ down to $4 \sr$.
 		\STATE Output the average revenue of the previous step, denoted as $\approxopt$.
 		\medskip
 	\end{algorithmic}
 	%
 	\caption{Computing an $O(1)$ approximation of the optimal revenue}
 	\label{alg:approxopt}
\end{algorithm}
 
We will first prove the following straightforward lemma that bounds the approximation guarantee from the first stage of the above algorithm.

\begin{lemma}
	\label{lem:napproxopt}
    Suppose $\sr$ is computed via Algorithm~\ref{alg:approxopt}.
    Then, with probability $1 - \gamma/2$, we have:
	\[
	\opt{\mathbf{D}} \le \sr \le 4n \cdot \opt{\mathbf{D}}
	\]
    This lemma holds generally in the downward-closed setting.
\end{lemma}

\begin{proof}
	First of all, step 2 is estimating the expected revenue of the $\frac{1}{2}$-guarded empirical pricing with $\Theta \big(\log(n/\gamma) \big)$, which gives a $4$-approximation $\sr_i$ of the optimal revenue of selling only to a single bidder $i$ for each $i \in [n]$ with probability at least $1 - \frac{\gamma}{2n}$ (e.g., \citet{HuangMR15}). 
	That is, we have:
	\[
	\frac{1}{2} \opt{D_i} \le \sr_i \le 2 \opt{D_i} ~.
	\]
	
	Then, by the union bound, the above holds for all $i$ with probability at least $1 - \frac{\gamma}{2}$.
	
    Note that $\opt{D_i} \le \opt{\Dbold}$ since ignoring all bidders other than $i$ and posting the monopoly price w.r.t.\ $i$ is a feasible mechanism for the case with $n$ bidders.
	Hence, we have:
%
	\[
	\sr = 2 \sum_{i=1}^n \sr_i \le 4 \sum_{i=1}^n \opt{D_i} \le 4n \cdot \opt{\Dbold} ~.
	\]
	
	On the other hand, $\sum_{i=1}^n \opt{D_i}$ equals the optimal revenue when the auctioneer has $n$ copies of the item, which is weakly highly than $\opt{\Dbold}$ when she has only one copy.
	Therefore, we have:
	\[
	\sr = 2 \sum_{i=1}^n \sr_i \ge \sum_{i=1}^n \opt{D_i} \ge \opt{\Dbold} ~.
	\]
\end{proof}
 
Now we are ready to analyze the approximation guarantee of the final output $\approxopt$ given by Algorithm~\ref{alg:approxopt}, and show that it is a constant approximation of $\opt{\Dbold}$ with high probability. 

\begin{lemma}
	\label{lem:approxopt}
    Suppose $\approxopt$ is computed via Algorithm~\ref{alg:approxopt}.
    Then, with probability $1 - \gamma$, we have:
	\[
	\frac{1}{8} \opt{\mathbf{D}} \le \approxopt \le 4 \opt{\mathbf{D}} ~.
	\]
    This lemma holds up to the matroid setting.
\end{lemma}

\begin{proof}
We will assume that the conclusion of Lemma~\ref{lem:napproxopt} holds, which happens with probability at least $1 - \frac{\gamma}{2}$.
Next, we show that under this assumption, $\approxopt$ satisfies the claimed approximation with probability at least $1 - \frac{\gamma}{2}$. 
The lemma then follows by the union bound.

Let $\Dboldbar$ be the distribution obtained by first sampling from $\Dbold$ and then truncating values higher than $4 \sr$ down to $4 \sr$ as in step 4 of the algorithm.
By Lemma~\ref{lem:napproxopt} we know that $\sr \ge \opt{\Dbold}$.
Then, by Lemma~\ref{lem:tailbound_regular} (with $\delta = \frac{1}{4}$), truncating values at $4 \sr \ge 4 \opt{\Dbold}$ decreases the revenue by at most a factor of $2$.
That is, we have:
\[
\opt{\Dboldbar} \ge \frac{1}{2} \opt{\Dbold} ~.
\]

Further, the truncated distributions are also regular due to the followings.
For any $i \in [n]$, the virtual value of $4\sr$ in the truncated distribution $\bar{D}_i$ is now $4\sr$, which is weakly higher than the virtual value in original distribution $D_i$, the virtual value $\phi_i(v) = v - \frac{1-F_i(v)}{f_i(v)}$ of any value $v < 4\sr$ remains unchanged by the definition of virtual values (e.g., for continuous value distributions, neither $F_i(v)$ nor $f_i(v)$ is affected by the truncation).
Putting together we get that $\phi_i$ is still monotonically non-decreasing for all $i \in [n]$.

As a result, the expected revenue of running the VCG mechanism with two copies of each bidder $i$ with values from $\Dboldbar$, denoted as $\overline{\textsc{VCG}}$, is a $2$-approximation of the optimal revenue of $\Dboldbar$ (\citet{HR09}, Theorem~4.4).
That is, we have:
\begin{equation}
\label{eqn:approxopt_1}
\overline{\textsc{VCG}} \ge \frac{1}{2} \opt{\Dboldbar} \ge \frac{1}{4} \opt{\Dbold} ~,
\end{equation}

On the other hand, having two copies of each bidder at best doubles the optimal revenue.
Hence, we have that:
\begin{equation}
\label{eqn:approxopt_2}
\overline{\textsc{VCG}} \le 2 \opt{\Dboldbar} \le 2 \opt{\Dbold} ~.
\end{equation}

Finally, the maximum value in the support of $\Dboldbar$ is %
\begin{align*}
4 \sr & \le 16n \cdot \opt{\Dbold} && \text{(Lemma~\ref{lem:napproxopt})} \\
& \le 64 n \cdot \overline{\textsc{VCG}} ~, && \text{(Eqn.~\eqref{eqn:approxopt_1})}
\end{align*}
which also upper bounds the revenue from each sample run in step 4 of the algorithm.
In other words, the maximum revenue from a single sample run is at most $O(n)$ times larger than the expected revenue of the VCG mechanism with two copies of each bidder.
Therefore, with $m_2 = \Theta \big(n \log(1/\gamma) \big)$ sample runs, Bernstein inequality shows that with probability at least $1 - \gamma$, we have:
\[
\frac{1}{2} \cdot \overline{\textsc{VCG}} \le \approxopt \le 2 \cdot \overline{\textsc{VCG}} ~.
\]

Putting together with Eqn.~\eqref{eqn:approxopt_1} and Eqn.~\eqref{eqn:approxopt_2} proves the lemma.
%
%
\end{proof}

\subsubsection{Sample Complexity Upper Bound}
\label{sec:regular_algorithm}

The final algorithm is essentially what we have sketched at the beginning of Section~\ref{sec:cr_regular}, replacing the precise value of $\opt{\Dbold}$ with the constant approximation $\approxopt$ obtained through Algorithm~\ref{alg:approxopt} in the previous subsection.
This is formalized as Algorithm~\ref{alg:empirical_myerson_regular}.

\begin{algorithm}[t]
 	\begin{algorithmic}[1]
	  \STATE \textbf{Parameters:~} $\epsilon \in (0, 1)$ (multiplicative revenue loss); $\gamma > 0$ (failure probability).
      \STATE \textbf{Input:~} $m$ i.i.d.\ samples $\vec{v}_1, \dots, \vec{v}_{m} \sim \vec{D}$.
 	  \STATE Compute $\approxopt$ using Algorithm \ref{alg:approxopt} with $\Theta \big( n \log(1/\gamma) \big)$ samples and failure probability $\frac{\gamma}{2}$.
 	  \STATE Let $\delta = \frac{\epsilon}{32}$.
	  \STATE Let $V = \big\{ 0 \big\} \cup \big\{ \delta \approxopt, (1 + \delta) \delta \approxopt, (1 + \delta)^2 \delta \approxopt, \dots, \frac{1}{\delta} \approxopt \big\}$.
 	  \STATE Run Algorithm~\ref{alg:empirical_myerson_finite} with the remaining samples, rounding values down to the ones in $V$.
 	  \STATE \textbf{Output:~} The empirical Myerson auction outputted by Algorithm~\ref{alg:empirical_myerson_finite}, treating any value\\ as the closest value in $V$ from below.
 	\end{algorithmic}
 	\caption{Empirical Myerson (Regular Distributions)}
 	\label{alg:empirical_myerson_regular}
\end{algorithm}
 
\begin{theorem}
 	\label{thm:CR_regular}
 	Suppose the value distribution $\vec{D}$ is regular. 
 	Then, for any $0 < \epsilon < 1$, Algorithm \ref{alg:empirical_myerson_regular} takes $m$ samples and learns a mechanism with revenue at least $\big(1 - \epsilon\big) \opt{\vec{D}}$ with probability at least $1 - \gamma$ if $m$ is at least:
 	%
 	\[
        \Theta \big( \epsilon^{-3} ( n \epsilon^{-1} \log \epsilon^{-1} \log (n \epsilon^{-1}) + \log \gamma^{-1} ) \big)
        = 
        \tilde{\Theta} \big( n \epsilon^{-4} \big)
        ~.
 	\]
\end{theorem}

\begin{proof}
Let $\Dboldbar$ be the distribution obtained by first sampling from $\Dbold$ and than rounding values larger than  $\bar{v} = \frac{1}{\delta} \approxopt \ge \frac{1}{8\delta} \opt{\vec{D}}$ down to $\bar{v}$.
Lemma~\ref{lem:tailbound_regular} gives:
\[
\opt{\Dboldbar} \ge (1 - 16 \delta) \opt{\Dbold} ~.
\]

Let $\Dboldtilde$ be the distribution obtained by first sampling from $\Dboldbar$ and then rounding values smaller than $\delta \approxopt$ down to $0$.
Since a bidder with value less than $\delta \approxopt$ cannot pay more than her value, we have:
\begin{align*}
\opt{\Dboldtilde} & \ge \opt{\Dboldbar} - \delta \approxopt \\
& \ge \opt{\Dboldbar} - 4 \delta \opt{\Dbold} ~.
\end{align*}

Let $\Dboldhat$ be the distribution obtained by rounding samples from $\Dboldtilde$ down to the closest power of $1+\delta$.
Lemma~\ref{lem:CR_bounded_multiplicative} gives:
\[
\opt{\Dboldhat} \ge (1-\delta) \opt{\Dboldtilde} ~.
\]

By Theorem~\ref{thm:CR_finite}, with $h = \frac{1}{\delta} \approxopt \le \frac{4}{\delta} \opt{\vec{D}}$, $|V| = O \big( \frac{\log(1/\delta)}{\delta} \big)$, $\alpha = \opt{\Dbold}$, and $\beta = \delta \opt{\Dbold}$, and failure probability $\frac{\gamma}{2}$, running Algorithm~\ref{alg:empirical_myerson_finite} with at least:
\[
\Theta \big( \delta^{-3} \big( n \delta^{-1} \log(\tfrac{1}{\delta}) \log (\tfrac{n}{\delta}) + \log(\tfrac{1}{\delta \gamma}) \big) \big)
\]
samples from $\Dboldhat$ returns a mechanism $\Mhat$ that, with probability at least $1 - \frac{\gamma}{2}$, gets an expected revenue at least:
\[
\rev{\Mhat}{\Dboldhat} \ge \opt{\Dboldhat} - \delta \opt{\Dbold} ~.
\]

Finally, note that if we run $\Mhat$ on $\Dbold$ rounding values down to the closest element of $V$, effectively we are running $\Mhat$ on $\Dboldhat$ by the definition of $\Dboldhat$.
Hence, we have:
\[
\rev{\Mhat}{\Dbold} = \rev{\Mhat}{\Dboldhat} ~.
\]

Putting together the sequence of equations and $\delta = \frac{\epsilon}{32}$ proves the desired sample complexity upper bound for regular value distributions.
%
\end{proof}

\subsubsection{Proof of Lemma~\ref{lem:tailbound_regular}}
\label{sec:tailbound}

Recall that we seek to prove the lemma not only for the single-item setting which is the focus of this section, but also more generally for an arbitrary problem in the downward-closed setting.

When there is only one bidder and, thus, we may assume wlog that there is a single item, the lemma is folklore as it follows by concavity of the revenue curves of regular distributions.
What are the extra challenges with multiple bidders?
If the bidders' values are all below the truncation point $\bar{v}$, the virtual values are the same in both $\Dbold$ and $\Dboldbar$ and, thus, the revenue remains the same.
If exactly one of the bidders' value is larger than $\bar{v}$, the analysis is similar to the single bidder case.
It remains to bound the revenue when more than one bidders' values are larger than $\bar{v}$, in which case the competition among the bidders may give an extra edge to the untruncated distribution $\Dbold$ over the truncated one $\Dboldbar$.
The key observation is that, the condition of $\bar{v} \ge \frac{1}{\delta} \opt{\Dbold}$ implies that truncations are so rare such that the difference in revenue from the third case can be charged to the revenue of $\Dboldbar$ from the second case.
Below we formalize this intuition.

\paragraph{Truncations are Rare.}
Let $\bar{q}_i$ denote the quantile of value $\bar{v}$ w.r.t.\ $D_i$ for any bidder $1 \le i \le n$.
We abuse notation and let $\phi_i(q) = \phi_i(v_i(q))$ and $\bar{\phi}_i(q) = \bar{\phi}_i(\bar{v}_i(q))$ denote the virtual value of the corresponding value $v_i(q)$ and $\bar{v}_i(q)$ with quantile $q$ in $D_i$ and $\bar{D}_i$ respectively.
Then, we have:
\[
\bar{\phi}_i(q_i) = 
\begin{cases}
\phi_i(q_i) & \text{if $\bar{q}_i < q_i \le 1$;} \\
\bar{v} & \text{if $0 \le q_i \le \bar{q}_i$.}
\end{cases}
\]
%

First, observe the following straightforward bound of $\bar{q}_i$'s.

\begin{lemma}
	\label{lem:highvalueprob1_app}
	For any bidder $i \in [n]$, we have:
	\[
        \bar{q}_i \le \delta 
        ~.
	\]
\end{lemma}

\begin{proof}
	Note that we could have ignored all bidders other than $i$ and offer a take-it-or-leave-it price of $\bar{v}$ to bidder $i$.
	The resulting revenue would be $\bar{q}_i \bar{v}$.
	This must be at most $\opt{\Dbold}$.
	The lemma then follows by $\bar{v} \ge \frac{1}{\delta} \opt{\Dbold}$.
\end{proof}

We further show the following more refined bounds on $\bar{q}_i$'s. 

\begin{lemma}
\label{lem:highvalueprob2_main}
The probability that at least one buyer has value greater than or equal to $\bar{v}$ is at most $\delta$ and is at least $(1 - \delta) \sum_{i \in [n]} \bar{q}_i$, i.e.,
\[
\delta \ge 1 - \prod_{i = 1}^n (1 - \bar{q}_i) \ge (1 - \delta) \sum_{i = 1}^n \bar{q}_i ~.
\]
\end{lemma}

\begin{proof}
	We first show the first inequality.
	Suppose we offer a take-it-or-leave-it price of $\bar{v}$ to the bidders one by one in lexicographical order and give it to the first bidder who accepts the offer (if such a bidder exists).
	Then, the expected revenue equals the price $\bar{v}$ multiplied by the probability that there is at least one bidder with value at least $\bar{v}$.
	The latter equals
	\[
        1 - \prod_{i=1}^n (1 - \bar{q}_i) ~.
	\] 
	
	Note that this expected revenue must be no larger than the optimal $\opt{\Dbold}$.
	The inequality then follows by $\bar{v} \ge \frac{1}{\delta} \opt{\Dbold}$.
	
	Next, we turn to the second inequality.
	The probability that there is at least one bidder with value at least $\bar{v}$ is lower bounded by the probability that there is exactly one bidder with value at least $\bar{v}$.
	That is, we have:
	\[
	1 - \prod_{i=1}^n (1 - \bar{q}_i) \ge \sum_{i = 1}^n \bar{q}_i \prod_{j \ne i} (1 - \bar{q}_j)
	\]
	
	Further, note that for any bidder $i \in [n]$, we have:
	\[
	\prod_{j \ne i} (1 - \bar{q}_j) \ge \prod_{j = 1}^n (1 - \bar{q}_j) \ge 1 - \delta ~.
	\]
	
	Here, the last inequality is what we have shown in the first part of the proof.
	Therefore, the second inequality in the lemma follows.
\end{proof}

\paragraph{Charging Argument.}
Note that both $\Dbold$ and $\Dboldbar$ are regular.
In both cases, the revenue optimal auction is Myerson's auction, which allocates to the bidder with the largest non-negative virtual value.
Further, the revenue is the virtual welfare.
Define $z^+ = \max \{z, 0\}$ to simplify the notations in the following arguments.
We have:
\[
    \opt{\Dbold} = \int_{[0, 1]^n} \max_{\vec{x} \in \allocset} \sum_{i \in \vec{x}} \phi_i(q_i) d \vec{q} 
    ~,
\]
and
\[
    \opt{\Dboldbar} = \int_{[0, 1]^n} \max_{\vec{x} \in \allocset} \sum_{i \in \vec{x}} \bar{\phi}_i(q_i) d \vec{q} 
    ~.
\]
     
We partition the quantile space $[0, 1]^n$ into two parts, the area with at least one high value, i.e., $[0, 1]^n \setminus \prod_{i} (\bar{q}_i, 1]$, and the area with all low values, i.e., $\prod_{i} (\bar{q}_i, 1]$.
We will account for their contributions to $\opt{\Dboldbar}$ separately.
\[
    \opt{\Dboldbar} = \int_{[0, 1]^n \setminus \prod_{i} (\bar{q}_i, 1]} \max_{\vec{x} \in \allocset} \sum_{i \in \vec{x}} \bar{\phi}_i(q_i) d \vec{q} + \int_{\prod_{i} (\bar{q}_i, 1]} \max_{\vec{x} \in \allocset} \sum_{i \in \vec{x}} \bar{\phi}_i(q_i) d \vec{q}
    ~.
\]

On the other hand, we will upper bound the optimal revenue w.r.t.\ $\Dbold$ by allowing it to choose two feasible sets of winners, one from the bidders whose values are at least $\bar{v}$, and the other from the bidders whose values are less than $\bar{v}$.
Concretely, given a quantile vector $\vec{q}$, let:
\[
    H(\vec{q}) = \big\{ i \in [n] : q_i \le \bar{q}_i \big\}
\]
denote the subset of bidders whose values are at least $\bar{v}$. 
Let:
\[
    L(\vec{q}) = [n] \setminus H(\vec{q})
\]
denote the set of bidders whose values are less than $\bar{v}$.
Then, we have:
\[
    \opt{\Dbold} \le \int_{[0, 1]^n} \max_{\vec{x} \in \allocset} \sum_{i \in \vec{x} \cap H(\vec{q})} \phi_i(q_i) d \vec{q} + \int_{[0, 1]^n} \max_{\vec{x} \in \allocset} \sum_{i \in \vec{x} \cap L(\vec{q})} \phi_i(q_i) d \vec{q}
    ~.
\]



It suffices to show the following inequalities:
%
\begin{equation}
    \label{eqn:tailboundhigh_main}
    \int_{[0, 1]^n \setminus \prod_{i} (\bar{q}_i, 1]} \max_{\vec{x} \in \allocset} \sum_{i \in \vec{x}} \bar{\phi}_i(q_i) d \vec{q}
    \ge 
    (1 - 2 \delta) \int_{[0, 1]^n} \max_{\vec{x} \in \allocset} \sum_{i \in \vec{x} \cap H(\vec{q})} \phi_i(q_i) d \vec{q}
    ~,
\end{equation}
and 
\begin{equation}
    \label{eqn:tailboundlow_main}
    \int_{\prod_{i} (\bar{q}_i, 1]} \max_{\vec{x} \in \allocset} \sum_{i \in \vec{x}} \bar{\phi}_i(q_i) d \vec{q}
    \ge 
    (1 - 2 \delta) \int_{[0, 1]^n} \max_{\vec{x} \in \allocset} \sum_{i \in \vec{x} \cap L(\vec{q})} \phi_i(q_i) d \vec{q}
    ~.
\end{equation}

%

\paragraph{Proof of Eqn.~\eqref{eqn:tailboundhigh_main}:}
Note that for any bidder $i$, whenever its quantile satisfies that $q_i < \bar{q}_i$, we have $\bar{\phi}_i(q_i) = \bar{v}$.
Thus, for any quantile profile $\vec{q} \in [0, 1]^n \setminus \prod_i (\bar{q}_i, 1]$, we have:
\[
    \max_{\vec{x} \in \allocset} \sum_{i \in \vec{x}} \bar{\phi}_i(q_i) d \vec{q} \ge \bar{v}
    ~.
\]

Hence, the left-hand-side of Eqn.~\eqref{eqn:tailboundhigh_main} is lower bounded by:
\begin{align*}
    \int_{[0, 1]^n \setminus \prod_{i} (\bar{q}_i, 1]} \max_{\vec{x} \in \allocset} \sum_{i \in \vec{x}} \bar{\phi}_i(q_i) d \vec{q}
    & 
    \ge \bigg( 1 - \prod_{i \in [n]} (1 - \bar{q}_i) \bigg) \bar{v} \\
& \ge (1 - \delta) \sum_{i \in [n]} \bar{q}_i \bar{v} ~. && \text{(Lemma~\ref{lem:highvalueprob2_main})}
\end{align*}
%

On the other hand, the right-hand-side (omitting the $1 - 2 \delta$ factor) is upper bounded by:
\begin{align*}
    \int_{[0, 1]^n} \max_{\vec{x} \in \allocset} \sum_{i \in \vec{x} \cap H(\vec{q})} \phi_i(q_i) d \vec{q}
    &
    \le \int_{[0, 1]^n} \sum_{i \in H(\vec{q})} \big( \phi_i(q_i) \big)^+ d \vec{q} \\
    &
    = \sum_{i = 1}^n \int_{[0, \bar{q}_i]} \big( \phi_i(q_i) \big)^+ d q_i 
    ~.
\end{align*}


Therefore, to show Eqn.~\eqref{eqn:tailboundhigh_main}, it suffices to show that for any bidder $i \in [n]$:
\[
\bar{q}_i \bar{v} \ge (1 - \delta) \int_{[0, \bar{q}_i]} \big( \phi_i(q_i) \big)^+ d q_i ~.
\]

Let $R_i(q)$ denote the revenue of a single-bidder auction with a reserve price that has quantile $q$ w.r.t.\ $D_i$.
The left-hand-side of the above inequality is the revenue of a single-bidder auction with reserve price $\bar{v}$ w.r.t.\ $D_i$, i.e., $R_i(\bar{q}_i)$.
The right-hand-side is the maximum revenue of a single-bidder auction with a reserve price at least $\bar{v}$, i.e., $\max_{q \in [0, \bar{q}_i]} R_i(q)$.
Note that $D_i$ is regular, which implies that $R_i(q)$ is a concave function over $[0, 1]$ and $R_i(1) \ge 0$.
Thus, we have that $R_i(\bar{q}_i) \ge (1 - \bar{q}_i) \max_{q \in [0, \bar{q}_i]} R_i(q)$.
The inequality then follows by $\bar{q}_i \le \delta$ (Lemma \ref{lem:highvalueprob1_app}).

\paragraph{Proof of Eqn.~\eqref{eqn:tailboundlow_main}:}
Note that for any $\vec{q} \in \prod_{j \in [n]} (\bar{q}_j, 1]$, we have $\bar{\phi}_i(q_i) = \phi_i(q_i)$ for all $i \in [n]$, and $L(\vec{q}) = [n]$.
So we can rewrite the left-hand-side of Eqn.~\eqref{eqn:tailboundlow_main} as follows:
\begin{align*}
    \int_{\prod_j (\bar{q}_j, 1]} \max_{\vec{x} \in \allocset} \sum_{i \in \vec{x}} \bar{\phi}_i(q_i) d \vec{q}
    &
    = \int_{\prod_j (\bar{q}_j, 1]} \max_{\vec{x} \in \allocset} \sum_{i \in \vec{x}} \phi_i(q_i) d \vec{q} \\
    &
    = \int_{\prod_j (\bar{q}_j, 1]} \max_{\vec{x} \in \allocset} \sum_{i \in \vec{x} \cap L(\vec{q})} \phi_i(q_i) d \vec{q}
    ~.
\end{align*}

Comparing the above integration with that in the right-hand-side of Eqn.~\eqref{eqn:tailboundlow_main}, the only difference lies in the domain.
The formal is over $\prod_{j=1}^n (\bar{q}_j, 1]$ while the latter is over the entire quantile space $[0, 1]^n$.
It remains to bound their difference by a factor of $1 - 2\delta$.

We will do so with a hybrid argument.
Let $Q_0 = [0, 1]^n, Q_1 = (\bar{q}_1, 1] \times [0, 1]^{n-1}, \dots, Q_n = \prod_{j = 1}^n (\bar{q}_j, 1]$.
We will show that for any $j \in [n]$:
\begin{equation}
    \label{eqn:tailbound_hybrid}
    \int_{Q_j} \max_{\vec{x} \in \allocset} \sum_{i \in \vec{x} \cap L(\vec{q})} \phi_i(q_i) d \vec{q} 
    \ge
    (1 - \bar{q}_j) \int_{Q_{j-1}} \max_{\vec{x} \in \allocset} \sum_{i \in \vec{x} \cap L(\vec{q})} \phi_i(q_i) d \vec{q} 
    ~.
\end{equation}

For any $\vec{q}_{-j} \in [0, 1]^{n-1}$, any $q_i < \bar{q}_i \le \hat{q}_i$, consider $\vec{q} = (q_i, \vec{q}_{-i})$ and $\vec{\hat{q}} = (\hat{q}_i, \vec{q}_{-i})$. 
We have that $L(\vec{q}) \subset L(\vec{\hat{q}})$, and for any $j \in L(\vec{q})$, $q_j = \hat{q}_j$.
Thus, we get the following inequality:
\[
\max_{i \in L(\vec{q})} \phi_i(q_i) \le \max_{i \in L(\vec{\hat{q}})} \phi_i(\hat{q}_i) ~,
\]
which implies Eqn.~\eqref{eqn:tailbound_hybrid}.

Then, putting together Eqn.~\eqref{eqn:tailbound_hybrid} for $j = 1, 2, \dots, n$, we get that:
\[
    \int_{\prod_j (\bar{q}_j, 1]} \max_{\vec{x} \in \allocset} \sum_{i \in \vec{x} \cap L(\vec{q})} \phi_i(q_i) d \vec{q} 
    \ge 
    \prod_{j = 1}^n (1 - \bar{q}_j) \int_{[0, 1]^n} \max_{\vec{x} \in \allocset} \sum_{i \in \vec{x} \cap L(\vec{q})} \phi_i(q_i) d \vec{q} 
    ~.
\]

Eqn.~\eqref{eqn:tailboundlow_main} now follows by Lemma~\ref{lem:highvalueprob2_main}.


\subsection{MHR Distributions}

In this subsection, we show an $\tilde{O} \left( n \epsilon^{-3} \right)$ upper bound on the sample complexity of learning a mechanism that is a $1 - \epsilon$-approximation when the value distribution $\Dbold$ is MHR. 
The algorithm is almost identical to that for regular distributions, except that we have a better tail bound due to the MHR assumption.
In particular, we will make use of the following extreme value theorem by~\citet{CaiD15extreme}.

\begin{lemma}[\citet{CaiD15extreme}, Theorem~18]
	\label{lem:tailbound_mhr}
	For any MHR value distribution $\Dbold$, any $0 < \delta < \frac{1}{4}$, any $\bar{v} \ge C \cdot \log(\frac{1}{\delta}) \opt{\mathbf{D}}$ for a sufficiently large constant $C$, 
	let $\Dboldbar$ be the distribution obtained by first sampling from $\Dbold$ and then rounding values larger than $\bar{v}$ down to $\bar{v}$.
	Then, we have that:
	\[
	\opt{\Dboldbar} \ge \big(1 - \delta \big) \opt{\Dbold}
	\]
\end{lemma}

We present the formal definition of the algorithm in Algorithm~\ref{alg:empirical_myerson_mhr}.
Both the algorithm and its analysis are almost verbatim to their counterparts for regular value distributions, differing only in the upper bound above which we truncate the values, and the corresponding parameters when we apply the base-case theorem (Theorem~\ref{thm:CR_finite}).

\begin{algorithm}[t]
 	\begin{algorithmic}[1]
	  	\STATE \textbf{Parameters:~} $\epsilon \in (0, 1)$ (multiplicative revenue loss); $\gamma > 0$ (failure probability).
	  	
        \STATE \textbf{Input:~} $m$ i.i.d.\ samples $\vec{v}_1, \dots, \vec{v}_{m} \sim \vec{D}$.
		
		\STATE Compute $\approxopt$ using Algorithm \ref{alg:approxopt} with $\Theta \big( n \log(1/\gamma) \big)$ samples and failure probability $\frac{\gamma}{2}$.
		
		\STATE Let $\delta = \frac{\epsilon}{8}$.
		
		\STATE Let $V = \big\{ 0 \big\} \cup \big\{ \delta \approxopt, (1 + \delta) \delta \approxopt, (1 + \delta)^2 \delta \approxopt, \dots, 8C \cdot \log(\frac{1}{\delta}) \approxopt \big\}$.
		
		\STATE Run Algorithm~\ref{alg:empirical_myerson_finite} with the remaining samples, rounding values down to the ones in $V$.
		
		\STATE \textbf{Output:~} The empirical Myerson auction outputted by Algorithm~\ref{alg:empirical_myerson_finite}, treating any value\\ as the closest value in $V$ from below.
 	\end{algorithmic}
 	\caption{Empirical Myerson (MHR Distributions)}
 	\label{alg:empirical_myerson_mhr}
\end{algorithm}

\begin{theorem}
 	\label{thm:CR_mhr}
	Suppose the value distribution $\vec{D}$ is MHR. 
 	Then, for any $0 < \epsilon < 1$, Algorithm \ref{alg:empirical_myerson_mhr} takes $m$ samples and learns a mechanism with revenue at least $\big(1 - \epsilon\big) \opt{\vec{D}}$ with probability at least $1 - \gamma$ if $m$ is at least:
 	\[
 	\Theta \big( \epsilon^{-2} \log(\tfrac{1}{\epsilon}) \big( n \epsilon^{-1} \log(\tfrac{1}{\epsilon}) \log (\tfrac{n}{\epsilon}) + \log(\tfrac{1}{\epsilon \gamma}) \big) \big) ~.
 	\]
\end{theorem}

\begin{proof}
Let $\Dboldbar$ be the distribution obtained by first sampling from $\Dbold$ and than rounding values larger than  $\bar{v} = 8 C \cdot \log(\frac{1}{\delta}) \approxopt \ge C \cdot \log(\frac{1}{\delta}) \opt{\vec{D}}$ down to $\bar{v}$.
Lemma~\ref{lem:tailbound_mhr} gives:
\[
\opt{\Dboldbar} \ge (1 - \delta) \opt{\Dbold} ~.
\]

Let $\Dboldtilde$ be the distribution obtained by first sampling from $\Dboldbar$ and then rounding values smaller than $\delta \approxopt$ down to $0$.
Since a bidder with value less than $\delta \approxopt$ cannot pay more than her value, we have:
\begin{align*}
\opt{\Dboldtilde} & \ge \opt{\Dboldbar} - \delta \approxopt \\
& \ge \opt{\Dboldbar} - 4 \delta \opt{\Dbold} ~.
\end{align*}

Let $\Dboldhat$ be the distribution obtained by rounding samples from $\Dboldtilde$ down to the closest power of $1+\delta$.
Lemma~\ref{lem:CR_bounded_multiplicative} gives:
\[
\opt{\Dboldhat} \ge (1-\delta) \opt{\Dboldtilde} ~.
\]

By Theorem~\ref{thm:CR_finite}, with $h = 2C \cdot \log(\frac{1}{\delta}) \opt{\Dbold} \le 8 C \cdot \log(\frac{1}{\delta}) \opt{\Dbold}$, $|V| = O \big( \frac{\log(1/\delta)}{\delta} \big)$, $\alpha = \opt{\Dbold}$, and $\beta = \delta \opt{\Dbold}$, and failure probability $\frac{\gamma}{2}$, running Algorithm~\ref{alg:empirical_myerson_finite} with at least:
\[
\Theta \big( \delta^{-2} \log(\tfrac{1}{\delta}) \big( n \delta^{-1} \log(\tfrac{1}{\delta}) \log (\tfrac{n}{\delta}) + \log(\tfrac{1}{\delta \gamma}) \big) \big)
\]
samples from $\Dboldhat$ returns a mechanism $\Mhat$ that, with probability at least $1 - \frac{\gamma}{2}$, gets an expected revenue at least:
\[
\rev{\Mhat}{\Dboldhat} \ge \opt{\Dboldhat} - \delta \opt{\Dbold} ~.
\]

Finally, note that if we run $\Mhat$ on $\Dbold$ rounding values down to the closest element of $V$, effectively we are running $\Mhat$ on $\Dboldhat$ by the definition of $\Dboldhat$.
Hence, we have:
\[
\rev{\Mhat}{\Dbold} = \rev{\Mhat}{\Dboldhat} ~.
\]

Putting together the sequence of equations and $\delta = \frac{\epsilon}{8}$ proves the desired sample complexity upper bound for MHR value distributions.
\end{proof}

\section{\Signalsmodel}
\label{sec:signal}

\subsection{Single Bidder Upper Bound	}
\label{sec:signal_ub_single_agent}

In this subsection we present a proof of the $O \big( q(\epsilon)^{-1} \epsilon^{-3} \log (\tfrac{1}{\epsilon q(\epsilon)}) \big)$ upper bound for the single bidder case of the \signalsmodel, i.e., $n=1$.
In the special case of a single bidder, an auction is simply a posted price.
The main challenge is the scarcity of samples with the same signal as the buyer (in fact, there might be none).
The idea is to find an auxiliary distribution such that (1) we have enough samples from it to find a nearly optimal price for it, and (2) it is stochastically dominated by the buyer's prior.
Then, we can lower bound the revenue of the auction with the optimal revenue of the auxiliary distribution.
To this end, we consider a number of signals in the samples that are just below the signal of the bidder, and use the mixture of the corresponding distributions as our auxiliary distribution. 
The price posted is the best price for the empirical distribution of values corresponding to these signals, guarded against using too few of these values. The algorithm is summarized in Algorithm \ref{alg:SingleBuyerSignalUB}.


\begin{algorithm}
\begin{algorithmic}[1]
\STATE \textbf{Samples:~} 
$(\hat{v}_1, \hat{\sigma}_1), (\hat{v}_2, \hat{\sigma}_2), \dots, (\hat{v}_m, \hat{\sigma}_m)$, where $\hat{\sigma}_1 \geq \hat{\sigma}_2 \geq \dots \geq \hat{\sigma}_m$.

\STATE \textbf{Parameter:~}
$\epsilon \in (0, 1)$ (multiplicative revenue loss)

\STATE \textbf{Input:~} 
Signal $\sigma$ (the value $v$ is unobserved); suppose $\hat{\sigma}_k \geq \sigma \geq \hat{\sigma}_{k+1}$.

\medskip

\STATE Let $c = \Theta \big( \epsilon^{-3} \log (\tfrac{m}{\epsilon}) \big)$, and $\ell = \min \left\lbrace c , m - k \right\rbrace$.

\STATE Post {\em the $\epsilon$-guarded empirical reserve price} w.r.t.\ $v_{k+1}, v_{k+2}, \dots, v_{k+\ell}$, defined as:
\[
p_\emp \defeq \argmax_{\hat{v}_{k + j} : \epsilon \ell \le j \le \ell} \hat{v}_{k + j} \cdot \big| \big\{ \hat{v}_{k + i} : \hat{v}_{k + i} \ge \hat{v}_{k + j}, 1 \le i \le \ell \big\} \big|
\]
\end{algorithmic}
\caption{$1 - O(\epsilon)$ approximation for the single bidder case in the \signalsmodel}
\label{alg:SingleBuyerSignalUB}
\end{algorithm}



\begin{theorem}
	\label{thm:signals_single}
    The expected revenue of Algorithm~\ref{alg:SingleBuyerSignalUB} is a $1 - O(\epsilon)$ approximation in the \signalsmodel with $m$ samples from regular and $q$-bounded distributions if $m$ is at least:
	\[
	\Theta \big( q(\epsilon)^{-1} \epsilon^{-3} \log (\tfrac{1}{\epsilon
 q(\epsilon)}) \big) 
	~.
	\] 
\end{theorem}

\subsubsection{The Analysis in an Idealized World}

Before we dive into the technical details of the actual analysis, let us first considered an idealized world assuming the followings:
\begin{enumerate}
	\item $\hat{\sigma}_i = \frac{m-i}{m-1}$; and
	\item the signal is uniformly distributed over $\hat{\sigma}_1, \hat{\sigma}_2, \dots, \hat{\sigma}_m$.
\end{enumerate} 
In other words, we assume that there is exactly one sample for each signal in the support, and the underlying signal distribution is a uniform one over its support.
The sample values $\hat{v}_i$'s are still independent samples from the corresponding distributions.
The actual signal $\sigma$ of the bidder, as well as her value, are independently sampled from the corresponding distributions.
The idealized assumption allows us to focus on the first idea of our analysis, namely, instead of aiming for an expected revenue comparable to $\opt{D^{\sigma}}$, it suffices to lower bound the revenue of the algorithm with the optimal revenue w.r.t.\ to a sample signal that is not too far from $\sigma$ from below

We first introduce the guarded optimal revenue.
Let $\opteps{\epsilon}{D}$ denote the $\epsilon$-guarded optimal expected revenue w.r.t.\ a distribution $D$ (whose cdf is denoted as $F$).
That is, we define:
\[
\opteps{\epsilon}{D} \defeq \max_{p : 1 - F(p) \ge \epsilon} p \big(1 - F(p) \big)
~.
\]

We need two lemmas from previous works.
We present the proof sketches of these lemmas for completeness.

\begin{lemma}
	[E.g., \citet{Dhangwatnotai2014revenue}]
	\label{lem:regular_opt_eps}
	For any regular distribution $D$, we have:
	\[
	\opteps{\epsilon}{D} \ge (1 - \epsilon) \opt{D} 
	~.
	\]
\end{lemma}

\begin{proof}[Proof Sketch]
    If the optimal price w.r.t.\ distribution $D$ has a sale probability, a.k.a., its quantile, of at least $\epsilon$, we have $\opteps{\epsilon}{D} = \opt{D}$.
    Otherwise, suppose the optimal price w.r.t.\ distribution $D$ has a quantile less than $\epsilon$.
	By the assumption that $D$ is regular, we get that the expected revenue as a function of the quantile of the price is a concave function.
	By concavity, we get that the expected revenue of a price with quantile $\epsilon$ is at least a $1 - \epsilon$ approximation.
	So the lemma follows.
\end{proof}

For any given distribution $D$, suppose there are $m$ samples $v_1 \ge v_2 \ge \dots \ge v_m$.
The $\epsilon$-guarded empirical pricing is the optimal price w.r.t.\ a uniform distribution over the $m$ samples values subject to having sale probability at least $\epsilon$ (e.g., \citet{Dhangwatnotai2014revenue}). 
That is, we define it to be:
\[
\argmax_{v_j : j \ge \epsilon m} v_j \cdot j 
~.
\]

In particular, the price $p_\emp$ in Algorithm~\ref{alg:SingleBuyerSignalUB} is the $\epsilon$-guarded empirical price w.r.t.\ the a distribution obtained by mixing $\dsigmahat{k+1}, \dsigmahat{k+1}, \dots, \dsigmahat{k+\ell}$.

\begin{lemma}[E.g., \citet{HuangMR15}]
	\label{lem:delta_guarded}
	For any given distribution $D$, the $\epsilon$-guarded empirical price with $m$ samples is a $1 - \epsilon$ approximation with probability at least $1 - \gamma$ if $m$ is at least:
	\[
	\Theta \big( \epsilon^{-3} \log( \tfrac{1}{\epsilon \gamma} ) \big)
	~.
	\]
\end{lemma}

\begin{proof}[Proof Sketch]
    For any price with quantile at least $\epsilon$, by Bernstein inequality, we get that its quantile w.r.t.\ the empirical distribution is within a $1 \pm \frac{\epsilon}{2}$ factor of its quantile w.r.t.\ the true distribution $D$ with probability at least $1 - \gamma$, when we have $m \ge \Theta \big( \epsilon^{-3} \log( \tfrac{1}{\epsilon \gamma} ) \big)$ samples.
	As a result, the $\epsilon$-guarded empirical price from these samples is a $1 - \epsilon$ approximation.
\end{proof}

Equipped with these lemmas, we first establish the following lemma which lower bounds the expected revenue \emph{conditioned on the realization of signals}.
We stress that this lemma does not use the idealized assumption made at the beginning of this subsection and will be reused in the actual proof of Theorem~\ref{thm:signals_single} in the next subsection.

\begin{lemma}
	\label{lem:pemp}
    For any $1 \le k \le m - c$, conditioned on the realization of sample signals $\hat{\sigma}_1, \hat{\sigma}_2, \dots, \hat{\sigma}_m$, and the actual signal $\sigma$ of the bidder, which is sandwiched between $\hat{\sigma}_k$ and $\hat{\sigma}_{k+1}$, the expected revenue of the empirical price $p_\emp$ on $\dsigmahat{k}$ is at least a $1 - 2\epsilon$ approximation to the optimal revenue w.r.t.\ $\dsigmahat{k+c}$ with probability at least $1 - \frac{\epsilon}{m}$.
	That is, we have:
	\begin{equation}
		\rev{p_\emp}{\dsigma{}} \geq (1 - 2\epsilon) \opt{\dsigmahat{k+c}}
		~.
	\end{equation}
	%
\end{lemma}

\begin{proof}
    By $k \le m - c$, we have $\ell = \min \{ m - k, c \} = c$.
	Let $\widehat{D}$ denote the mixture of distributions $\dsigmahat{k+1}, \dsigmahat{k+2}, \dots, \dsigmahat{k+c}$.
    By the stochastic dominance assumption on the distributions w.r.t.\ different signals, we have:
	\[
	\dsigmahat{k+1} \succeq \widehat{D} \succeq \dsigmahat{k+c} ~.
	\]
	
	The lemma follows from a sequence of inequalities as follows:
	\begin{align*}
	\rev{p_\emp}{\dsigma{}} 
    & \ge \rev{p_\emp}{\widehat{D}} && \text{($\sigma \ge \hat{\sigma}_{k+1}$, and $\dsigma{k+1} \succeq \widehat{D}$)} \\
    & \ge (1 - \epsilon) \opteps{\epsilon}{\widehat{D}} && \text{(Lemma~\ref{lem:delta_guarded})} \\[.5ex]
	& \ge (1 - \epsilon) \opteps{\epsilon}{\dsigmahat{k+c}} && \text{($\widehat{D} \succeq \dsigma{k+c}$)} \\[.5ex]
	& \ge (1 - \epsilon)^2 \opt{\dsigmahat{k+c}} && \text{(Lemma~\ref{lem:regular_opt_eps})} \\[.5ex]
	& \ge (1 - 2\epsilon) \opt{\dsigmahat{k+c}}
	~.
	\end{align*}

	Here, the second inequality holds with probability at least $1 - \frac{\epsilon}{m}$, while the other inequalities holds with certainty.

    There is a caveat. 
    The sample values $v_{k+1}, v_{k+1}, \dots, v_{k+c}$ are not i.i.d.\ samples from $\widehat{D}$.
    Instead, there is exactly one sample from each $\dsigmahat{k+i}$, $1 \le i \le c$, which form the mixture $\widehat{D}$.
    The reason that Lemma~\ref{lem:delta_guarded} still works is because its proof relies on a Bernstein inequality, with the revenue of a give price when the value is $v_{k+i}$, $1 \le i \le c$, as the random variables;
    it only requires (1) boundedness of the random variables and (2) the expectated average of the random variables is equal to the expected revenue of the price w.r.t.\ $\widehat{D}$.
    In particular, the second condition is satisfied by $v_{k+i}$, $1 \le i \le c$, even though they are not i.i.d.\ from $\widehat{D}$.
\end{proof}

By union bound, we may assume that the above lemma holds for all $k$ with probability at least $1 - \epsilon$.
Then, we can lower bound the expected revenue achieved by Algorithm~\ref{alg:SingleBuyerSignalUB} in the idealized world as follows:
\begin{align*}
	 \textsc{Alg} & = \frac{1}{m} \sum_{k = 1}^m \rev{p_\emp}{\dsigmahat{k} ~|~ \sigma = \hat{\sigma}_k} && \text{(idealized assumption)} \\
	 & \ge \frac{1}{m} \sum_{k = 1}^{m-c} \rev{p_\emp}{\dsigmahat{k} ~|~ \sigma = \hat{\sigma}_k} && \text{($\rev{\cdot}{\cdot} \ge 0$)} \\
	 & \ge \frac{1}{m} \sum_{k = 1}^{m-c} (1 - 2\epsilon) \opt{\dsigmahat{k+c}} && \text{(Lemma~\ref{lem:pemp})} \\
	 & = \frac{1}{m} \sum_{k = c+1}^m (1 - 2\epsilon) \opt{\dsigmahat{k}}
	 ~.
\end{align*}

Comparing the final bound (ignoring the $1 - 2\epsilon$ factor) with the optimal expected revenue, i.e., 
\[
\frac{1}{m} \sum_{k = 1}^m \opt{\dsigmahat{k}}
~,
\] 
according to the idealized assumption, the only difference is that the former drops the top $\frac{c}{m}$ fraction of the signals in the calculation.
When $m \ge \Theta \big( q(\epsilon)^{-1} \epsilon^{-3} \log (\tfrac{1}{\epsilon q(\epsilon)}) \big)$, this is at most a $q(\epsilon)$ fraction, which contributes at most an $\epsilon$ fraction of the optimal expected revenue by the assumption of $q$-boundedness.
Putting together, we get the $1 - 3\epsilon$ approximation guarantee with probability at least $1 - \epsilon$.
Hence, in expectation it is a $1 - 4 \epsilon$ approximation. 

\subsubsection{Proof of Theorem~\ref{thm:signals_single}}

Without the idealized assumption of having a set of perfectly representative sample signal as in the previous subsection, we need a more subtle treatment on the randomness of the signals. 
In particular, it is generally not true that $\sigma$ has the same probability of being sandwiched between $\hat{\sigma}_k$ and $\hat{\sigma}_{k+1}$ for all $k$.
As a result, even though Lemma~\ref{lem:pemp} still holds, that is, when the actual signal $\sigma$ is sandwiched between $\hat{\sigma}_k$ and $\hat{\sigma}_{k+1}$, the expected revenue is still lower bounded by a $1 - 2\epsilon$ factor of the optimal revenue of $\hat{\sigma}_{k+c}$, we cannot finish the proof as in the previous subsection with the idealized assumption.

The main idea is to divide the realization of all the sample value-signal pairs, $(\hat{v}_j,\hat{\sigma}_j)$ for 
$j \in [m]$, and the eventual realization of $(v,\sigma)$ of the actual bidder, into 3 stages as follows. 
\begin{enumerate}
    \item Sample $m+1$ signals $\sigma^\prime_1, \sigma^\prime_2, \dots, \sigma^\prime_{m+1}$ i.i.d.\ from  $\fdsigma$; suppose $\sigma^\prime_1 \geq \sigma^\prime_2 \ge \dots \ge \sigma^\prime_{m+1}$.
\item Pick $i \in [m+1]$ uniformly at random and let $\sigma = \sigma^\prime_i$; let $\hat{\sigma}_1, \hat{\sigma}_2, \dots, \hat{\sigma}_m$ be the other $m$ signals.
\item Sample $m+1$ values, $\hat{v}_j \sim \dsigmahat{j}$ and $v\sim\dsigma{}$, independently.
\end{enumerate}

We will lower bound the expected revenue of the chosen price successively over the randomization in the different stages.

First, let us fix the realization of random variables in step 1 and 2, and take expectations over only the randomness in step 3.
Then, we have that $\hat{\sigma}_j = \sigma'_j$ for $1 \le j \le i-1$ and $\hat{\sigma}_j = \sigma'_{j+1}$ for $i \le j \le m$.
Further, the actual signal $\sigma = \sigma'_i$ of the bidder is sandwiched between $\hat{\sigma}_{i-1} = \sigma'_{i-1}$ and $\hat{\sigma}_i = \sigma'_{i+1}$.
By Lemma~\ref{lem:pemp}, for any $1 \le i \le m+1-c$, we get that the following bound holds with probability at least $1 - \frac{\epsilon}{m}$:
\begin{align*}
\rev{p_\emp}{\dsigma{}} & \geq (1 - 2\epsilon) \opt{\dsigmahat{i-1+c}} \\
& = (1 - 2\epsilon) \opt{D^{\sigma'_{i+c}}}
~.
\end{align*}
By union bound, we get that it holds for all $i$ with probability at least $1 - \epsilon$.

Next, consider the randomness in step 2 (together with step 3).
We have the following bound where the expectation is taken over the randomness of $i$, which is uniformly chosen from $[m + 1]$.
\begin{align*}
\E \big[ \rev{p_\emp}{\dsigma{}} ~|~ \sigma'_1, \sigma'_2, \dots, \sigma'_{m+1} \big] & \geq \frac{1}{m+1} \sum_{i = 1}^{m-c} (1 - 2\epsilon) \opt{D^{\sigma'_{i+c}}} \\
& = \frac{1}{m+1} \sum_{i = c+1}^{m+1} (1 - 2\epsilon) \opt{D^{\sigma'_{i}}}
~.
\end{align*}

Finally, we will lower bound the revenue of the above auction over the randomness of all three stages.
Taking expectations over the randomness of the realization $\sigma'_1, \sigma'_2, \dots, \sigma'_{m+1}$ on both sides of the above inequality, we get that the expected revenue achieved by the algorithm is lower bounded by the followings with probability at least $1 - \epsilon$:
\begin{equation}
\label{eqn:signal_single_alg}
\E \big[ \rev{p_\emp}{\dsigma{}} \big] = (1 - 2\epsilon)  \E \bigg[ \frac{1}{m+1} \sum_{i = c+1}^{m+1} \opt{D^{\sigma'_{i}}} \bigg]
~.
\end{equation}

By definition, the optimal expected revenue that we compare to is equal to:
\begin{equation}
\label{eqn:signal_single_opt}\int_0^{1} \opt{D^\sigma} d\fdsigma = \E \bigg[ \frac{1}{m+1} \sum_{i = 1}^{m+1} \opt{D^{\sigma^\prime_{i}}} \bigg]
~,
\end{equation}
where the equality holds by the linearity of expectations.
At this point, if we compare the lower bound of the expected revenue of the algorithm in Eqn.~\eqref{eqn:signal_single_alg} and the optimal in Eqn.~\eqref{eqn:signal_single_opt}, once again the only difference is that the former drops a $\frac{c}{m+1}$ mass of signals in the calculation.
By the assumption that $m \ge \Theta \big( q(\epsilon)^{-1} \epsilon^{-3} \log (\tfrac{1}{\epsilon q(\epsilon)}) \big) = q(\epsilon)^{-1} c$, this is at most a $q(\epsilon)$ fraction which contributes at most an $\epsilon$ portion of the optimal expected revenue by $q$-boundedness.
So we get that the expected revenue of Algorithm~\ref{alg:SingleBuyerSignalUB} is a $1 - 3\epsilon$ approximation with probability at least $1 - \epsilon$.
Hence, it is a $1 - 4 \epsilon$ approximation in expectation.

%
%

\subsection{Multiple Bidders}
\label{sec:signal_ub_multi}

In this subsection, we show how to solve the more general case of multiple bidders in the \signalsmodel, building on techniques we have developed for the single bidder case and those in the \crmodel with multiple bidders and regular value distributions. 
Formally, we will prove the following sample complexity upper bound.

\begin{theorem}[Upper bound for multiple agents]
\label{thm:multibuyeralg}
There exists a mechanism for selling to $n$ agents in the \signalsmodel, that has a sample complexity, against regular distributions which have a $q$-bounded tail in the signal space, of:
\[
O \left( n^2 q(\epsilon)^{-1} \epsilon^{-4} \log^2 ( \tfrac{n}{\epsilon} ) \right)
~.
\]
\end{theorem}

\subsubsection{Algorithm}

\paragraph{High-level ideas.}
We combine the ideas from the single-agent signal case (Section \ref{sec:signal_ub_single_agent}) and the multi-agent no-signal case (Section \ref{sec:no_signal}).
For each agent, consider $\ell = \Theta(\frac{n^2}{\epsilon^4} \log \frac{n}{\epsilon})$ sample signals that are closest to the agent's signal from below, and using the corresponding $\ell$ sample values in place of the actual samples from his true distribution, we run the the empirical Myerson's auction with preconditioning in Section \ref{sec:no_signal}.
Like in the single agent case, we would like to show that the expected revenue of the auction is almost as good as the optimal revenue when each agent's distribution is replaced by the distribution whose signal is the $\ell$-th closest to the agent's signal from below.
Then, we can use the small tail assumption in the signal space to show that the auction gives nearly optimal revenue.

There are two caveats, one in the algorithm and one in the analysis.
First, the implementation of the empirical Myerson's auction in Section \ref{sec:no_signal} uses the VCG mechanism with duplicate agents to estimate the optimal revenue up to a constant factor, exploiting regularity of the distributions. 
The algorithm in the \signalsmodel, however, needs to estimate the optimal revenue when each agent's distribution is replaced by the mixture of the $\ell$ distributions whose signals are closest to the agent's signal from below. 
The mixture of regular distributions is not necessarily regular.
Thus, we turn to a coarser $n$-approximation. 

Further, the analysis for the single-agent case uses a strong notion of revenue monotonicity, which holds 
\emph{regardless of the mechanism being used}.\footnote{This is true in the single agent case since these mechanisms are posted price mechanisms.}
With multiple agents, we need to make do with the  weaker notion of revenue monotonicity in \prettyref{thm:revenueMonotonicity}, 
which only holds for  \emph{the optimal auction for the dominated distribution}.
This is still not directly applicable because we can only find the optimal auction for the \emph{empirical} distribution w.r.t.\ the mixture  of the $\ell$ distributions whose signals are closest to the agent's signal from below, which is not necessarily dominated by the \emph{true} distribution  for the agent's signal, due to the randomness in the sample values.
To this end, we further introduce a coupling argument to show that the optimal auction for the empirical distribution gives essentially the same revenue guarantee. 

\begin{algorithm}
\begin{algorithmic}[1]
\STATE \textbf{Samples:~} 
$(\hat{v}_{1}, \hat{\sigma}_{1}), (\hat{v}_{2}, \hat{\sigma}_{2}), \dots, (\hat{v}_{m}, \hat{\sigma}_{m})$, where $\hat{\sigma}_{1} \geq \hat{\sigma}_{2} \geq \dots \geq \hat{\sigma}_{m}$.

\STATE \textbf{Parameter:~}
$\epsilon$ (multiplicative revenue loss)

\STATE \textbf{Input:~} 
For each bidder $i$, signal $\sigma_i$ (value $v_i$ is unobserved); 
suppose $\hat{\sigma}_{k_i} \geq \sigma_i \geq \hat{\sigma}_{k_i+1}$.

\medskip

%
\STATE Let $c = \Theta(\frac{n^2}{\epsilon^4} \log^2 \frac{n}{\epsilon})$.
\STATE Let the auxiliary distribution for each bidder $i$, denoted as $\widehat{D}_i$, be the mixture of $D^{\hat{\sigma}_{k_i+1}}, D^{\hat{\sigma}_{k_i+2}}, \dots, D^{\hat{\sigma}_{k_i+c}}$. 
Let $\widehat{F}_i$ be the corresponding cdf. 

\medskip

\item[] \textbf{Step 1: Estimating the Optimal Revenue}
    \STATE For any $i \in [n]$, find a $2$-approximation of $\max_{p : \widehat{F}^i(p) \le \frac{1}{2}} p \big( 1 - \widehat{F}^i(p) \big)$, denoted as $\sr_i$:\\
    \[
        \tfrac{1}{2} \sr_i \le \max_{p : \widehat{F}^i (p) \le \frac{1}{2}} p \big( 1 - \widehat{F}^i (p) \big) \le \sr_i
        ~.
    \]
    \label{step:signal_approx_opt}
%
\STATE Let $\approxopt = \sum_{i = 1}^{n} \sr_i$.

\medskip

\item[] \textbf{Step 2: Empirical Myerson with Preconditioning}

\STATE Let $V = \{ \frac{2}{\epsilon} \approxopt, \frac{2 (1 - \epsilon)}{\epsilon} \approxopt, \frac{2 (1 - \epsilon)^2}{\epsilon} \approxopt, \dots, \frac{\epsilon}{n^2} \approxopt, 0 \big\}$.
\STATE For each buyer $i$, round $\hat{v}_{k_i+1}, \dots, \hat{v}_{k_i+c}$ down to the closest value in $V$.
\STATE Let $\bar{E}^i$ denote the uniform distribution over the rounded values.
\STATE Run Myerson auction w.r.t.\ $\bar{\mathbf{E}} = \bar{E}^1 \times \dots \times \bar{E}^n$.
\end{algorithmic}
\caption{$M_{\text{emp}}$, a $1 - O(\epsilon)$ approximately optimal signal auction for multiple agents}
\label{alg:signal_empmyerson}
\end{algorithm}

\paragraph{Algorithm.}
We present it in Algorithm \ref{alg:signal_empmyerson}.
Step~\ref{step:signal_approx_opt} can be implemented, for example, using the approach by \citet{HuangMR15}, with an extra complication that $\hat{v}_{k_i + 1}, \dots, \hat{v}_{k_i + \ell}$ are not i.i.d.\ samples from $\widehat{D}^i$.
Note that this extra complication is the same problem that we have in the single-buyer case (Lemma \ref{lem:pemp}) and, thus, can be handled the same way. 
We will omit the details.

\paragraph{Analysis.}
Similar to the single-buyer case, we divide the realization of all the random variables into 3 stages as follows. 
\begin{enumerate}
\item Sample $m+n$ signals $\sigma^\prime_1, \sigma^\prime_2, \dots, \sigma^\prime_{m+n}$ i.i.d.\ from  $\fdsigma$; suppose $\sigma^\prime_1 \geq \sigma^\prime_2 \ge \dots \ge \sigma^\prime_{m+n}$.
\item Pick a subset of $n$ indices $i_1 \ne \dots \ne i_n \in [m+n]$ uniformly at random and let $\sigma_j = \sigma^\prime_{i_j}$ for all $j \in [n]$; let $\hat{\sigma}_1 \ge \dots \ge \hat{\sigma}_m$ be the other $m$ signals.
(Note that the  $k_i$ in Algorithm \ref{alg:SingleBuyerSignalUB} is fixed by the end of this stage, and is equal to $j_i-1$.)
\item Sample $m+n$ values, $\hat{v}_j \sim D^{\hat{\sigma}_j}$ for $j \in [m]$ and $v_i \sim \dsigma{i}$ for $i \in [n]$, independently.
\end{enumerate}





\noindent
In the following discussion, let $\underline{D}^i = D^{\sigma_{k_i+\ell}}$ with CDF $\underline{F}^i$ and let $\underline{\vec{D}} = \underline{D}^1 \times \dots \times \underline{D}^n$.
For notation convenient, let $\underline{D}^i$ be a point mass at $0$ for $i > m$

\begin{lemma}
\label{lem:signalmulti_approxopt}
With probability at least $1 - O(\epsilon)$, we have:
\[
    \frac{1}{2} \opt{\underline{\vec{D}}} \le \approxopt \le n \opt{\vec{\widehat{D}}}
    ~.
\]
This lemma holds more generally in the matroid setting.
\end{lemma}

\begin{proof}
The first inequality holds because 
\begin{align*}
    \sr_i & \ge \max_{p  : \bar{F}^i (p) \le \frac{1}{2}} p \big( 1 - \bar{F}^i (p) \big) & & \text{(definition of $\sr_i$)} \\[1ex]
    & \ge \max_{p : \underline{F}^i (p) \le \frac{1}{2}} p \big( 1 - \underline{F}^i (p) \big) & & \text{($\bar{F}^i \succeq \underline{F}^i$)} \\
    & \ge \frac{1}{2} \max_{p \ge 0} p \big( 1 - \underline{F}^i (p) \big) & & \text{(regularity of $\underline{D}^i$)}
\end{align*}
and $\sum_{i \in [n]} \max_{p \ge 0} p \big( 1 - \underline{F}^i (p) \big) = \sum_{i \in [n]} \opt{\underline{D}^i} \ge \opt{\underline{\vec{D}}}$. 
The second inequality holds because $\sr_i \le \opt{\vec{\widehat{D}}}$.
\end{proof}

Next, let $\Dboldbar$ be the distribution obtained by rounding the values sampled from $\vec{\widehat{D}}$ down to the closest value in $V$, the discretized value set defined in Algorithm~\ref{alg:signal_empmyerson}.

\begin{lemma}
    \label{lem:multi_signal_opt_relation}
    With probability at least $1 - O(\epsilon)$, we have:
    \[
        \opt{\vec{\bar{D}}} \geq \big( 1 - O(\epsilon) \big) \cdot \opt{\vec{\underline{D}}} - O(\epsilon) \cdot \opt{\vec{D}} \enspace.
    \]
    This lemma holds more generally in the matroid setting.
\end{lemma}

\begin{proof}
    First, $\vec{\widehat{D}}$ stochastically dominates $\vec{\underline{D}}$ and therefore, its discretized version $\Dboldbar$ also dominates the discretized version of $\vec{\underline{D}}$.
    
    It remains the show the inequality by bounding the optimal revenue of the discretized version of $\vec{\underline{D}}$.
    By Lemma~\ref{lem:signalmulti_approxopt} and Lemma~\ref{lem:tailbound_regular}, truncating $\vec{\underline{D}}$ at $\frac{2}{\epsilon} \approxopt \ge \frac{1}{\epsilon} \opt{\underline{\vec{D}}}$ decreases the optimal revenue by at most a $1 - O(\epsilon)$ factor.
    By Lemma \ref{lem:CR_bounded_multiplicative}, rounding the values down to the closest power of $1 - \epsilon$ decreases the optimal revenue by at most another $1 - \epsilon$ factor.
    Finally, by Lemma \ref{lem:signalmulti_approxopt} rounding the values less than $\frac{\epsilon}{n^2} \approxopt \le \frac{\epsilon}{n} \opt{\vec{\widehat{D}}} \le \frac{\epsilon}{n} \opt{\Dbold}$ down to $0$ decreases the revenue by at most $\epsilon \opt{\Dbold}$.
\end{proof}

As a simple corollary of the above lemma and Theorem \ref{thm:prodconcentrate},\footnote{Again, there is an extra complication that the sample values are not i.i.d.\ from the mixture distributions, which can be handled using the approach in Lemma~\ref{lem:pemp}.} we have the following lemma.

\begin{lemma}
\label{lem:signal_multi_emprev}
With probability $1 - O(\epsilon)$, we have
\[
    \opt{\bar{\vec{E}}} \geq \big( 1 - O(\epsilon) \big) \cdot \opt{\vec{\underline{D}}} - O(\epsilon) \cdot \opt{\vec{D}} \enspace.
\]
\end{lemma}



Next, we prove the key lemma that lower bounds the expected revenue of the empirical Myerson auction w.r.t.\ $\bar{\vec{E}}$ if we run it on $\vec{D}$, conditioned on the realization of signals.

\begin{lemma}
With probability $1 - O(\epsilon)$ (over randomness of the sample values), we have
\[
\E \big[ \rev{M_{\emp}}{\vec{D}} | \hat{\sigma}_1, \dots, \hat{\sigma}_m, \sigma_1, \dots, \sigma_n \big] \geq \big( 1 - O(\epsilon) \big) \cdot \opt{\vec{\underline{D}}} - O(\epsilon) \cdot \opt{\vec{D}} \enspace.
\]
\end{lemma}

\begin{proof}
Let us first explain the proof structure.
If $\vec{D}$ stochastically dominates $\bar{\vec{E}}$, then the lemma follows by Lemma \ref{lem:signal_multi_emprev} and revenue monotonicity (Theorem \ref{thm:revenueMonotonicity}).
However, despite the fact that $\vec{D}$ stochastically dominates $\bar{\vec{D}}$, it does not necessarily dominates $\bar{\vec{E}}$ due to randomness in the sample values.
To get around this obstacle, we designed for each buyer $i$ a set of random values coupled with the sample values $\hat{v}_{k_i + j}$, $1 \le j \le \ell$, and the uniform distribution over them $E^i$ such that (1) $\vec{E}$ stochastically dominate $\bar{\vec{E}}$, and (2) the revenue of any optimal mechanism w.r.t.\ some distribution over the set of discretized values when we run it on $\vec{E}$ is a $1 - O(\epsilon)$ approximation of its revenue on $\vec{D}$.
Then, we can show that the revenue of the empirical Myerson auction (w.r.t.\ $\bar{\vec{E}}$) on $\vec{D}$ is at least a $1 - O(\epsilon)$ fraction of its revenue on $\vec{E}$, which is greater than or equal to its revenue on $\bar{\vec{E}}$.

Concretely, for every $1 \le j \le \ell$, let $\hat{q}_j = \big( 1 - F^{\hat{\sigma}_{k_i+j}}(\hat{v}_{k_i + j}) \big)$ be the quantile of $\hat{v}_{k_i + j}$ w.r.t.\ $D^{\hat{\sigma}_{k_i+j}}$; let $\tilde{v}_j$ be the value with quantile $\hat{q}_j$ w.r.t.\ $D^i$.
Note that $\hat{q}_j$ is uniform from $[0, 1]$ over random realization of $\hat{v}_{k_i + j}$ conditioned on $\hat{\sigma}_{k_i + j}$.
So $\tilde{v}_j$, $1 \le j \le \ell$, are i.i.d.\ samples from $D^i$.
Let $E^i$ be the uniform distribution over $\tilde{v}_j$, $1 \le j \le \ell$.

By our construction and that $D^i$ stochastically dominates $D^{\hat{\sigma}_{k_i+j}}$ for any $1 \le j \le \ell$, we get that $\vec{E}$ stochastically dominates $\bar{\vec{E}}$.
Further, by revenue monotonicity (Theorem \ref{thm:revenueMonotonicity}), we have $\rev{M_{\emp}}{\vec{E}} \ge \rev{M_{\emp}}{\bar{\vec{E}}}$. Recall that we can alter $\vec{E}$ such that the optimal auction is unique (Remark~\ref{rem:uniqueMechanism}).
Thus, by Theorem \ref{thm:prodconcentrate}, we have that $\rev{M_{\emp}}{\vec{D}} \ge \big( 1 - O(\epsilon) \big) \rev{M_{\emp}}{\vec{E}}$.
Putting together with Lemma \ref{lem:signal_multi_emprev} proves the lemma.
\end{proof}


As a simple corollary, we have the following lemma.

\begin{lemma}
\label{lem:signal_multi_rev}
With probability $1 - O(\epsilon)$, we have
\[
\E \big[ \rev{M_{\emp}}{\vec{D}} \big] \geq \big( 1 - O(\epsilon) \big) \E \big[ \opt{\vec{\underline{D}}} \big] - O(\epsilon) \E \big[ \opt{\vec{D}} \big] \enspace.
\]
\end{lemma}

Finally, we relate $\E \big[ \opt{\vec{\underline{D}}} \big]$ with $\E \big[ \opt{\vec{D}} \big]$ by proving the following lemma.

\begin{lemma}
\label{lem:signal_multi_rev2}
Over the randomness of all three stages, we have
\[
\E \big[ \opt{\vec{\underline{D}}} \big] \ge ( 1 - \epsilon ) \E \big[ \opt{\vec{D}} \big]
\]
\end{lemma}
\begin{proof}
Given Lemma \ref{lem:signal_multi_rev}, it suffices to relate $\E \big[ \opt{\vec{\underline{D}}} \big]$ with $\E \big[ \opt{\vec{D}} \big]$.
First, we consider the expectation over the randomness of the last two stages.
Conditioned on the realization of signals $\sigma^\prime_1, \dots, \sigma^\prime_{m+n}$, over the randomness of the choice of $i_1, \dots, i_n$ and random realization of the corresponding values, we can write the expectation of the optimal revenue w.r.t.\ $\vec{\underline{D}}$ as follows.
%
%
%
\begin{eqnarray*}
\E \big[ \opt{\underline{\vec{D}}} | \sigma^\prime_1, \dots, \sigma^\prime_{m+n} \big] & = & \frac{1}{{m+n \choose n}} \sum_{\{i_1, \dots, i_n\} \subset [m+n]} \opt{D^{\sigma^\prime_{i_1 + \ell}} \times \dots \times D^{\sigma^\prime_{i_n + \ell}}} \\
& = & \frac{1}{{m+n \choose n}} \sum_{\{i_1, \dots, i_n\} \subset \{ \ell+1, \dots, \ell + m + n\}} \opt{D^{\sigma^\prime_{i_1}} \times \dots \times D^{\sigma^\prime_{i_n}}} \\
& \ge & \frac{1}{{m+n \choose n}} \sum_{\{i_1, \dots, i_n\} \subset \{ \ell+1, \dots, m + n\}} \opt{D^{\sigma^\prime_{i_1}} \times \dots \times D^{\sigma^\prime_{i_n}}}
\end{eqnarray*}

Next, we consider the expectation of over the randomness of the first stage.
Note that $\ell \le q(\epsilon) m + n$.
Taking expectation over random realization of $\sigma^\prime_1, \dots, \sigma^\prime_{m+n}$, the RHS of the above inequality is lower bounded by
\[
\mathbf{E}_{\sigma_1, \dots, \sigma_n \sim \fdsigma} [ \mathds{1}\{\forall i \in [n], \fdsigma(\sigma_i) \leq 1-q(\epsilon)\} \opt{\dsigma {1} \times \dots \times \dsigma {n}}]
\]
By the small tail assumption in signal space (Definition~\ref{def:onebuyeralg3}), this is at least $(1- \epsilon) \E \big[ \opt{\vec{D}} \big]$.
So the lemma follows.
\end{proof}

\begin{proof}[Proof of Theorem \ref{thm:multibuyeralg}]
It follows from Lemma \ref{lem:signal_multi_rev} and \ref{lem:signal_multi_rev2}.
\end{proof}

\subsection{Lower bounds for $n=1$}
\label{app:signal_lb}


\begin{theorem}
The sample complexity of any single agent mechanism against any $q$-bounded distribution in the Signals model is 
$$\Omega\left( \frac{\log\left( 1/q(\epsilon) \right) }{\epsilon^4} \right).$$
\end{theorem}

\begin{proof}

Our instance will have $\frac{\log\frac{1}{\alpha}}{2\log\left(1-18\epsilon\right)} + 1$ conditional distributions. Let $N = \frac{\log\frac{1}{\alpha}}{2\log\left(1-18\epsilon\right)} \approx \frac{\log \alpha}{\epsilon}$. For the $i$-th conditional distribution (that is the distribution conditioned on signal $\sigma_i$) we will have two candidate distributions $D^{\sigma_i}_1$ and $D^{\sigma_i}_2$, and it will be impossible to extract high revenue from that conditional, without seeing at least $\frac{1}{\epsilon^3}$ samples (from \textit{that} conditional), even when having exact knowledge of the rest of the distributions. The expected revenue from each conditional will be approximately the same, thus, in order for a pricing algorithm to achieve high expected revenue overall, it will have to achieve high expected revenue in most of the conditionals.

Our construction works as follows: for each $i \in \left[ 0, N \right]$ let $$c_i = \frac{1}{\left( 1 - 2 \epsilon_0 \right)^{2i}} .$$  $D^{\sigma_i}_1$ is the distribution with c.d.f $$F^{i}_1 \left( x \right) = 1 - \frac{c_i}{x+c_i} $$ and p.d.f $$f^{i}_1 = \frac{c_i}{\left(x+c_i\right)^2} .$$ $D^{\sigma_i}_2$ has c.d.f

\[ F^{i}_2 \left( x \right) =
   \begin{cases}
     1 - \frac{c_i}{x+c_i}       & \quad \text{if } x \in \left[ 0 , \frac{\left( 1 - 2\epsilon_0\right) c_i}{2 \epsilon_0}\right] \\
     1 - \frac{c_i \left(1 - 2\epsilon_0 \right) ^2}{x - c_i \left( 1-2\epsilon_0 \right)}  & \quad \text{if } x > \frac{\left( 1 - 2\epsilon_0\right) c_i}{2 \epsilon_0}\\
   \end{cases}
 \]

and p.d.f.
\[ f^{i}_2 \left( x \right) =
   \begin{cases}
     \frac{c_i}{\left(x+c_i\right)^2}      & \quad \text{if } x \in \left[ 0 , \frac{\left( 1 - 2\epsilon_0\right) c_i}{2 \epsilon_0}\right] \\
     \frac{ c_i \left( 1 - 2\epsilon_0\right)^2}{\left( x - c_i \left( 1-2\epsilon_0 \right) \right)^2}  & \quad \text{if } x > \frac{\left( 1 - 2\epsilon_0\right) c_i}{2 \epsilon_0}\\
   \end{cases}
 \]

where $\epsilon_0 = 9\epsilon$.

The revenue curves as a function of the price are 
\begin{equation*}
\begin{gathered}
  R^i_1 \left( x \right) = \frac{x c_i}{x + c_i} \\
\end{gathered}\hspace{6em}
\begin{gathered}
  R^i_2 \left( x \right) =
  \begin{cases}
     \frac{x c_i}{x+c_i}      & \quad \text{if } x \in \left[ 0 , \frac{\left( 1 - 2\epsilon_0\right) c_i}{2 \epsilon_0}\right] \\
     \frac{x c_i \left(1 - 2\epsilon_0 \right) ^2}{x - c_i \left( 1-2\epsilon_0 \right)}  & \quad \text{if } x > \frac{\left( 1 - 2\epsilon_0\right) c_i}{2 \epsilon_0}\\
   \end{cases} \\
\end{gathered}
\end{equation*}

and as a function of the quantile space
\[ R^i_1 \left( q \right) = c_i \left( 1 - q \right) \]
\[ R^i_2 \left( q \right) =
  \begin{cases}
     c_i \left( \left(1 - 2\epsilon_0\right)^2 + q \left( 1 - 2\epsilon_0 \right) \right)  & \quad \text{if } q \in \left[ 0 , 2 \epsilon_0\right]\\
     c_i \left( 1 - q \right)      & \quad \text{if } q \in \left[ 2 \epsilon_0, 1\right] \\
   \end{cases} \]

Conditioned on seeing signal $\sigma_i$ the Myerson revenue is $c_i$ for $D^{\sigma_i}_1$ and $c_i \left(1-2\epsilon_0\right)$ for $D^{\sigma_i}_2$. Notice that both $D^{\sigma_i}_1$ and $D^{\sigma_i}_2$  first-order stochastically dominate both $D^{\sigma_{i-1}}_1$ and $D^{\sigma_{i-1}}_2$ (and are dominated by the $D^{\sigma_{i+1}}$'s). Also, we have that $c_0 = 1$ and $c_N = \left( 1 - 2 \epsilon_0 \right)^{-2N} = \left( 1 - 18 \epsilon \right)^{- 2 \frac{1}{2} \log_{1 - 18 \epsilon} \frac{1}{\alpha}} = \alpha$, thus the gap between the highest and the lowest revenue is $\alpha$.

In order to make the expected revenue from each conditional approximately the same, we set the probability $p_i$ of seeing signal $\sigma_i$ equal to $\gamma \cdot \frac{1}{c_i}$, where $\gamma$ is the normalizing factor, i.e.  
$$ \gamma = \frac{1}{\sum_{j=0}^{N} \frac{1}{c_j} }$$.

The optimal expected revenue is upper bounded by:
\[ OPT = \sum_{i=0}^{N} p_i \opt{{D^{\sigma_i}}} \leq \sum_{i=0}^{N} c_i p_i  = \gamma (N+1) \]

Dropping $p_N$ in probability mass in signal space results in dropping about $\gamma$ in terms of revenue, i.e. $\frac{1}{N}$ of the total revenue. Dropping $p_N + p_{N-1}$ results in dropping a $\frac{2}{N}$ fraction of the total revenue, and so on; thus, in order to drop an $\epsilon$ fraction of the total revenue, and since $N$ is approximately $\frac{\log \alpha}{\epsilon}$, we need to drop $\sum_{ i = N - log \alpha}^{N} p_i$ in probability. This simplifies to 
\begin{equation*}
\left( \frac{(1-2\epsilon)}{4\epsilon(1-\epsilon)} \right)^2  \cdot \frac{1}{\alpha} \cdot \frac{1}{\alpha^{2 \log(1-2\epsilon)}}.
\end{equation*}

Using this, we can get the following relation between $q(\epsilon)$ and $\alpha$:
\begin{equation}
\label{eq:alpha-to-qe}
\log \alpha \approx \log \left( \frac{1}{q(\epsilon)}\right)
\end{equation}

\begin{claim}
Any pricing algorithm that is $(1-\epsilon)$ approximate for this instance must be $(1- 3 \epsilon)$ approximate for at least half of the signals.
\end{claim}

\begin{proof}
Assume that we have a $(1-\epsilon)$ approximate algorithm, and there exists a set of signals $A$ such that $|A| > \frac{N+1}{2}$ and for every signal $\sigma_i \in A$: $\E \left[ R | \sigma_i \right] < \left( 1- 3 \epsilon \right) \opt{D^{\sigma_i}}$. We will show that the total expected revenue $R$ of the algorithm is less than $\left( 1 -\epsilon \right)$ of $OPT$.

\begin{align*}
R & = \sum_{i=0}^{N} p_i \E \left[ R | \sigma_i \right] = \sum_{i \in A} p_i \E \left[ R | \sigma_i \right] + \sum_{i \notin A} p_i \E \left[ R | \sigma_i \right] \\
& < \sum_{i \in A} \left( 1 - 3\epsilon \right) p_i \opt{D^{\sigma_i}} + \sum_{i \notin A} p_i \opt{D^{\sigma_i}} = OPT - 3 \epsilon \sum_{i \in A} p_i \opt{D^{\sigma_i}} \\
& < OPT - 3 \epsilon |A| \left( 1 - 2 \epsilon_0 \right) \gamma
\end{align*}

Remember that $\gamma \geq \frac{1}{N+1} OPT$, so $3 \epsilon |A| \left( 1 - 2 \epsilon_0 \right) \gamma \geq \frac{3}{2} \left( 1 - 2 \epsilon_0 \right) \epsilon OPT > \epsilon OPT$.

Thus we have:
\[ R < \left( 1 - \epsilon \right) OPT \] a contradiction.

\end{proof}

We will proceed to show that in order to $(1-3 \epsilon)$ approximate the revenue for a given signal $\sigma_i$, any pricing algorithm requires $\Omega\left( \frac{1}{\epsilon^3}\right)$ samples from that signal. Thus, since any $(1- \epsilon)$ approximate pricing algorithm requires $\Omega\left( \frac{1}{\epsilon^3}\right)$ approximate in at least $\frac{N+1}{2} > \frac{\log\frac{1}{\alpha}}{4 \log\left(1-18\epsilon\right)} $ of the signals, the total number of signals required is $\Omega\left( \frac{\log\frac{1}{\alpha}}{2\log\left(1-18\epsilon\right)}\frac{1}{\epsilon^3} \right)$, which is $\Omega\left( \frac{\log\alpha }{\epsilon^4} \right)$, and thus using equation~\ref{eq:alpha-to-qe} we get the desired lower bound for $q$-bounded distributions.

\begin{lemma}
Any pricing algorithm that is $(1-3 \epsilon)$ approximate for both $D^{\sigma_i}_1$ and $D^{\sigma_i}_2$ requires $\Omega\left( \frac{1}{\epsilon^3}\right)$ samples from the distribution conditioned on $\sigma_i$.
\end{lemma}

\begin{proof}
First, notice that even exact knowledge of the distributions conditioned on the rest of the signals does not help here; we're back to a one-distribution problem.

We will follow a similar approach to~\cite{HuangMR15}: we will first bound the KL divergence of $D^{\sigma_i}_1$ and $D^{\sigma_i}_2$, then show that the $(1-9\epsilon) = (1-\epsilon_0)$ optimal price sets are disjoint. We state here, without proofs, a theorem and a lemma from \cite{HuangMR15} that we use:

\begin{theorem}[Theorem 4.2 from \cite{HuangMR15}]\label{thm:samplebound}
If value distributions $D_1$ and $D_2$ have disjoint $(1-3\epsilon')$ approximate price sets, and there is a pricing algorithm that is $(1-\epsilon')$ approximate for both $D_1$ and $D_2$, then the algorithm uses at least $\frac{4}{9} \frac{1}{D_{KL}\left(D_1\Vert D_2\right) + D_{KL}\left(D_2\Vert D_1\right)}$ samples.
\end{theorem}

\begin{lemma}[Lemma 4.4 from \cite{HuangMR15}]\label{lem:KLbound}
If distributions $D_1$ and $D_2$ with densities $f_1$ and $f_2$ satisfy that $\left(1+\epsilon'\right)^{-1} \leq \frac{f_1(\omega)}{f_2(\omega)} \leq (1+\epsilon')$ for every $\omega \in \Omega$, and there is a subset of outcomes $\Omega'$ such that $p_1(\omega) = p_2(\omega)$ for every $\omega \in \Omega'$, then
\[ D_{KL}\left(D_1\Vert D_2\right) + D_{KL}\left(D_2\Vert D_1\right) \leq (\epsilon')^2 \left( 1 - p_1\left( \Omega' \right)\right)\]
\end{lemma}

Notice that:

\[ \frac{f^{i}_1}{f^i_2} =
   \begin{cases}
     1			 & \quad \text{if } x \in \left[ 0 , \frac{\left( 1 - 2\epsilon_0\right) c_i}{2 \epsilon_0}\right] \\
     \frac{c_i}{ c_i \left( 1 - 2\epsilon_0\right)^2} \cdot \frac{\left( x - c_i \left( 1-2\epsilon_0 \right) \right)^2}{\left(x+c_i\right)^2}  & \quad \text{if } x > \frac{\left( 1 - 2\epsilon_0\right) c_i}{2 \epsilon_0}\\
   \end{cases}
 \]

Thus it holds that $\frac{f^{i}_1}{f^i_2} \in \left[ \left( 1-2\epsilon_0 \right)^2, \left(  1-2\epsilon_0 \right)^2 \right]$, and thus $\frac{f^{i}_1}{f^i_2} \in \left[ \left( 1 + \frac{\epsilon_0}{1-2 \epsilon_0} \right)^{-1}, \left( 1 + \frac{\epsilon_0}{1-2 \epsilon_0} \right) \right]$.

$D^{\sigma_i}_1$ and $D^{\sigma_i}_2$ are identical for $x \leq \frac{\left( 1 - 2\epsilon_0\right) c_i}{2 \epsilon_0}$. In quantile space, this corresponds to $q \geq 2 \epsilon_0$. Thus, by applying Lemma~\ref{lem:KLbound} we have that:
\begin{equation}\label{eq:KLbound}
D_{KL}\left(D^{\sigma_i}_1\Vert D^{\sigma_i}_2\right) + D_{KL}\left(D^{\sigma_i}_2\Vert D^{\sigma_i}_1\right) \leq \left(\frac{\epsilon_0}{1-2 \epsilon_0} \right)^2 2 \epsilon_0 = \frac{2 \epsilon_0^3}{\left( 1-2 \epsilon_0 \right)^2}
\end{equation}

Moreover, the $\left(1-\epsilon_0\right)$ price sets are disjoint, by comparing the revenue as a function of the price:

\[ R^i_1(x) \geq  (1-\epsilon_0)c_i  \text{ iff } x \geq \frac{c_i(1-\epsilon_0)}{\epsilon_0} \] and 
\[ R^i_2(x) \geq (1-\epsilon_0)(1-2\epsilon_0)c_i   \text{ iff } x \in \left[ \frac{(1-3\epsilon_0+2\epsilon_0^2)c_i}{3\epsilon_0-2\epsilon_0^2} , \frac{c_i(1-\epsilon_0)(1-2\epsilon_0)}{\epsilon_0} \right] \]

By combining equation~\ref{eq:KLbound} with the above observation, and Theorem~\ref{thm:samplebound} ($\epsilon' = 3\epsilon$) we get the desired bound.
\end{proof}

\end{proof}

\section{Extensions}
\label{sec:extensions}

An earlier version of this paper section sketched how to use the techniques developed in this paper to show sample complexity upper bounds for several natural extensions of the settings discussed in the Section~\ref{sec:no_signal} and Section~\ref{sec:signal}.
We omit these extensions in the current version since these bounds are subsumed by subsequent works.
Instead, we briefly summarize the technical ingredients developed in this paper and the scopes in which they hold.
We also introduce a new ingredient called Mapping Viewpoint that is necessary for generalizing our techniques to derive sample complexity upper bounds for problems in the downward-closed setting in the \crmodel. 

\begin{itemize}
    \item \textbf{Tail-bound for Regular Distributions:}
        Lemma~\ref{lem:tailbound_regular} proves that for any $\epsilon$, values larger than a certain threshold can be rounded down to the threshold while preserving at least a $1 - \epsilon$ fraction of the optimal revenue.
        This result holds for \emph{any problem in the downward-closed setting}.
        It implies a reduction from the case of regular distributions to that of bounded distributions.
        Generalizing the extreme value theorem for MHR distributions by \citet{CaiD15extreme} to derive better tail-bounds for MHR distributions is beyond the scope of this paper; 
        we remark that the sample complexity upper bounds for regular distributions apply to MHR distributions since the latter is a special case of the former.
    \item \textbf{Discretization in Value Space:}
        Lemma~\ref{lem:CR_bounded_additive} (resp., Lemma~\ref{lem:CR_bounded_multiplicative}) shows that it suffices to consider values that are multiples of $\epsilon$ (resp., powers of $1 + \epsilon$) for the purpose of getting an $\epsilon$ additive (resp., multiplicative $1 - \epsilon$) approximation, for \emph{arbitrary single parameter problems}, even those beyond the downward-closed setting.
        As a corollary of this result (and tail-bounds for regular and MHR distributions in the corresponding settings), it suffices to consider a value set of size $\tilde{O}(\epsilon^{-1})$.
        Section~\ref{sec:no_signal} uses the extreme value theorem by \citet{CaiD15extreme}
    \item \textbf{Ranking Viewpoint:}
        Theorem~\ref{thm:CR_bounded_additive} shows that when the value set is finite, Myerson's optimal auction can be characterized by a total order over $n |V|$ bidder-value pairs (according to their virtual values) and a dummy element (corresponding to having virtual value $0$), for any problem in the single-item setting.
        We shall refer to this observation as the ranking viewpoint, and show in Lemma~\ref{lem:ranking_viewpoint} that it holds more generally for \emph{any problem in the matroid setting}. 
        As a result, it suffices to consider auctions that correspond to these $(n|V|+1)!$ possible rankings.
        Readers who are familiar with standard arguments in statistical learning theory may realize the ranking viewpoint together with the discretization in value space immediately lead to an \emph{information theoretic} sample complexity upper bound.
    \item \textbf{Mapping Viewpoint (Discretization in Virtual Value Space):}
        In order to generalize our results to problems in the downward-closed setting, we will show in Lemma~\ref{lem:mapping_viewpoint} that to get an $\epsilon$ additive (resp., $1 - \epsilon$ multiplicative) approximation, it suffices to approximately maps values to virtual values up to a $\epsilon$ additive (resp., $1 - \epsilon$ multiplicative) approximation.
        As a result, there always exists an $\epsilon$ additive (resp., $1 - \epsilon$ multiplicative) approximately optimal auction that can be characterized by $n$ mappings from the value set to $\tilde{O}(\epsilon^{-1})$ discretized virtual values; 
        there are $\tilde{O}(\epsilon^{-1})^{n|V|}$ such mappings.
        We shall refer to this result as the mapping viewpoint, which holds for \emph{any single parameter problems in the downward-closed setting}.
    \item \textbf{Concentration Inequality for Product Distributions:}
        Theorem~\ref{thm:prodconcentrate} show a Bernstein-style concentration inequality for product distributions, which allows us to derive constructive and computationally efficient sample complexity upper bounds.
        The theorem holds for \emph{any single parameter problems}.%
        \footnote{In fact, it goes beyond auctions and holds generally for any bounded functions and any product distributions.}
    \item \textbf{Strong Revenue Monotonicity:}
        Theorem~\ref{thm:revenueMonotonicity} shows that Myerson's optimal auction w.r.t.\ a distribution $\Dboldhat$ gets at least the optimal revenue w.r.t.\ $\Dboldhat$ on any distribution that first-order stochastically dominates $\Dboldhat$. This result holds for \emph{any single parameter problem in the matroid setting}.
\end{itemize}

Table~\ref{tab:lemma_scope} summarizes the scopes in which different technical ingredients hold, and Table~\ref{tab:lemma_dependence} shows how the sample complexity upper bounds in different settings depend on these ingredients.

\begin{table}
    \renewcommand\arraystretch{1.2}
    \centering
    \begin{tabular}{|c|c|c|c|}
        \hline
        & Single-item & Matroid & Downward-closed \\
        \hline 
        Regular Tail Bound & \cmark & \cmark & \cmark \\
        \hline
        Value Discretization & \cmark & \cmark & \cmark \\
        \hline
        Ranking Viewpoint & \cmark & \cmark & \\
        \hline
        Mapping Viewpoint & \cmark & \cmark & \cmark \\
        \hline
        Concentration for Product Distributions & \cmark & \cmark & \cmark \\
        \hline
        Strong Revenue Monotonicity & \cmark & \cmark & \\
        \hline
    \end{tabular}
    \caption{Summary of the scopes in which different lemmas hold}
    \label{tab:lemma_scope}
    %
    \bigskip
    %
    \renewcommand\arraystretch{1.2}
    \centering
    \begin{tabular}{|c|c|c|c|}
        \hline
        & \multicolumn{2}{c|}{\CRmodel} & \multirow{2}{*}{\Signalsmodel} \\
        \cline{2-3}
        & Regular (MHR) & Bounded & \\
        \hline 
        Regular Tail Bound & \cmark & & \cmark \\
        \hline
        Value Discretization & \cmark & \cmark & \cmark \\
        \hline
        Ranking/Mapping Viewpoint & \cmark & \cmark & \\
        \hline
        Concentration for Product Distributions & \cmark & \cmark & \cmark \\
        \hline
        Strong Revenue Monotonicity & & & \cmark \\
        \hline
    \end{tabular}
    \caption{Summary of how the sample complexity upper bounds in different models depend on different ingredients}
    \label{tab:lemma_dependence}
\end{table}

\paragraph{Ranking Viewpoint.}

We first formalize the ranking viewpoint and argue that it genearlizes to the matroid setting.

\begin{lemma}[Ranking Viewpoint]
    \label{lem:ranking_viewpoint}
    For any single parameter problem in the matroid setting with a finite set of values $V$, there is a total order $\sigma$ over $n |V|$ bidder-value pairs and a dummy element $\perp$, where $\sigma(i, v)$'s and $\sigma(\perp)$ are the ranks of the elements, such that the following auction is incentive compatible and revenue optimal:
    \begin{enumerate}
        \item Given any value profile $\vec{v} = (v_1, v_2, \dots, v_n)$, pick bidders into the winner set one by one:
            \begin{enumerate}
                \item Initialize $\vec{\alloc} = \emptyset$.
                \item While there exists bidder $i$ such that (i) it is feasible to add bidder $i$, i.e., $\vec{\alloc} \cup \{ i \} \in \allocset$, and (ii) its rank is higher than the dummy element, i.e., $\sigma(i, v_i) < \sigma(\perp)$, add the bidder with the highest rank, i.e., the smallest $\sigma(i, v_i)$, to the winner set $\vec{\alloc}$.
            \end{enumerate}
        \item For each bidder $i$ in the winner set $\vec{\alloc}$, charge a price that equals the minimum value that would have still let bidder $i$ win in the previous step.
    \end{enumerate}
\end{lemma}

\begin{proof}
    Let $\sigma$ be the total order derived from the ironed virtual values, where the dummy element $\perp$ corresponds to having ironed virtual value $0$; break ties in lexicographical order of the bidders, and the order of the values.
    More precisely, let it be such that $\sigma(i, v_i) < \sigma(\perp)$ if and only if the ironed virtual value of $v_i$ w.r.t.\ $D_i$ is positive, and $\sigma(i, v_i) < \sigma(j, v_j)$ if and only if one of the following conditions hold:
    \begin{enumerate}
        \item The ironed virtual value $\bar{\phi}_i(v_i)$ of value $v_i$ w.r.t.\ $D_i$ is strictly larger than that of $v_j$ w.r.t.\ $D_j$.
        \item The ironed virtual values are equal, but $i < j$.
        \item The ironed virtual value are equal and $i = j$, but $v_i > v_j$.
    \end{enumerate}

    By the definition of Myerson's optimal auction, its allocation rule pick a winner set $\vec{\alloc} \in \allocset$ such that the sum of the corresponding ironed virtual values, i.e., $\sum_{i \in \vec{\alloc}} \bar{\phi}_i(v_i)$, is maximized.
    Further, it is known that in the matroid setting, the following greedy algorithm solves the above maximization problem exactly (see, e.g., Theorem 40.1 in \citet{Schrijver/2003/Springer}):
    \begin{enumerate}[label={(\alph*)}]
        \item Initialize $\vec{\alloc} = \emptyset$.
        \item While there exists bidder $i$ such that (i) it is feasible to add bidder $i$, i.e., $\vec{\alloc} \cup \{ i \} \in \allocset$, and (ii) its ironed virtual value is positive, i.e., $\bar{\phi}_i(v_i) > 0$, add the bidder with the largest ironed virtual value to the winner set $\vec{\alloc}$.
    \end{enumerate}

    By the choice of the total order $\sigma$, the above allocation rule is identical to that in the lemma.
    Further by Myerson's characterization, the unique payment rule that makes the auction incentive compatible is the threshold pricing stated in the lemma.
\end{proof}

\paragraph{Mapping Viewpoint.}

\begin{lemma}[Mapping Viewpoint]
    \label{lem:mapping_viewpoint}
    Consider any single parameter problem in the downward-closed setting with a finite set of values $V$.
    Let $v_{\max}$ be the largest number in $V$.
    Then, there is a mapping $\sigma$ from $n |V|$ bidder-value pairs to non-negative multiples of $\frac{\epsilon}{n}$ that are at most $v_{\max}$, such that the following auction is incentive compatible and an $\epsilon$ additive approximation:
    \begin{enumerate}
        \item Given any value profile $\vec{v} = (v_1, v_2, \dots, v_n)$, pick a feasible winner set $\vec{\alloc} \in \allocset$ that maximizes the sum of values of the $\sigma$ mapping, i.e., $\sum_{i \in \vec{\alloc}} \sigma(i, v_i)$ (wlog, no $i \in \vec{\alloc}$ has $\sigma(i, v_i) = 0$).
        \item For each bidder $i$ in the winner set $\vec{\alloc}$, charge a price that equals the minimum value that would have still let bidder $i$ win in the previous step.
    \end{enumerate}
\end{lemma}

\begin{proof}
    Consider the following mapping $\sigma(i, v_i)$.
    If the ironed virtual value $\bar{\phi}_i(v_i) \le 0$, let $\sigma(i, v_i) = 0$;
    otherwise, define $\sigma(i, v_i)$ by rounding the ironed virtual value $\bar{\phi}_i(v_i)$ down to the closest multiple of $\frac{\epsilon}{n}$ from below.
    Here, $\sigma(i, v_i) \le v_{\max}$ follows by the fact that ironed virtual values are no greater than the corresponding values.
    Then, consider the allocation $\vec{\alloc}$ that maximizes $\sum_{i \in \vec{\alloc}} \sigma(i, v_i)$, which wlog does not take any $i$ whose $\sigma(i, v_i) = 0$ (and thus $\bar{\phi}_i(v_i) \le 0$).

    First, this allocation rule is monotone in the sense that for any bidder $i$, having a larger value $v_i$ weakly increases the allocation $x_i$ to the bidder. 
    To see this, fix any bidder $i$, $\sigma(i, v_i)$ as a function of $v_i$ is obtained by rounding the ironed virtual value $\bar{\phi}_i(v_i)$.
    Since $\bar{\phi}_i(v_i)$ is non-decreasing in $v_i$, so is $\sigma(i, v_i)$.
    Therefore, the allocation in the lemma, which maximizes $\sum_{i \in \vec{\alloc}} \sigma(i, v_i)$, is monotone.

    As a result, Myerson's characterization asserts that the threshold payment in the lemma is the unique payment rule that makes the auction DSIC and IR. 
    Further, the expected revenue is equal to the expectation of the virtual welfare, i.e.:
    \[
        \E \bigg[ \sum_{i \in \vec{\alloc}} \phi_i(v_i) \bigg]
        ~.
    \]

    Further, all values within the same ironed interval have the same ironed virutal values and thus, the same $\sigma(i, v_i)$, which implies they have the same allocation.
    We get that the expected virtual welfare is identical to the expected ironed virtual welfare, i.e.:
    \[
        \E \bigg[ \sum_{i \in \vec{\alloc}} \bar{\phi}_i(v_i) \bigg]
        ~.
    \]

    By that $\|\vec{\alloc}\|_\infty \le 1$, $\vec{\alloc}$ also maximizes $\sum_{i \in \vec{\alloc}} \bar{\phi}_i(v_i)$ up to an $\epsilon$ additive approximation.
    Putting together shows the stated approximation guarantees.
\end{proof}

\subsection*{Acknowledgements}
This work was done in part while the authors were visiting the Simons Institute for the Theory of Computing.
We thank Matt Weinberg for the useful discussions about Revenue Monotonicity.
We also thank Zhihao Gavin Tang for spotting a bug in our original proof for strong revenue monotonicity.

\bibliographystyle{plainnat}
\bibliography{references}

\end{document}